\documentclass[a4paper,onecolumn,superscriptaddress,11pt]{aux/quantumarticle}
\pdfoutput=1
\usepackage{graphicx}%
\usepackage{dcolumn}%
\usepackage{bm}%
\usepackage{amsthm,amsmath,amsfonts,amssymb}
\usepackage{cancel}
\usepackage{subcaption} 
\usepackage{enumitem}
\usepackage[usenames,dvipsnames]{xcolor}
 \usepackage{cite}
\usepackage{mathtools}

\input{aux/Qcircuit.tex}

\renewcommand*{\eqref}[1]{Eq.~\ref{#1}}
\newcommand*{\eqrefp}[1]{(Eq.~\ref{#1})}
\newcommand{\eqsref}[2]{Eqs.~\ref{#1}~and~\ref{#2}}
\newcommand{\eqssref}[3]{Eqs.~\ref{#1}, \ref{#2}, and~\ref{#3}}
\newcommand{\eqsssref}[4]{Eqs.~\ref{#1}, \ref{#2},~\ref{#3}, and~\ref{#4}}
\newcommand*{\thmref}[1]{Thm.~\ref{#1}}
\newcommand{\thmsref}[2]{Thms.~\ref{#1}~and~\ref{#2}}
\newcommand{\thmssref}[3]{Thms.~\ref{#1},~\ref{#2}, and~\ref{#3}}
\newcommand*{\lemref}[1]{Lem.~\ref{#1}}
\newcommand{\lemsref}[2]{Lem.~\ref{#1}~and~\ref{#2}}
\newcommand{\lemssref}[3]{Lems.~\ref{#1},~\ref{#2}, and~\ref{#3}}
\newcommand*{\corref}[1]{Cor.~\ref{#1}}
\newcommand*{\propref}[1]{Prop.~\ref{#1}}
\newcommand*{\remref}[1]{Rem.~\ref{#1}}

\newcommand*{\figref}[1]{Fig.~\ref{#1}}
\newcommand*{\tabref}[1]{Tab.~\ref{#1}}
\newcommand*{\secref}[1]{Sec.~\ref{#1}}
\newcommand{\secsref}[2]{Secs.~\ref{#1}~and~\ref{#2}}
\newcommand{\secssref}[3]{Secs.~\ref{#1},~\ref{#2},~and~\ref{#3}}
\newcommand*{\appref}[1]{App.~\ref{#1}}

\newcommand{\braket}[2]{\left<#1|#2\right>}

\newcommand{\naturals} {{\mathbb N}}
 
\newcommand{\complex}{{\mathbb C}}
\newcommand{\reals}{{\mathbb R}}
\newcommand{\integers}{{\mathbb Z}}

\renewcommand{\figurename}{Fig.}

\newtheorem{theorem}{Theorem}
\numberwithin{theorem}{subsection}
\newtheorem{rem}{Remark}
\numberwithin{rem}{subsection}
\newtheorem{lem}{Lemma}
\numberwithin{lem}{subsection}
\newtheorem{prop}{Proposition}
\numberwithin{prop}{subsection}
\newtheorem{cor}{Corollary}
\numberwithin{cor}{subsection}

\numberwithin{example}{subsection}

\newcommand\e{\varepsilon}

\usepackage{comment}
\newcommand{\sh}[1]{{\color{black} #1}}

\begin{document}
\title{Analytical Framework for Quantum Alternating Operator Ans\"atze}

\def\QUAIL{
\affiliation{Quantum Artificial Intelligence Lab. (QuAIL), NASA Ames Research Center, Moffett Field, CA 94035, USA}} 
\def\USRA{
\affiliation{USRA Research Institute for Advanced Computer Science (RIACS), Mountain View, CA 94043, USA}}

\author{Stuart Hadfield}
\QUAIL \USRA
\author{Tad Hogg} \QUAIL%
\author{Eleanor G. Rieffel} \QUAIL

\date{December 2022}%

\begin{abstract} 
We develop a framework for analyzing layered quantum algorithms such as quantum alternating operator ans\"atze.   
In the context of combinatorial optimization, our framework relates quantum cost gradient operators, derived from the cost and mixing Hamiltonians, to classical cost difference functions that reflect cost function neighborhood structure. By considering QAOA circuits from the Heisenberg picture, we derive exact general expressions for %
expectation values as series expansions in the algorithm parameters, cost gradient operators, and cost difference functions. This enables novel interpretability %
and 
insight into QAOA behavior in various parameter regimes. 
For single-level QAOA$_1$ we %
show the leading-order changes in the 
output probabilities and cost expectation value explicitly in terms of classical cost differences, %
for arbitrary cost functions. 
This demonstrates that, for sufficiently small positive parameters, probability flows from lower to higher cost states on average.
By selecting signs of the parameters, we can control the direction of flow.
We use these results to derive a classical random algorithm emulating QAOA$_1$ in the small-parameter regime, i.e., that produces bitstring samples with the same probabilities as QAOA$_1$ up to small error. %
For deeper QAOA$_p$ circuits %
we apply our framework to derive analogous and additional results in several settings. In particular we show QAOA always beats random guessing.  
We describe how our framework incorporates cost Hamiltonian locality for specific problem classes, including 
causal cone approaches, and %
applies to QAOA performance analysis with arbitrary parameters.  
We %
illuminate our results with a number of examples %
including applications 
to QUBO problems, MaxCut, and variants of MaxSAT. %
We illustrate the generalization of our framework to QAOA circuits using mixing %
unitaries beyond the transverse-field mixer  through two examples of constrained optimization problems, Max Independent Set and Graph Coloring. 
We conclude by outlining some of the further applications we envision for the framework.
\end{abstract}

\maketitle
\tableofcontents

\section{Introduction} \label{sec:intro}
Parameterized quantum circuits have gained much attention %
both as potential near-term algorithms and as new paradigms for quantum algorithm design more generally. 
Approaches based on %
Quantum Alternating Operator Ans\"atze (QAOA)~%
\cite{hogg2000quantum,hogg2000quantumb, Farhi2014,Hadfield17_QApprox,hadfield2019quantum} %
have been extensively studied in recent years
\cite{farhi2014quantum,farhi2016quantum,Wecker2016training,Shabani16,lin2016performance,jiang2017near,farhi2017quantum,verdon2017quantum,wang2018quantum,crooks2018performance,mcclean2018barren,hadfield2018thesis,pichler2018quantum,zhou2018quantum,lloyd2018quantum,brandao2018fixed,niu2019optimizing,bapat2018bang,guerreschi2019qaoa,marsh2019quantum,verdon2019cvqaoa,akshay2019reachability,morales2019universality,farhi2019quantum,wang2019xy,hastings2019classical,szegedy2019qaoa,bravyi2019obstacles,marshall2020characterizing,farhi2020quantum,farhi2020quantumb,shaydulin2020classical,ozaeta2020expectation,wurtz2020bounds,stollenwerk2020toward,streif2020quantum,streif2020training,marwaha2021local,barak2021classical,marwaha2021bounds,chou2021limitations,harrigan2021quantum,kremenetski2021quantum,brady2021optimal,brady2021behavior,wurtz2022counterdiabaticity}. 
These approaches  alternate $p$ times between applications of a 
cost-function-based phase-separation operator and a probability-amplitude mixing operator, illustrated in \figref{fig:QAOAcircuit}.
Nevertheless, relatively few rigorous performance guarantees are known in nontrivial settings, especially beyond a small constant number of circuit layers and the few specific problems studied thus far. 
Thus, the power and underlying mechanisms of such algorithms remains unclear, 
as do the problems and regimes where quantum advantage may be possible. 
Hence it is important to develop novel tools and approaches for analyzing and understanding such algorithms.

\begin{figure}[h]
\centerline{
\Qcircuit @C=0.25em @R=.1em {
\lstick{} &  \multigate{2}{U_P(\gamma_1)} & \qw & \multigate{2}{U_M(\beta_1)} & \qw & \multigate{2}{U_P(\gamma_2)} & \qw & & & & & \hdots &  & & & & \qw &  \multigate{2}{U_P(\gamma_p)} & \qw & \multigate{2}{U_M(\beta_p)} & \qw & \meter \\
\lstick{\ket{s}}%
& \ghost{U_P(\gamma_1)} & \qw 
& \ghost{U_M(\beta_1)} & \qw & \ghost{U_P(\gamma_2)} & \qw & & & & & %
 \vdotswithin{\hdots} & & & & & %
 & \ghost{U_P(\gamma_p)} & \qw & \ghost{U_M(\beta_p)}  & \qw & :\\%\vdots  \\
\lstick{} & \ghost{{U_P(\gamma_1)}} & \qw & \ghost{U_M(\beta_1)}& \qw & \ghost{U_P(\gamma_2)}& \qw & & & & & \hdots & & & & & \qw & \ghost{U_P(\gamma_p)} & \qw &\ghost{U_M(\beta_p)} & \qw & \meter
}}
\caption{A Quantum Alternating Operator Ansatz circuit with~$p$ levels that alternate between application of the phase and mixing operators $U_P$ and $U_M$. The $\gamma_j,\beta_j$ are %
parameters for these operators and $\ket{s}$ is a suitable initial state. Our framework particularly applies to such layered quantum circuit ans\"atze.
} 
\label{fig:QAOAcircuit}
\end{figure}
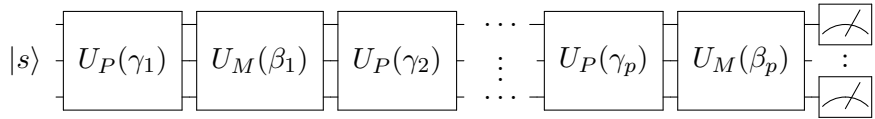

This analysis includes evaluating QAOA behavior and performance for both individual problem instances and for classes of problems or instances, e.g., worst- or average-case performance for instances in a class. Evaluating average behavior can be particularly useful for large $p$ since that involves many QAOA parameters. In this case, finding optimal parameters for each instance can be difficult. Instead, focusing on characteristic or average behavior can simplify parameter selection by identifying a single set of parameters that work well for typical instances of the class, as well as lead to new performance guarantees. 

To help address these issues, in this paper we develop theoretical tools in the form of an analytic framework, with notation, concepts, and results, to provide better understanding of the behavior of these algorithms, their strengths and weaknesses, and ultimately to design better quantum algorithms. 
The main underlying ideas of our framework are as follows, and %
apply to computing algorithm operators and expectation values of interest:
\begin{itemize}
    \item Quantum circuit observable expectation values can be equivalently computed by acting on the observable by the circuit (by conjugation $H\rightarrow U^\dagger H U=H+\dots $, i.e., the Heisenberg picture), followed by taking the initial state expectation value.
    \item For layered circuits we can recursively conjugate to obtain exact operator series in the algorithm parameters; for alternating ansatz terms of the same or similar form will reappear in the series.
    \item For QAOA, in particular with the originally proposed transverse-field mixer, the resulting terms, their action, and %
    their initial state expectation values can be related to classical functions derived from the cost function that capture its structure.
\end{itemize}
A number of useful properties for analysis are shown, including that initial state expectation values of many of the resulting terms are identically zero. For specific problem classes we can directly incorporate problem and operator locality. Further, 
in particular parameter regimes many such terms may be close to zero, allowing for 
compact approximate expressions by neglecting such terms, particularly for higher-order behavior. The expressions we derive are useful for more efficient numeric explorations of QAOA, as well as direct analytical insights. Further, the framework can be used to simplify the generation of more complicated expressions with computer algebra systems, such as for exact expressions or higher-order approximations, enabling further %
exploration of QAOA behavior.

Our framework can be applied to illuminate various aspects of QAOA. We illustrate the use of the framework in several general results and applications including:
\begin{itemize}
\item exact formulas for QAOA$_p$ probabilities and expectation values as power series in the algorithm angles (parameters). 
\item %
leading-order %
approximations for %
QAOA$_p$ expectation values and probabilities.  %
\begin{itemize}
    \item for arbitrary (nonconstant) cost functions there exist angles such that QAOA$_p$ beats random guessing.
    \item better or more %
    applicable approximations may be systematically obtained by including higher-order terms.
    \item %
    simple classical algorithms %
emulate sampling from QAOA circuits in some of these regimes, with small error.
\item identification of some %
general parameter regimes for which QAOA$_p$ performance analysis is %
classically tractable, particularly a regime of polynomially-small parameters, which %
hence precludes quantum advantage, as illustrated in \figref{fig:phaseDiagram}. %
\end{itemize}
\item a generalized formalization unifying %
several previous approaches for deriving QAOA$_p$ %
performance bounds, %
encapsulating %
previous results for specific problems such as exact results for MaxCut; we demonstrate this approach by obtaining results 
for QAOA$_1$ 
on a variant of Max-$2$-SAT as a specific example. 
\item a number of analytical and numerical examples illustrating the application of our framework to QAOA with the %
transverse-field mixer. 
\item two examples showing how our framework and results %
extend to quantum alternating ans\"atze beyond the %
transverse-field mixer, %
involving problems with different domains and encodings, 
or hard feasibility constraints.
\end{itemize}

Though we focus on the original QAOA, the framework naturally generalizes to other cases with structured ans\"atze and so would be useful for other applications, such as realizations of the variational quantum eigensolver for quantum chemistry, though in this paper we only hint at these further applications 
\sh{through the examples of \secref{sec:generalizedCalculus}
and}
concluding 
discussion of \secref{sec:discussion}. 

\begin{figure}[ht]
\centering
  \includegraphics[width=8cm]{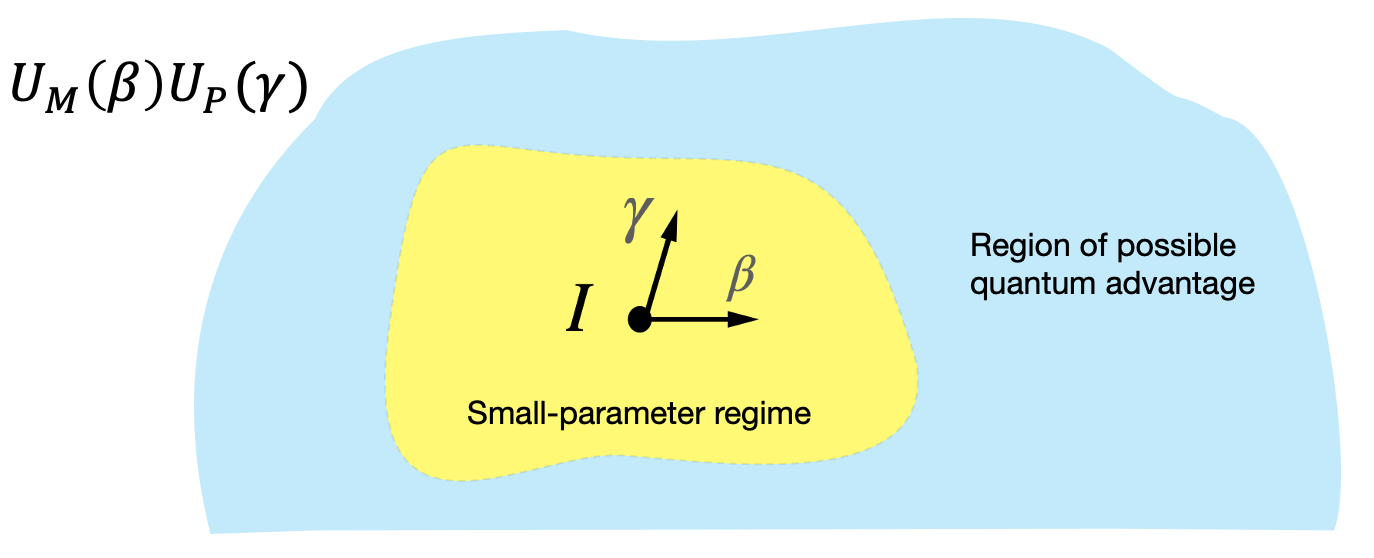}
  \caption{Schematic regimes of the QAOA$_1$ operator $U_M(\beta)U_P(\gamma)$ with parameters $\beta$ and $\gamma$; similar ideas apply for QAOA$_p$.  The inner region indicates operators close to 
  the Identity $I=U_M(0)U_P(0)$, corresponding to small parameters, where the QAOA probabilities and cost expectation values are characterized by the leading order terms of the exact %
  series expressions we derive.
 In particular, sufficiently (i.e., inverse-polynomially) small parameters allow for efficient classical emulation generally (i.e., classical sampling up to small error). In contrast, for arbitrary parameters, efficient classical sampling %
  from  QAOA$_1$ circuits is impossible under standard beliefs in computational complexity theory~\cite{farhi2016quantum}. }
  \label{fig:phaseDiagram}
\end{figure}

We remark here on some particularly related prior work.
Our results build off, generalize, and unify various previous analytical approaches and results for the performance of QAOA for specific problems or in particular settings. For example, \cite{wang2018quantum,hadfield2018thesis,lin2016performance,brandao2018fixed,bravyi2019obstacles,wurtz2020bounds,ozaeta2020expectation,marwaha2021local,marwaha2021bounds,chou2021limitations} obtain bounds to the cost expectation value for relatively few layers. 
Whereas most prior work concerns specific %
problems or problem classes, many of our results apply to arbitrary cost functions. 
For MaxCut, ~\cite[App. 22-24]{szegedy2019qaoa} proposes a %
related approach for algorithmically generating polynomials capturing the QAOA cost expectation value. Another recent paper, similar in spirit to our framework though employing different technical approaches, explores a number of related but distinct viewpoints on low-depth quantum optimization~\cite{mcclean2020low}, in particular relating the QAOA mixing operator to a graph Laplacian (of which the transverse-field mixer is a special case). Our work provides complementary insights on the connection between 
mixers and cost function structure. Several of recent works consider classical algorithms matching the performance of low-depth QAOA circuits~\cite{hastings2019classical,marwaha2021local,barak2021classical}. Classical algorithms for generally computing expectation values of quantum observables are considered in~\cite{bravyi2019classical}; in particular they show quantum circuits exponentially close to the identity can be efficiently sampled from. 
Our work extends this result, in the context of QAOA circuits, from exponentially small angles to polynomially small. 

\sh{We emphasize that many of our formulas and results apply to QAOA circuits of arbitrary depth. 
As an example we consider QAOA for Grover's unstructured search problem, and demonstrate in \secref{sec:smallAngleQAOAp} how small-angle analysis easily reproduces the result obtained rigorously in~\cite{jiang2017near} that QAOA with the transverse-field mixer can reproduce the famous quadratic speedup of Grover's algorithm.}

The remainder of the Introduction gives an informal description of our framework and the results we obtained by applying it. Subsequent sections give full definitions, generalizations and precise statements of results.

\subsection{Overview of analytical framework} \label{sec:frameworkOverview}
Quantum algorithms based on  quantum alternating operator ans\"atze 
utilize the noncommutativity of the cost and mixing operators in a fundamental way. 
From the Heisenberg perspective, %
a quantum circuit may be equivalently viewed as acting on quantum observables\footnote{Many important quantities for parameterized quantum circuits may be expressed as observable expectation values such as the expected cost and approximation ratio, or measurement probabilities.} 
by conjugation, rather than on a quantum state (by matrix multiplication); observable expectation values may be equivalently computed as the initial state expectation of the conjugated observable. For a circuit with $\ell$ layers, this corresponds to $\ell$ iterated conjugation operations. 
In some cases this formulation gives advantages in computing or approximating expectation values. For example, this %
idea was applied to derive an exact formula for the level-$1$ QAOA cost expectation value for MaxCut in \cite{wang2018quantum}. 
We extend these approaches to more general cost functions and settings using the %
correspondence between %
commutators %
and unitary conjugation (i.e., between Lie algebras and Lie groups) 
to derive exact expressions for relevant operators and expectation values as power series in the algorithm parameters with terms that reflect cost function structure. 

For QAOA we show how the resulting operators and expectation values 
reflect cost function changes over neighborhoods induced by the mixing operator. 
For simplicity, the primary example we consider in the paper is the transverse-field mixer Hamiltonian $B=\sum_j X_j$ %
and corresponding initial state $\ket{s}:=\ket{+}^{\otimes n}$
as originally proposed for the quantum approximate optimization algorithm~\cite{farhi2014quantum}, though our framework may be applied more generally.  
This mixer induces a neighborhood structure on classical bitstrings related to Hamming (i.e., bit-flip) distance on the Boolean cube. 
Different mixers induce different neighborhood structures, and
our framework %
extends to these cases as well, 
as we illustrate with two examples in \secref{sec:generalizedCalculus}. 

At the heart of our approach is a fundamental correspondence 
 $$\textrm{quantum operators}\,\,\,\, \longleftrightarrow \,\,\,\,  \textrm{classical functions}$$
between derived operators on the space of quantum states, and classical functions that describe their behavior in terms of the structure of the cost function. The functions and operators appearing in our framework are related to discrete versions of functions and operators from vector calculus, particularly discrete difference operators. Given a mixer $B$ and cost Hamiltonian $C$, %
iteratively taking commutators generates a sequence of Hamiltonians (up to factors of $i$). 
The most fundamental such operator is their commutator
$$ \nabla C := [B,C]=BC-CB,$$ 
which we call the \textit{cost gradient operator}, as motivated by its action on particular %
quantum states. %
(For convenience, we will sometimes refer to higher-order commutators of $B$ and $C$ as \textit{cost gradients} generically.) 
In particular, we show $\nabla C \ket{s}=\tfrac1{\sqrt{2^n}}\sum_{x\in\{0,1\}^n}dc(x)\ket{x}$,
for the classical \textit{cost divergence function}
$$ dc(x):=\sum_{j=1}^n \partial_j c(x),$$
where $c(x)$ is the cost function to be optimized, and $\partial_j c(x)$ gives the change in cost for each string $x\in\{0,1\}^n$ with respect to flipping its $j$th bit (i.e., $dc$ captures average cost function structure over single bit-flip neighborhoods).  
For example, $dc(x^*)\leq 0$ is a necessary condition for $x^*$ to maximize the cost function. We may then use properties of the functions to understand properties of the operators, and vice versa; for example we have in general $\tfrac1{2^n}\sum_x dc(x)=\bra{s}\nabla C\ket{s}=0$.

The cost gradient $\nabla C$ corresponds to infinitesimal conjugation of $C$ by the QAOA mixing operator, i.e., %
$e^{i\beta B}Ce^{-i \beta B}=C+i\beta\nabla C+O(\beta^2)$ as $|\beta|\rightarrow 0$.
Including additional terms leads to higher-order cost gradients (nested commutators, e.g., $\nabla^2 C:=[B,[B,C]]$, and so on) that reflect cost changes over neighborhoods of greater Hamming distance, at higher powers of the mixing angle.  
To include the QAOA phase operator we also require commutators with respect to $C$, $\nabla_C:=[C,\cdot]$, which act in relation to the underlying cost function $c(x)$. Hence, treating the layers of QAOA as iterated conjugations of an observable leads to series expressions which may be used to compute or approximate expectation values and other important  quantities.

As we discuss briefly in Sec.~\ref{sec:discussion}, our framework can be applied 
to recently proposed variants of QAOA \cite{bravyi2019obstacles,zhu2020adaptive}, and more 
generally to layered quantum algorithms
including %
applications beyond combinatorial  optimization; in such cases the resulting operators will reflect structure resulting from both the problem and choice of ansatz.

\subsection{Overview of application to QAOA}
The series expansions resulting from our framework are especially informative when the leading terms in the series dominate the behavior. This is the case for the \textit{small-angle regime}, i.e., when all angles in the QAOA parameter schedule are relatively small in magnitude. 
Hence, as a main application of our framework, 
we investigate the small-angle setting for QAOA circuits generally. We analyze single-layer QAOA$_1$ in detail in \secref{sec:QAOA1}, 
and generalize to $p\geq 1$ layers (denoted QAOA$_p$ throughout) in \secref{sec:QAOAp}.  
\begin{table}[h]
\centering
\begin{tabular}{ |c|c|c| } 
 \hline
 Quantity & Initial value & Leading-order contribution%
 \\
 \hline
 $P_1(x)$ & $\tfrac1{2^n}$ & $-\tfrac2{2^n}\gamma \beta dc(x)$ \\
  $\langle C \rangle_1$ & $ \tfrac1{2^n}\sum_x c(x)$ & $-\tfrac2{2^n}\gamma \beta \sum_x c(x) dc(x)$   \\
  \hline
  $P_p(x)$ & $\tfrac1{2^n}$ &  $-\tfrac2{2^n} \left(\sum_{1\leq i\leq j}^p \gamma_i \beta_j\right) dc(x)$  \\
  $\langle C \rangle_p$ & $ \tfrac1{2^n}\sum_x c(x)$ &  $-\tfrac2{2^n} \left(\sum_{1\leq i\leq j}^p \gamma_i \beta_j\right) \sum_x c(x)dc(x)$ \\ 
 \hline
\end{tabular}
\caption{%
Leading-order cost expectation and probabilities for QAOA$_p$ for %
cost function~$c(x)$, which dominate in particular when the QAOA angles are small, $|\gamma_j|,|\beta_j|\ll 1$.  
The initial probability of measuring $x\in\{0,1\}^n$ is $P_0(x)=\tfrac1{2^n}$ which corresponds to QAOA with all angles zero. The first two rows follow from \thmref{thm1:smallAngles}, and the following two rows follow from \thmsref{thm:allanglessmall}{thm:smallprecursed}. 
In each case the next contributing terms are order $4$ and higher in the QAOA angles; all other terms up to order $3$ are shown to be identically~$0$. We show the three additional  terms that contribute up to fifth order explicitly in \thmref{thm:smallprecursed}. We emphasize the expression 
$\sum_{1\leq i\leq j}^p \gamma_i \beta_j$ depends on both the algorithm parameter values and their ordering.}
\label{tab:tab1smallangles}
\end{table}

For QAOA$_1$, we show that, to third order in the angles $\gamma,\beta$, the probability to measure each bitstring $x$ changes from its initial value by a single contribution 
proportional to the cost divergence  $dc(x)$ and the product $\gamma\beta$. Hence, with respect to the $2^n$-dimensional space of bitstring probabilities, to leading order QAOA is reminiscent of a step of classical gradient descent, with learning rate (step size) proportional to $\gamma\beta$. The leading-order contribution to the cost expectation then follows as the expectation of 
$c(x)dc(x)$ taken uniformly over bitstrings $x$.   
Both leading-order expressions are shown in \tabref{tab:tab1smallangles}.
These expressions are derived using the correspondence between cost gradients and classical functions. In particular, we give several general lemmas showing that the terms corresponding to other low-order combinations of $\gamma,\beta$ (i.e., in this case $\gamma,\beta,\gamma^2, \beta^2,\gamma^2\beta,\gamma\beta^2)$ are identically zero, independent of the particular cost function, and that the same terms contribute for the case of QAOA$_p$ with $p\geq 1$ (with coefficients depending on all $2p$ angles). 
Higher-order contributing terms are shown to similarly relate to classical functions and are relatively straightforward to derive using our framework, and such that increasingly accurate approximate formulas may be systematically generated.

As a consequence, we show that the measurement outcomes of QAOA$_1$ %
are effectively classically emulatable in the 
regime of polynomially small angles. More precisely, we provide a simple classical randomized algorithm that produces bitstrings with probabilities matching the leading order behavior of QAOA$_1$, and bound the resulting error from the neglected terms for the case where $\gamma,\beta$ are bounded in magnitude by an inverse-polynomial in the problem size. (A similar argument applies to QAOA$_p$ (cf. \tabref{tab:tab1smallangles}), though we do not  analyze the error in detail for the general $p$ case.) We use these results to show that QAOA always beats random guessing for any nonconstant cost function. 

Our general results point to the tradeoff between parameter size and number of QAOA levels for potential quantum advantage.
Indeed, %
it is known that sampling from QAOA$_p$ circuits with unrestricted parameters, even for $p=1$, cannot be efficiently performed classically under widely believed complexity theoretic assumptions~\cite{farhi2016quantum}. Our results take steps towards a clearer demarcation of when quantum advantage may be possible (cf. \figref{fig:phaseDiagram}).  
Extending the analysis to QAOA with arbitrary number of levels~$p$ shows that QAOA$_p$ gives the same leading-order behavior as QAOA$_1,$ with a suitable choice of effective angles for QAOA$_1$, as indicated by the quantities shown in  \tabref{tab:tab1smallangles}. Thus, there is no possibility of quantum advantage for QAOA$_p$ for fixed $p$ when all angles are %
small. 

We similarly show how to obtain leading-order terms in particular cases where only some of the angles are small, which yield further generally applicable expressions. For QAOA$_1$ applied to quadratic unconstrained binary optimization (QUBO) problems, we show the case of small-mixing angle $\beta$ and unrestricted phase angle $\gamma$ can again be efficiently emulated classically, whereas this does not appear generally possible in the converse case of small phase angle but arbitrary mixing angle. The latter case can be understood in light of our framework from the fact that cost differences over neighborhoods of arbitrary Hamming distance contribute as the mixing angle grows in magnitude, which eventually become superpolynomial in size.  

Application of our framework is not restricted to small angles. Including higher order terms in our series expansions results in a sequence of increasingly accurate approximations valid over larger angle regimes. Alternatively, for specific problems and relatively small $p$ we can often derive compact exact expressions using our formalism by explicitly incorporating problem structure. We give high-level algorithms computing approximate or exact QAOA cost expectation values in this way in \secref{sec:classicalAlgGeneral2}. %
\sh{Our approach %
yields more tractable expressions than alternative approaches to obtaining general formulas; we contrast with one such %
\lq\lq sum-of-paths\rq\rq\ %
approach to QAOA in %
\appref{app:sumOfPaths}}.  %

\subsection{Roadmap of paper}
The Table of Contents gives the detailed structure of the paper. Here, we make some remarks on dependencies between the various sections.  

\secref{sec:calculus} gives a self-contained presentation of our framework, 
including general cost difference functions and cost gradient operators and their properties; see \tabref{tab:summary} below for a summary of important notation. 
As an immediate demonstration 
of our formalism we state our general leading-order results for QAOA in \secref{sec:smallAngleQAOA1initial}, 
and illustrate how the framework facilitates comparison between different quantum circuit ans\"atze by comparing
leading-order QAOA$_1$ to a %
quantum quench in \secref{sec:quench}. 
In \secref{sec:probLocality}, we show how to incorporate into our framework, when applicable, additional problem or instance-wise locality considerations (for example, MaxCut restricted to bounded-degree graphs), 
providing a basis for more refined QAOA results incorporating locality 
in \secsref{sec:lightcones}{sec:lightconesp}.

\secref{sec:QAOA1} provides a detailed application of our framework to QAOA$_1$. After deriving general results, we use the framework to obtain results for QAOA$_1$ for small mixing angle (but arbitrary phase) and vice versa (small phase, but arbitrary mixing angle). The section presents several examples; subsequent sections are independent of these examples. %

\secref{sec:QAOAp} applies the framework to the general case of QAOA$_p$. While a number of the
results in \secref{sec:QAOAp} generalize those in \secref{sec:QAOA1}, this section can be read immediately after \secref{sec:calculus} by those readers who wish to focus on the $p > 1$ case.

\secref{sec:generalizedCalculus} discusses how our framework applies to QAOA using different encodings and mixers other than the transverse-field mixer. In this case, the operators in our framework are defined in terms of the neighborhood structures induced by those alternative mixers, which we illustrate with 
two examples of constrained problems: Maximum Independent Set and Graph Coloring. This section builds %
on the %
results 
of \secref{sec:calculus}, but does not %
directly depend on those of %
\secref{sec:QAOA1} or \secref{sec:QAOAp}.

Finally, we comment on several additional applications of our framework in \secref{sec:discussion}. %
Some additional technical results and proofs are deferred to the Apendices and may also be read independently.

\section{Analytical framework for quantum optimization}
\label{sec:calculus}
Consider a classical cost function $c(x)$ we seek to optimize over bitstrings $x\in\{0,1\}^n$, i.e., the $n$-dimension Boolean hypercube. For example, $c(x)$ gives the number of cut edges in MaxCut, or the number of satisfied clauses in MaxSAT. Here for clarity we assume all bitstrings encode feasible candidate solutions; we consider examples which extend our methods to optimization over non-trivial feasible subspaces in \secref{sec:generalizedCalculus}. 

We may encode each bitstring $x\in\{0,1\}^n$ with the corresponding $n$-qubit \textit{computational basis} state $\ket{x}$. A computational basis (Pauli~$Z$) measurement of this quantum register at the end of an algorithm then returns each candidate solution~$x$ with some probability~$P(x)$. 
(Throughout the paper $n$ denotes the number of (qu)bits and $X_j,Y_j,Z_j$ denote the Pauli matrices 
$\sigma_X,\sigma_Y,\sigma_Z$ acting on the $j$th qubit, respectively.)

The cost function is naturally represented 
as the \textit{cost Hamiltonian} 
\begin{equation} \label{eq:costHam}
    C = \sum_x c(x) \ket{x}\bra{x}
\end{equation}
which acts diagonally in the computational basis as
$C\ket{x}=c(x)\ket{x}$ for each $x\in\{0,1\}^n$. In quantum algorithms such as quantum annealing or QAOA, the usual initial state is 
\begin{equation} \label{eq:initState}
    \ket{s}:=\ket{+}^{\otimes n}=\left(\tfrac{\ket{0}+\ket{1}}{\sqrt{2}} \right)^{\otimes n} =  \frac1{\sqrt{2^n}} \sum_{x\in\{0,1\}^n} \ket{x}
\end{equation}
which gives a uniform probability distribution over bitstrings with respect to computational basis measurements 
(i.e., %
$P_0(x):=|\bra{x}\ket{s}|^2=\tfrac1{2^n}$, and $\langle c\rangle_0 = \tfrac1{2^n}\sum_x c(x)$),
and the  
(transverse-field) 
\textit{mixing Hamiltonian}  
\begin{equation} \label{eq:mixingHam}
    B :=\sum_{j=1}^n X_j 
\end{equation}
is utilized to mediate probability flow between states.\footnote{Though different initial states or mixing Hamiltonians have been proposed to some extent 
in the literature, for simplicity we %
use $\ket{s}$ to denote the equal superposition state \eqrefp{eq:initState} and $B$ to denote the transverse-field Hamiltonian \eqrefp{eq:mixingHam} throughout the paper. This paradigm is called X-QAOA in \cite{streif2020quantum}.} 

Here we briefly review the original QAOA~\cite{Farhi2014} of Farhi, Goldstone, and Gutmann, %
which is the main application we consider in the paper. 
A QAOA$_p$ circuit creates the parameterized quantum state (cf.~\figref{fig:QAOAcircuit}) 
\begin{equation}
    \ket{\boldsymbol{\gamma \beta}}_p 
    = U_M(\beta_p)U_P(\gamma_p)U_M(\beta_{p-1})\dots U_P(\gamma_{1})\ket{s}
    = e^{-i\beta_p B}e^{-i\gamma_p C}e^{-i\beta_{p-1} B}\dots e^{-i\gamma_1 C} \ket{s}
\end{equation}
by applying the mixing and phase separation unitaries $U_M(\beta), U_P(\gamma)$ in alternation $p$ times each to the initial state $\ket{s}=\tfrac1{\sqrt{2^n}} \sum_x \ket{x}$.
The original %
QAOA mixing unitary $U_M(\beta)=\exp(-i\beta B)$ is specified as time evolution under the transverse-field Hamiltonian of \eqref{eq:mixingHam},
and the phase operator $U_P(\gamma)=\exp(-i\gamma C)$ by time evolution under the cost Hamiltonian $C$. %
In optimization applications typically each QAOA state preparation is followed by a measurement in the computational basis which returns some $x\in\{0,1\}^n$ probabilistically. 
For QAOA$_p$ with fixed parameters we let $P_p(x)$ denote the probability of a %
such a measurement returning the bitstring $x$, and  $\langle C \rangle_p = \langle c \rangle_p :=\bra{\boldsymbol{\gamma \beta}}C\ket{\boldsymbol{\gamma \beta}}$ %
the expected value of the cost Hamiltonian (function). 
More generally, %
we use
$ \langle A \rangle_p := \bra{\boldsymbol{\gamma \beta}}A\ket{\boldsymbol{\gamma \beta}}$ to denote the QAOA$_p$ expectation value of an operator~$A$. %
We refer generically to $\gamma_1,\beta_1,\dots,\beta_p$ and $p$ as QAOA \textit{parameters}, and will often use the term QAOA \textit{angles} when the QAOA \textit{level} $p$ is understood. The quantum gate depth for implementing a QAOA$_p$ circuit is related to but not the same as~$p$.

Repeated preparation and measurement of a QAOA$_p$ state on a quantum computer produces a collection of candidate solution samples, with the overall best solution found returned. The samples may be used to estimate quantities such as $\langle C\rangle_p$ which in turn may be used to search for better phase and mixing angles, e.g., variationally. Alternatively, angles may be selected a priori or restricted to a specific domain or schedule. %

The performance of such algorithms arise from the action of the mixing operator, %
which for the choice of $B$ of  \eqref{eq:mixingHam} relates to Hamming %
distances between bitstrings. Its action on quantum states can be readily interpreted in terms of classical bitstrings (with respect to the computational basis). 
Indeed, let $x^{(j)}$ denote the bitstring $x\in\{0,1\}^n$ with its $j$th bit flipped. 
Then from its action on basis states $X_j\ket{x}=\ket{x^{(j)}}$
we identify each Pauli operator $X_j$ as the $j$th \textit{bit-flip} operator. %
Hence the Hamiltonian~$B$ maps a given $\ket{x}$ to a %
superposition of strings differing from~$x$ by a single bit-flip
\begin{equation} \label{eq:mixHam2}
    B\ket{x}=\sum_{j=1}^n\ket{x^{(j)}} =: \sum_{y\in nbd(x)}\ket{y}, 
\end{equation}
where $nbd(x)\subset\{0,1\}^n$ denotes the %
Hamming distance-$1$ neighbors of $x$. Thus we see that the \textit{quantum} operator $B$ induces a \textit{classical neighborhood} structure on the Boolean cube. %
Similarly, %
iteratively applying \eqref{eq:mixHam2} one easily shows the action of $B^k$ relates to Hamming distance-$k$ neighborhoods. Indeed, for the QAOA mixing operator $U_M(\beta)=e^{-i\beta B}$, 
 the identities $X^2=I$ and $[X_i,X_j]=0$ %
give 
$U_M(\beta)%
=\prod_{j=1}^n (I \cos(\beta) - i\sin(\beta)X_j)$, %
and so its matrix elements %
with respect to the computational basis, $x,y\in\{0,1\}^n$, are
\begin{eqnarray} \label{eq:mixingmatrixelements}
\bra{x}U_M(\beta)\ket{y} 
=\cos^n(\beta)(-i\tan(\beta))^{d(x,y)}=:u_{d(x,y)},
\end{eqnarray}
which depend only on $\beta$ and the Hamming distance $d=d(x,y)$ between $x$ and $y$.\footnote{ 
The quantities $|u_d(\beta)|$, $d=1,\dots,n$ each have a single maximum at $\beta_d^* = \arctan \sqrt{d/(n-d)}$.  This shows that state transitions of increasing Hamming weight are relatively favored as $|\beta|$ increases.
This observation %
motivates a complementary sum-of-paths approach discussed in %
\sh{\appref{app:sumOfPaths}.
}}

We use this mixing neighborhood structure %
to build a general calculus relating classical functions encoding how the cost function changes over each neighborhood  
to quantum operators capturing the action of corresponding QAOA circuits. 
Our results %
formalize a number of %
results in prior literature concerning more specific settings. 
Though we focus on the transverse-field mixer~\eqrefp{eq:mixingHam}, %
our results provide a guide for applying our framework to other mixers, or initial states. 
We consider examples of more general mixing Hamiltonians and unitaries in \secref{sec:generalizedCalculus}, which naturally induce different classical neighborhoods in an analogous way. 
In each case the %
mixing operator neighborhoods are independent of the particular cost function and any neighborhood structure that it may carry.

\subsection{Cost difference and cost divergence functions}  \label{sec:costDiffs}
The aggregate %
action of the QAOA 
phase and mixing operators fundamentally relates to the underlying (classical) operations of \textit{cost function evaluations} and \textit{bit flips}.  
To use this observation to characterize the %
behavior of QAOA circuits 
we define a sequence of
classical \textit{cost difference} functions derived from the cost function.\footnote{Some of our definitions to follow are labeled in analogy with familiar notions from vector calculus as an aid to identifying their behavior. However, we note important differences in our discrete domain setting.}

For an arbitrary cost function $c(x)$ we define for  $j=1,\dots,n$ the 
\textit{jth (partial) cost difference} function~$\partial_j c$
 as 
\begin{eqnarray}  \label{eq:costdiffs}
\partial_j c\,(x):= c(x^{(j)})-c(x).
\end{eqnarray}
Each $\partial_j c$ encodes how the cost function changes with respect to flipping the $j$th bit of its input.\footnote{A related but distinct definition %
sometimes called the \lq\lq$j$th derivative\rq\rq\ is $c(x|_{x_j=1})-c(x|_{x_j=0})$; see %
\cite{boros2002pseudo}.} 
Considering the local neighborhood of single bit flips about each bitstring $x$ we define the 
\textit{cost divergence}
function as the total cost difference over each neighborhood
\begin{equation}   \label{eq:costdiv}
     dc(x) := \sum_{j=1}^n \partial_j c(x).
\end{equation}
The cost differences and 
divergence satisfy $\sum_x \partial_jc(x) = \sum_x dc(x) = 0$. 

These functions capture information about the %
cost function structure.  
We show an example in \figref{fig:fig1costdiv}. For a solution~$x^*$ that maximizes~$c(x)$, we have $\partial_j c(x^*)\leq 0$ for each $j\in [n]$, which implies $dc(x^*)\leq 0$. 
Thus $\forall j\in [n]\; \partial_j c(x^*)\leq 0$ gives a necessary condition for a string $x^*$ to give a global maximum. 
(Likewise, $\partial_j c(x^*)\geq 0$ and $dc(x^*)\geq 0$ for %
minima.) 
Hence we expect the cost divergence to be relatively large in magnitude and negative (positive) near local  maxima (minima). 
On the other hand, for \lq\lq typical\rq\rq\ bitstrings $y$ with cost close to the mean~$\langle c \rangle_0=\tfrac1{2^n}\sum_xc(x)$, %
we expect %
nearly as many possible bit flips will increase the cost function as decrease, which suggests $dc(y)\simeq 0$ for such typical %
strings~$y$. %
\begin{figure}[ht]
\centering 
\includegraphics[width=5cm]%
{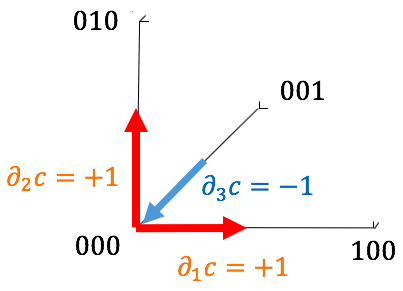}
\caption{Example: Local cost differences $\partial_ic(x)$ for the $3$-variable Max-$2$-SAT instance $c(x)=(x_1\vee x_2)+(x_2\vee \overline{x}_3)$, depicted for the bottom-left corner of the Boolean cube $x=000$; %
i.e., flipping the first or second bit increases the cost function ($\partial_1c(000)=\partial_2c(000)=1$), where as flipping the third decreases it ($\partial_3c(000)=-1$). For problems on $n$ variables we have $n$ cost difference functions. 
}
\label{fig:fig1costdiv}
\end{figure}

\subsubsection{Higher-order and mixed cost differences}
We likewise consider higher-order changes in the cost function due to multiple bit flips. This follows in the usual way from considering each partial difference $\partial_j$ as an operator %
mapping classical functions to classical functions $\partial_j:c\rightarrow \partial_j c$ which may then be iterated.  
A useful property %
is the antisymmetry relation
\begin{eqnarray} \label{eq:partialjCsym}
\partial_j c(x^{(j)}) \,=\, -\partial_j c(x).
\end{eqnarray}
Applying this twice gives 
\begin{equation}
    \partial_j^2 c (x) = \partial_j (\partial_j c(x))=-2 \partial_j c(x)
\end{equation}
which we write as $\partial_j^2 c = -2\partial_j c$, 
and hence it follows
$\partial_j^k c  =(-2)^{k-1} \partial_j c$ for $k \in \naturals$.

For $j\neq k$, %
let $x^{(j,k)}$ denote $x$ with its $j$th and $k$th bits flipped, and so $\ket{x^{(j,k)}}=X_jX_k\ket{x}$.
The second-order (mixed) cost difference function $ \partial_j \partial_k c$ %
is given by
$$  \partial_j ( \partial_k c) (x) 
= c(x)  - c(x^{(j)}) - c(x^{(k)})+ c(x^{(j,k)})$$
which is symmetric with respect to interchange of $j$ and $k$. Hence,  $\partial_j \partial_k c = \partial_j \partial_k c\,$
so mixed difference operations mutually commute. 
These functions generalize to third order cost differences $\partial_i \partial_j \partial_k c$ and higher orders in the natural way.%

We may likewise iterate the cost divergence operator $d:c\rightarrow dc$ to give
\begin{equation} \label{eq:d2c}
d^2c(x) := d(dc)(x) = -2 dc(x) + 2  \sum_{i<j}\partial_i \partial_j c(x) ,
\end{equation}
which relates to the change in cost function over the %
Hamming distance $2$ neighborhood 
of each~$x$.
More generally, the higher-order cost divergence $d^kc$ is  defined recursively as $d^kc(x):=d(d(\dots d(dc))..)(x)$, which relates to the change in cost function over neighborhoods of $k$ bit flips or less.

\subsection{Application: Leading-order behavior of QAOA}
\label{sec:smallAngleQAOA1initial}
Before developing the operator side of our framework, we  state two results relating the 
leading-order behavior of 
QAOA, generally, to the cost difference functions. %
The details are given in \secsref{sec:QAOA1}{sec:QAOAp}. We first consider QAOA with a single level. 
By %
expressing the action of a %
QAOA$_1$ circuit as a series in the parameters $\gamma,\beta$, the leading-order contributions to the measurement probabilities and cost function expectation 
(beyond the initial contributions $P_0(x)=1/2^n$ and $\langle C\rangle_0=\langle c\rangle_0$ obtained from uniform random guessing)
are %
expressible in terms of the cost divergence function. Hence, the cost divergence determines the behavior of QAOA$_1$ in the \lq\lq small-angle\rq\rq\ regime $|\gamma|,|\beta|\ll 1$; 
see \remref{rem:normInvariance} below, 
which also considers additional important quantities such as $\|C\|$.  Moreover, we show that for sufficiently small parameters %
QAOA$_p$ behaves as an effective QAOA$_1$.

\begin{theorem}[Leading-order QAOA$_1$]\label{thm1:smallAngles}
Consider QAOA$_1$ with an arbitrary cost function~$c(x)$. %
Then to lowest order in $\gamma,\beta$ the probability $P_1(x)$ of measuring a %
bitstring~$x\in\{0,1\}^n$ is given by
\begin{equation} \label{eq:px1}
    P_1(x) \,\simeq \,\,
P_0(x) \,-\, \frac{2\gamma \beta}{2^n}dc(x), %
\end{equation}
with 
expected value of the cost Hamiltonian (cost function) 
\begin{eqnarray} \label{eq:expecC1}
\langle C \rangle_1  \,\simeq \, \langle C \rangle_0 \, -\,  \frac{2 \gamma \beta}{2^n} \sum_x c(x) \,dc(x), %
\end{eqnarray}
where the neglected terms in both equations are 
degree~4 or higher in $\gamma,\beta$. %

Moreover, %
for any non-constant $c(x)$ we have %
\begin{eqnarray}\label{eq:expecC1b}
\sum_x c(x) dc(x) < 0.
\end{eqnarray}
\end{theorem}
The theorem gives novel insight into QAOA.
We see that %
to leading order, QAOA$_1$ causes %
probability to flow in (or out) %
of each bitstring $x$ %
in proportion to its cost divergence $dc(x)$. 
In particular, for a maximization problem with $c(y)\geq 0$ for every $y$, when $0< \gamma,\beta \ll 1$ %
from \eqref{eq:expecC1} and  \eqref{eq:expecC1b} we see that to leading order \textit{probability flows to increase the cost expectation value}. 
Similarly, selecting the product $\gamma \beta <0$ guarantees that probability flows to reduce the cost expectation value, to lowest order. 

As an example, for MaxCut, \tabref{tab:summary} below has
$dc(x)=2m-4c(x)$, 
so \thmref{thm1:smallAngles} gives  
$$ \langle C \rangle_1 - \langle C \rangle_0 \simeq -4m\gamma\beta\langle C \rangle_0 + 8\gamma\beta\langle C^2 \rangle_0=2\gamma\beta m$$
to leading order for $|\gamma|,|\beta|\ll 1$, 
where notably no terms quadratic in $m$ contribute. 
In particular, when $\gamma\beta$ is positive, the change in cost function $\langle C \rangle_1 - \langle C \rangle_0 $ is strictly positive as, e.g.,  $\gamma,\beta\rightarrow 0^+$. Such information may %
help in %
initial selection and search for suitable algorithm parameters. 

We %
bound the error arising from the neglected higher-order terms in \eqsref{eq:px1}{eq:expecC1} in \secref{sec:QAOA1}, from which it follows that angles can always be selected such that QAOA beats random guessing in expectation; see \secref{sec:randomGuessing}. %
In terms of the solution probabilities, observe that, while not the same, the %
lowest order approximation \eqref{eq:px1} is reminiscent of a step of classical gradient descent for optimizing functions over continuous domains, in this case the $2^n$-dimensional probability vector, with \lq\lq learning rate\rq\rq\ proportional to the product $\gamma\beta$.  
Observe that, letting $\widetilde{P_1}(x)$ denote the righthand side of \eqref{eq:px1}, the property $\sum_x dc(x)=0$ implies $\widetilde{P_1}(x)$ gives a normalized probability distribution (assuming large enough $n$ or sufficiently small  $|\gamma|,|\beta|$ such that $\widetilde{P_1}(x)\geq 0$), which implies \eqref{eq:px1} may be interpreted classically. %
Indeed, using the theorem we show a simple classical randomized algorithm based on coin flipping in \secref{sec:smallQAOA1class} that reproduces the lowest-order QAOA probabilities, and hence samples from QAOA circuits to small error in a suitably defined \lq\lq small-angle\rq\rq regime.

The proof of \thmref{thm1:smallAngles} is relatively straightforward using our framework and is given in \secref{sec:smallAngleQAOA1}. We emphasize that the proof %
shows the 
other possible terms up to third order (i.e., those proportional to $\gamma,\gamma^2,\gamma^3,\beta,\beta^2,\beta^3,\gamma^2\beta$ or $\gamma\beta^2$) to be identically zero in both~\eqsref{eq:px1}{eq:expecC1}, in general. We show the nonzero contributing terms up to fifth order explicitly in \secref{sec:QAOAp}. 
Furthermore, we show that for any cost Hamiltonian given as a linear combination of Pauli $Z$ operators $C=a_0I+\sum_ja_jZ_j +\sum_{i<j}a_{ij}Z_iZ_j+\dots$, %
the quantity $\tfrac1{2^n}\sum_x c(x) \,dc(x)$ appearing in \eqref{eq:expecC1} may be %
efficiently computed from the Hamiltoian coefficients~$a_\alpha$; see \lemref{lem:expecCDC} below.

For $p>1$ layers we similarly derive %
leading order expansions %
for QAOA$_p$, with respect to all or some angles being treated as %
expansion parameters, in \secref{sec:smallAngleQAOAp}.  
As remarked, the leading order contribution to probabilities and cost expectation for QAOA$_p$ is shown to be of the same form as that of QAOA$_1$, but with a generalized coefficient that depends on all $2p$ phase and mixing angles. 

\begin{theorem}[3rd order QAOA$_p$] \label{thm:allanglessmall}
For QAOA$_p$ we have cost expectation 
\begin{eqnarray} \label{eq:expecCp0a}
\langle C \rangle_p \, %
&=& %
\langle C \rangle_{0}
\,-\, \frac2{2^n} \left(\sum_{1\leq i\leq j}^p \gamma_i \beta_j\right) \sum_x c(x)dc(x)\,
  + \,\dots
\end{eqnarray}
where the terms not shown to the right are order four or higher in $\gamma_1,\beta_1,\dots, \gamma_p,\beta_p$. %
\end{theorem}

\begin{table}[ht]
\centering
\begin{tabular}{ |c|c|c||c|} 
 \hline
 Label & Symbol & Definition & Example: Value for MaxCut\\ 
  \hline
    Cost function & $c(x)$ & $c:\{0,1\}^n \rightarrow \reals$ & $\sum_{(ij)\in E } x_i \oplus x_j$\\ 
      $j$th partial difference & $\partial_j c(x)$ & $c(x^{(j)})-c(x)$  & %
      $|N_j|-2\sum_{i \in N(j)} x_i \oplus x_j$ \\
 Cost divergence &  $dc(x)$ & $\sum_{i=1}^n \partial_i c(x)$ &  $2|E| - 4c(x)$ \\
   $\ell$th cost divergence &  $d^\ell c(x)$ & $\sum_{i=1}^n \partial_i d^{\ell-1} c(x)$ &  $(-4)^\ell (c(x)-\tfrac{|E|}2)$  \\
 \hline
Mixing Hamiltonian & $B$ & %
transverse-field & $\sum_{i=1}^n X_i$\\
  Cost Hamiltonian& $C$  & $C\ket{x}=c(x)\ket{x}$  &  $\tfrac{|E|}{2}I-\tfrac12\sum_{(ij)\in E} Z_i Z_j$\\
Cost divergence Ham.  & $DC$ & $DC \ket{x}= dc(x)\ket{x}$  & $2\sum_{(ij)\in E} Z_i Z_j$ \\
$\ell$th cost div. Ham.  & $D^\ell C$ & $D^\ell C \ket{x}= d^\ell c(x)\ket{x}$ & $-\tfrac12 (-4)^\ell \sum_{(ij)\in E} Z_i Z_j$ \\
Cost gradient op. & $\nabla C$ & $[B ,C]$ & $ i \sum_{(ij)\in E}(Y_iZ_j + Z_iY_j)$ \\
Order 2 cost gradient & $\nabla^2 C$ & $[B,[B ,C]]$ & $4 \sum_{(ij)\in E}(Y_iY_j - Z_iZ_j)$ \\
$\ell$th-order gradient & $\nabla^\ell C$ & $[B,[B,\dots [B ,C]]..]$ &
$\begin{cases}
      \,4^{\ell-1}\,\nabla C & \text{for  $\ell$ odd}\\
      \,4^{\ell-2}\,\nabla^2 C & \text{for  $\ell$ even}\\
    \end{cases}$\\
Gradient w.r.t. $A$ & $\nabla_A $ & $[A ,\cdot]$ &  (various; see \secsref{sec:costGrads}{sec:QAOA1})\\ %
Mixed gradient  & $\nabla_C \nabla C $ & $[C,[B ,C]]$ & $- \sum_{i=1}^n ( |N_i| X_i + \sum_{j,\ell \in N(i)} X_i Z_j Z_\ell)$  \\
 \hline
  Initial state & $\ket{s}$ & $\ket{+}^{\otimes n}=\frac1{\sqrt{2^n}}\sum_x\ket{x}$& uniform distrib. of possible cuts \\
  Initial expectation & $\langle \,\cdot\, \rangle_0$ & $\bra{s} \cdot \ket{s}$ & $\langle C \rangle_0 = \tfrac1{2^n}\sum_x c(x)=|E|/2$\\%\tfrac{|E|}2$ \\
  \hline
\end{tabular}
\caption{
Summary of important %
notation, with the MaxCut problem, which seeks to partition the vertices of a graph so as to maximize the number of cut edges, as an example. 
The top of the table shows functions on bits and the middle shows related operators on qubits. 
The $n$-qubit computational basis states $\ket{x}$ represent bitstrings $x\in\{0,1\}^n$, and the bottom of the table shows %
the %
standard QAOA initial state we primarily consider in this paper, though other choices are possible. 
The operators $X_j,Y_j,Z_j$ are the single qubit Pauli matrices acting on the $j$th qubit. 
The symbols $\partial_j,d$ denote right-acting operators on functions, %
and likewise $D,\nabla$ denote superoperators on the space of $n$-qubit matrices; in both cases $k$th powers indicate operator iteration. The rightmost column shows their realization for MaxCut on a graph $G=(V,E)$ with $|V|=n$ nodes, %
where here the graph neighborhood of each vertex variable is denoted $N(j)=\{i:(ij)\in E\}\subset [n]$ with  
$|N(j)|=\deg(j)$ the degree of vertex~$j$. %
The formula $dc(x)=2|E|-4c(x)$ follows from %
\eqref{eq:DCb} and 
is particular to MaxCut, as is the formula given for $\nabla^\ell C$ which is a special case of \lemref{lem:quboGrads}. Similar but more complicated formulas can be derived for higher-order mixed gradients $\nabla_{A_\ell}^{a_\ell}\dots\nabla_{A_2}^{a_2}\nabla_{A_1}^{a_1}C$ %
and for applications to problems other than MaxCut.
}
\label{tab:summary}
\end{table}

\begin{rem}
\thmsref{thm1:smallAngles}{thm:allanglessmall} apply more generally than stated. 
Viewed as expansions about the Identity operator $I$, i.e., the QAOA circuit with all parameters set to $0$, we immediately obtain corresponding results about any choice of parameters $\gamma'_j,\beta'_j$ such that the resulting QAOA circuit becomes equal or proportional to $I$. In such cases the resulting expansions are given in terms of the parameter differences $\gamma_j-\gamma'_j$ and $\beta_j-\beta'_j$. 

Indeed, QAOA circuits are often periodic: 
\eqref{eq:mixingmatrixelements} shows that the QAOA mixing operator $U_M(\beta)$ is proportional to the Identity when $\beta$ is an integer multiple of $\pi$, and hence $\langle C\rangle_p$ is (at least) $\pi$-periodic in each mixing angle $\beta_j$. The periodicity of the phase operator $U_P(\gamma)$ relates to the coefficients of the cost Hamiltonian, which are uniform for many problems (e.g. MaxCut) and hence lead to similar insights (cf. \eqref{eq:expecC1maxcut}). 
For %
simplicity we do not deal explicitly with periodicity of QAOA circuits  %
in the remainder of the paper, though similar considerations apply to many of our results to follow. 
\end{rem}%

\subsubsection{Example: leading-order QAOA versus a simple quench}\label{sec:quench}
Our framework facilitates comparison between %
different quantum circuit ans\"atze,
in particular by allowing a term-by-term comparison of the resulting series expressions.
To illustrate this we show that  %
to leading order the QAOA cost expectation value is the same as that of a simple quantum quench algorithm, %
where for %
a fixed Hamiltonian $H=aC+bB$ %
the state $\ket{\tau}:=e^{-i\tau H}\ket{s}$ 
is prepared and measured in the computational basis. %
Such a quantum approach to classical optimization is considered, e.g.,  in~\cite{hastings2019duality,callison2020energetic}. 
\begin{prop} \label{prop:quench}
Let $\tau > 0$, $\gamma,\beta\in\reals$, and Hamiltonian $H$ be such that $\tau H=\sqrt{2}\beta B + \sqrt{2}\gamma C$. Then a quench of $H$ for time $\tau$ produces the same leading-order change \eqrefp{eq:expecC1} in $\langle C\rangle$ as QAOA$_1(\gamma,\beta)$, i.e., 
\begin{equation} \label{eq:quenchLeadingOrder}
      \langle C \rangle_\tau   \,=\, \bra{\tau}C\ket{\tau} \,=\, \langle C \rangle_0   %
      -\,  \frac{2 \gamma \beta}{2^n} \sum_x c(x) \,dc(x)\dots
\end{equation}
where the terms to the right are order four or higher in $\gamma,\beta$.

For QAOA$_p$, similarly, selecting instead $\gamma,\beta\in\reals$ such that $\gamma\beta=\sum_{1\leq i\leq j}^p \gamma_i \beta_j$ gives $\langle C \rangle_\tau$ as in \eqref{eq:expecCp0a}. 
\end{prop}

Hence, QAOA with sufficiently small parameters resembles a short-time quench. As we show, our framework allows for simple derivation of such formulas.  %
Higher order contributing terms may be systematically derived with our framework.

\begin{rem}
A simple quench conserves $H$ %
in expectation %
$\langle H\rangle_{\tau} = \langle H \rangle_0$, so for $H=aC+bB$ linearity gives 
   $ a(\langle C \rangle_\tau -   \langle C \rangle_0 )
     = -b\,(\langle B \rangle_\tau -   \langle B \rangle_0)= bn -b\langle B \rangle_\tau$, and 
thus the change in $\langle C\rangle$ %
is determined from the change in $\langle B\rangle$. 
As $-n \leq \langle B \rangle \leq n$ for any quantum state $\ket{\tau}$, we have $|\langle C \rangle_\tau -   \langle C \rangle_0| \leq 2|b/a|n$ for any possible quench time $\tau$. 
Therefore, when $b/a$ is constant, the above quench can shift the expectation value of $C$ by at most $ O(n)$, whereas $|C(x^*)-\langle C\rangle_0|$ may in general be much larger (e.g., $O(n^2)$ for MaxCut).  
This shows the limitations of a simple quench, i.e., the quench time or ratio $|b/a|$ must often grow with the problem size if large improvements in the expected value of the cost function are sought. Alternatively, to circumvent this requirement, the Hamiltonian $B$ can be replaced or augmented with additional $k$-local terms as proposed in~\cite{hastings2019duality}.
QAOA does not in general obey %
a similarly simple conservation law. %
\end{rem}

In the remainder of \secref{sec:calculus} 
we derive the operators used in our framework and show the connection to the cost difference functions, as well as a number of useful properties. 
For the convenience of the reader %
we %
both summarize our main notation in \tabref{tab:summary}
and as an guiding example show its explicit realization for MaxCut, an NP-hard constraint satisfaction problem considered extensively for QAOA~\cite{Farhi2014,wang2018quantum,hadfield2018thesis,crooks2018performance,zhou2018quantum,guerreschi2019qaoa}. 

\subsection{Hamiltonians representing cost differences} \label{sec:costHam}
For a cost function $c:\{0,1\}^n\rightarrow \reals$ %
with corresponding $n$-qubit diagonal\footnote{Throughout the paper \textit{diagonal} operators are %
with respect to the computational basis.} 
Hamiltonian~$C$ given in \eqref{eq:costHam}, i.e., acting as $C\ket{x}=c(x)\ket{x}$ for each $x\in\{0,1\}^n$, we may uniquely express~$C$ as a multinomial in the Pauli $Z$ operator basis as
\begin{equation} \label{eq:costHamZs}
C=a_0 I + \sum_{j=1}^n a_j Z_j + \sum_{i<j} a_{ij}Z_iZ_j+\dots    
\end{equation}
We say $C$ is $k$-local for $k$ the largest degree of a nonzero term in \eqref{eq:costHamZs}; we further discuss problem and Hamiltonian locality considerations in~\secref{sec:probLocality}. 
The coefficients $a_\alpha$ of $C$ are given by the Fourier expansion of the cost function $c(x)$; see \cite{hadfield2018representation} for details. In particular the Identity coefficient~$a_0=
\tfrac1{2^n}\sum_x c(x)$ is %
the cost expectation under the uniform probability distribution. 

For example, for constraint satisfaction problems with cost function given by a sum of %
clauses $c=\sum_j c_j$, with each clause acting on at most $k$ variables, $k=O(1)$, we may efficiently construct a $k$-local cost Hamiltonian of the form \eqref{eq:costHamZs} with a number of terms polynomial in $n$. For such cost Hamiltonians the corresponding QAOA phase operator $U_P(\gamma)=e^{-i\gamma C}$ is efficiently implementable in terms of single-qubit and CNOT gates~\cite{hadfield2018representation}.  
Motivating examples are the NP-hard optimization problems MaxCut and Max-$2$-SAT, studied for QAOA in \cite{Farhi2014,Wecker2016training,wang2018quantum,wilson2019optimizing}, which each have quadratic cost Hamiltonians.

Using the results of \cite{hadfield2018representation}
we may similarly lift the classical 
partial difference and cost divergence functions 
of \eqsref{eq:costdiffs}{eq:costdiv}
to diagonal Hamiltonians~$\partial_j C$ and $DC=\sum_j \partial_j C$, respectively, which act on computational basis states $\ket{x}$ as
$\partial_j C\ket{x} = \partial_j c (x) \ket{x}$ and $DC\ket{x} = dc(x)\ket{x}$. 
In particular, given a cost Hamiltonian $C$ explicitly as in \eqref{eq:costHamZs}, the operators $\partial_j C$ and $DC$ are easily computed from %
the useful relations    
\begin{equation}  \label{eq:djC}
\partial_j C = X_jCX_j - C , 
\end{equation}
$j=1,\dots, n$, 
from which 
it trivially follows that 
\begin{equation} \label{eq:djCb}
 \bra{s}\partial_j C \ket{s} =0 \,\;\;\;\;\;  \text{and} \;\;\;\;\;\, \bra{s} DC \ket{s} =0.  
\end{equation}
Results such as \eqsref{eq:djC}{eq:djCb} apply generally to any cost Hamiltonian. 

Next observe that 
from the %
Pauli anti-commutation relation $\{X,Z\}:=XZ+ZX=0$, 
the action of $X_j\cdot X_j$ on $C$ either flips the sign of or leaves alone each term of $C$ in \eqref{eq:costHamZs}.
Hence if for each $j$ we define a partition of the terms of %
$C$  as $C=C^{\{j\}}+C^{\{\setminus j\}}$, where $C^{\{j\}}$ is defined as the partial sum of the terms of \eqref{eq:costHamZs} that contain a $Z_j$ factor, and $C^{\{\setminus j\}}$ the remainder, %
then 
\begin{equation} \label{eq:partialjC}
    \partial_j C = -2 C^{\{j\}}.
\end{equation}
For example, for MaxCut, $C^{\{j\}}$ contains a $ZZ$ term for each neighbor (adjacent vertex) of the $j$th vertex. 
Summing over $j$, the cost divergence Hamiltonian satisfies 
\begin{equation} \label{eq:DCa}
    DC = -2 \sum_{j=1}^n C^{\{j\}},
\end{equation}
Each strictly $k$-local term in $C$ appears $k$ times in the sum above.
We may alternatively partition a 
 $k$-local Hamiltonian as  $C=C_{(0)}+C_{(1)}+\dots C_{(k)}$, where each $C_{(j)}$ contains strictly $j$-local terms. Then we also have 
\begin{equation} \label{eq:DCb}
    DC=-2(C_{(1)}+2C_{(2)}+\dots+kC_{(k)}),
\end{equation}
Thus we see that for general cost Hamiltonians $C$, we may represent %
$DC$ either in terms of the problem locality (i.e., the problem graph structure reflected in the $C^{\{j\}}$) as in \eqref{eq:DCa}, or operator locality of~$C$ as in \eqref{eq:DCb}. Both of these representations of $DC$ are useful in our analysis to follow.

Furthermore, it is often useful to view $D:C\rightarrow DC$ as a superoperator on the space of diagonal Hamiltonians. Iterating this operation $D^\ell C := D(D(\dots (DC))..)$ gives
\begin{equation} \label{eq:DlC}
    D^\ell C = (-2)^\ell \sum_{j=0}^k j^\ell C_{(j)}
\end{equation}
for $\ell = 0,1,2,...$, so we see $D^\ell C$ closely relates to $C$ (and moreover $d^\ell c$ relates to $c$ similarly). 
Hence if $C$ is strictly $k$-local, then $D^\ell C=(-2k)^\ell C$ is proportional to $C$, and so they represent the same classical function up to a constant~\cite{hadfield2018representation}. 
For example, the cost Hamiltonian for MaxCut is $2$-local, up to its Identity term; see \tabref{tab:summary}.

\begin{rem} \label{rem:probProj}
An important family of diagonal Hamiltonians are the projectors 
\begin{equation} \label{eq:Hx}
    H_y := \ket{y}\bra{y}  \;\;\;\; \text{ for } \;\; y\in\{0,1\}^n, 
\end{equation}
which as observables give the probability of measuring the string $x$ for a generic state $\ket{\psi}$ in expectation as $P_\psi(y)=|\bra{y}\ket{\psi}|^2=\bra{\psi}H_y\ket{\psi}$. Importantly, $H_y$ represents the classical function $\delta_y(x)$ in the sense of \eqref{eq:costHam} which is $1$ 
when its input $x$ equals $y$ and $0$ otherwise.
\end{rem}

\subsection{Cost gradient operators} \label{sec:costGrads}
Non-diagonal operators generated by (iterated) commutators %
of~$B$ and~$C$ play a fundamental role in our approach to %
analyzing QAOA. 
Specifically, given a cost Hamiltonian~$C$, we identify the commutator of $B$ and~$C$ as the \textit{cost gradient operator} $\nabla C := [B,C]$, with $\nabla:=[B,\cdot]$.\footnote{Commutator operators such as $\nabla=[B,\cdot]$ are often called \textit{adjoint operators} %
and written $\rm{ad}_B:=\nabla$, sometimes with an additional $i=\sqrt{-1}$ factor $\rm{ad}_B:=i\nabla$ in the context of Lie algebras \cite{woit2017quantum}.  
These operators are  \textit{inner derivations} on matrices, i.e., they satisfy the \textit{Leibniz property} $\nabla(C_1C_2) = (\nabla C_1)C_2 + C_1(\nabla C_2)$; see e.g. \cite[Sec. 5]{hatano2005finding}.} 
We refer to $\nabla C$ as a \lq\lq gradient\rq\rq\ from its action on computational basis states $\ket{x}$, and on the equal superposition state $\ket{s}=\frac{1}{\sqrt{2^n}}\sum_x \ket{x}$, which %
is easily expressed in terms of the cost difference functions is reminiscent of vector calculus over %
the reals.\footnote{Our notion of gradient is 
distinct from others %
in the literature, e.g., %
the %
(vector calculus) gradient of the classical function $f(\vec{\theta}\,):=\bra{\psi} U^\dagger(\vec{\theta}) A U(\vec{\theta}) \ket{\psi}$ for a parameterized quantum circuit $U(\vec{\theta}) \ket{\psi}$ \cite{crooks2019gradients}.} 

Observe that $\nabla C$ is skew-adjoint and traceless; in general we can make $\nabla C$ into a Hamiltonian if we multiply by $\pm i$.

First consider the commutators 
$\nabla_jC:=[X_j,C]=X_jC-CX_j$, $j=1,\dots,n$, so that $\nabla C =\sum_{j=1}^n \nabla_j C$. For each computational basis state $\ket{x}$ we have
$$ \nabla_jC\ket{x}\,
= (c(x)-c(x^{(j)}))\ket{x^{(j)}}
= - \partial_j c(x)\ket{x^{(j)}}.$$
Comparing the action on basis states, we may relate the operators $\partial_jC$ and $\nabla_j C$ as 
$\nabla_j C = \partial_j C \,X_j = - X_j\, \partial_j C.$
Thus we have 
\begin{equation} \label{eq:actionGradC}
  \nabla C \ket{x} %
= \,-\sum_{j=1}^n \partial_j c(x)\ket{x^{(j)}},  
\end{equation}
and it follows 
\begin{eqnarray} \label{eq:DCsuperpos0} 
 \nabla C \ket{s} \, =\,  %
\frac{1}{\sqrt{2^n}} \sum_x dc(x)\ket{x}.
\end{eqnarray}
Hence $\bra{x}\nabla C\ket{x}=0$ %
for each $x\in\{0,1\}^n$, and $\bra{s}\nabla C\ket{s}=0$ from $\sum_x dc(x)=0$.

The superoperator $\nabla:=[B,\cdot]$ mapping $C$ to $\nabla C$ cannot increase the locality of each Pauli term in $C$ of \eqref{eq:costHamZs}, as $B$ is a sum of $1$-local terms, and hence the resulting number of terms in the Pauli expansion of $\nabla C$ can increase at most by a multiplicative factor of $k$ for $k$-local $C$.
For problems with cost Hamiltonians of bounded locality, 
e.g. QUBO problems, we can easily calculate $\partial_j C$ and $\nabla C$ explicitly, %
as demonstrated in \tabref{tab:summary}.

\subsubsection{Higher-order gradients}
To connect our framework to QAOA circuits, consider the QAOA mixing operator $U_M(\beta)=e^{-i\beta B}$. 
The superoperator $i\nabla$ is the infintesimal generator of conjugation by $U_M$,
i.e., %
\begin{equation} \label{eq:conj1}
   e^{\beta B} C e^{-i\beta B} = C + \sum_{k=1}^\infty  \frac{(i\beta)^k}{k!} \nabla^k C =:e^{i\beta \nabla}C.
\end{equation}
We will elaborate on and use this relationship extensively %
in \secref{sec:QAOA1}. %
From this perspective %
our formalism naturally generalizes to %
include \textit{higher-order gradients}
\begin{equation} \label{eq:higherOrderDerivs}
    \nabla^\ell C := [B, [B, [B, \dots, [B,[B,C]]]..]].
\end{equation}
$\nabla^\ell C$ is skew-adjoint for $\ell$ odd and self-adjoint (i.e., a Hamiltonian) for $\ell$ even.

Our notation evokes but differs from %
that of vector calculus. The symbol $\nabla^2$ denotes \lq\lq $\nabla \cdot \nabla$\rq\rq\ in the latter, i.e., denotes the Laplacian operator, whereas in our notation $\nabla^2 := \nabla \nabla$ and $\nabla^2 C = [B,[B,C]]$.
Indeed, we may instead identify $\nabla^2 C$ with the  \textit{Hessian} of $c(x)$ from its action on basis states
\begin{equation}   \label{eq:grad2Cx}
\nabla^2 C \ket{x} = 
- 2 \sum_{j=1}^n \partial_jc(x )\ket{x}
+ 2\sum_{j<k} \partial_j  \partial_k c (x) \ket{x^{(j,k)}}.\end{equation}
Hence for each computational basis state $x$ we have
$ \bra{x}\nabla^2 C \ket{x} = -2 dc(x).$
On the uniform superposition state,  
the symmetry $\partial_j \partial_k c(x^{(j,k)}) = \partial_j \partial_k c(x)$
gives
\begin{equation}   \label{eq:grad2Cs}
 \nabla^2 C \ket{s}=
 \frac{1}{\sqrt{2^n}} \sum_x d^2c(x) \ket{x},
\end{equation}
where $d^2c(x)$ is defined in \eqref{eq:d2c}. This implies $\bra{s}\nabla^2 C\ket{s}=\tfrac1{2^n}\sum_x d^2c(x)$.  

More generally, for an operator $A$ which acts on $\ket{s}$ as $A\ket{s}=\frac{1}{\sqrt{2^n}}\sum_x a(x)\ket{x}$ for a real function $a(x)$, 
it follows 
\begin{equation} \label{eq:gradAs}
    \nabla A\,\ket{s} =\frac1{\sqrt{2^n}}\sum_x da(x)\ket{x}.
\end{equation}
This equation is useful for computing the action and expectation values of higher-order cost gradients, as we shall see in \secref{sec:expecVals}. 
In particular this relates the action of $\nabla^\ell C$ on $\ket{s}$ to 
higher order cost divergence functions $d^\ell c(x)$ as 
\begin{equation} \label{eq:gradCs}
 \nabla^\ell C \ket{s}=
 \frac{1}{\sqrt{2^n}} \sum_x d^\ell c(x) \ket{x}.
\end{equation}
In \lemref{lem:expecGrad} below we show
$\langle\nabla^\ell C\rangle_0:=\bra{s}\nabla^\ell C\ket{s}=0$ for $\ell \in \integers_+$ and any cost Hamiltonian, which implies $\sum_x d^\ell c(x) =0$ for any cost function. Note $\langle\nabla^\ell C\rangle_\psi \neq 0$ in general.

\begin{rem}
From the discussion for $\nabla C$ above, %
if $C$ is $k$-local then so is $\nabla^\ell C$, and the number of Pauli terms in $\nabla^\ell C$ is at most $k^\ell * \binom{n}k =O(k^\ell n^k)$. 
Hence, when $k=O(1)$, if $\ell=O(\log n)$ the number of terms remains $poly(n)$ and so we can represent and compute $\nabla^\ell C$ efficiently as a linear combination of Pauli terms. %
For arbitrary cost Hamiltonians %
with $poly(n)$ Pauli terms, the same argument applies for $\ell=O(1)$. 
\end{rem}

\subsubsection{Directional and mixed gradients}
To extend \eqref{eq:conj1} to include the QAOA phase operator, we %
require still more general %
notions of gradient operators. 
\textit{Directional cost gradients} are defined as commutators of the cost Hamiltonian  taken with respect to a general $n$-qubit operator~$A$ as 
\begin{equation}
    \nabla_A C := [A,C], 
\end{equation}
corresponding to the superoperator $\nabla_A := [A,\cdot]$. Our above definition satisfies $\nabla := \nabla_B$. 
Trivially, two compatible linear operators $A,G$ commute if and only if $\nabla_A G = \nabla_G A = 0$. 
Further observe  $\nabla_A$ is linear in~$A$ as %
$\nabla_{cA}=c\nabla_A$ for $c\in\complex$ and $\nabla_{A+G}=\nabla_A+\nabla_{G}$.

Gradients with respect to the cost Hamiltonian $\nabla_C = [C,\cdot]$ %
naturally arise in analysis of QAOA. 
Generally, higher-order (mixed) gradients of $C$ will likewise reflect the structure of the cost function over neighborhoods of increasing Hamming distance. Clearly, $\nabla_C C =0$ and $\nabla_C B = - \nabla C$. 
The operator $\nabla_C \nabla C := [C,[B,C]]$ is easily shown to act on basis states as 
\begin{equation}  \label{eq:DCDCx}
    \nabla_C \nabla C \ket{x} = -\sum_{j=1}
^n \left(\partial_jc(x)\right)^2\ket{x^{(j)}},
\end{equation}
and on the initial state $\ket{s}$ as 
\begin{eqnarray} \label{eq:DCDCs}
    \nabla_C \nabla C \ket{s} &=&\tfrac{1}{\sqrt{2^n}}\sum_x\left(-\sum_{j=1}^n \left(\partial_jc(x)\right)^2\right)\ket{x}
=\tfrac{1}{\sqrt{2^n}}\sum_x \left(c(x) dc(x) - 
\sum_{j=1}
^n c(x^{(j)})\partial_j c(x)
\right)\ket{x}, \nonumber
\end{eqnarray}
which implies 
 $\langle \nabla_C \nabla C \rangle_0 =\frac2{2^n}\sum_x c(x) dc(x) \leq 0$ for all cost Hamiltonians $C$; see \lemref{lem:expecCDC} below. 
Similarly, for $\ell=1,2,3,\dots$ it follows from induction that  
\begin{equation} \label{eq:DlCDCx}
    \nabla_C^\ell \nabla C \ket{x} = - \sum_{j=1}^n (\partial_j c(x))^{\ell+1} \ket{x^{(j)}},
\end{equation}
which using $\partial_j c (x^{(j)})=-\partial_j c(x)$ gives  
\begin{equation} \label{eq:DlCDCs}
    \nabla_C^\ell \nabla C \ket{s} =  \frac{-1}{\sqrt{2^n}} \sum_x \sum_{j=1}^n (-\partial_j c(x))^{\ell+1}  \ket{x}.
\end{equation}

Finally, from the %
Jacobi identity~\cite{woit2017quantum} for Lie algebras 
it follows that any triple of $n$-qubit operators %
$F,G,A$ satisfy 
\begin{equation} \label{eq:Jacobi}
    \nabla_F \nabla_G A  + \nabla_A \nabla_F G  + \nabla_G \nabla_A F= 0.
\end{equation}
Applying this property to %
$B$, $C$, and $\nabla C$ gives the useful identities
\begin{equation} \label{eq:costJacobi}
    \nabla \nabla_C \nabla C = \nabla_C \nabla^2 C,
\end{equation}
which in particular implies $\nabla_C \nabla \nabla_C \nabla C = \nabla_C^2 \nabla^2 C$. %
Similarly, in general we have 
\begin{equation} \label{eq:crossOperator}
    \nabla_{\nabla C} A =-\nabla_A\nabla C =(\nabla \nabla_C - \nabla_C \nabla)A
\end{equation}
for any matrix $A$ acting on $n$-qubits.\footnote{%
\eqref{eq:crossOperator} implies that \lq\lq gradients with respect to gradients\rq\rq\ can be written in terms of gradients with respect to $B$ and $C$, for example we have 
$\nabla_{\nabla C}\nabla^2 C = \nabla \nabla_C \nabla^2 C-\nabla_C \nabla^3 C$, 
$\nabla_{\nabla_C\nabla C}C=-\nabla^2_C \nabla C$ and $\nabla^2_{\nabla C}C=-\nabla_{\nabla_C}\nabla_C \nabla C
=\nabla^2_C \nabla^2 C-\nabla \nabla^2_C \nabla C$.} 
Further such identities %
may be derived, e.g., by applying higher-order Jacobi identities \cite{alekseev2016higher}, or when considering particular problems.

\subsubsection{General gradient operators} \label{sec:genGrads}
General higher-order mixed cost gradient operators %
$\nabla_{A_\ell}^{b_\ell}\nabla_{A_{\ell-1}}^{b_{\ell-1}}\dots \nabla_{A_1}^{b_1} C$  follow in the natural way. %
We refer to $\sum_j b_j$ 
as the \textit{order} of a gradient operator, where for convenience we refer to all such %
operators as \textit{cost gradient operators}.  
We remark that nested commutators similarly appear in a number of quantum computing applications such as, e.g., analysis of Hamiltonian simulation with Trotter-Suzuki formulas~\cite{childs2019theory}. 

We give the following general  characterization of general gradients with respect to $B$ and $C$ that will appear in our application to QAOA. 
\begin{lem}  \label{lem:genGrad}
For $C,C'$ diagonal Hamiltonians (representing %
classical cost functions as in \eqref{eq:costHam}), 
we define the general cost gradient operator
\begin{equation} \label{eq:costGradOpGen}
          A=\nabla^{b_\ell}\nabla_C^{b_{\ell-1}}\dots\nabla^{b_3}\nabla_C^{b_2}\nabla^{b_1} C'.%
\end{equation}
for $(b_1,b_2,\dots,b_\ell)$ a tuple of positive integers. Then $A$ has the following properties:
\begin{enumerate}
    \item $A$ is a real matrix with respect to the computational basis:  $a_{xy}:= \bra{x}A\ket{y}\,\in \reals$. 
    
    Hence, there exists a real function $a(x)$ such that $A\ket{s} = \frac{1}{\sqrt{2^n}}\sum_x a(x) \ket{x}$. 
     \item $A$ is traceless, which implies $\sum_x a_{xx}=0$, and $A$ is either self- or skew-adjoint: %
    \begin{enumerate}
     \item If $\sum_j b_j$ is even then %
     $A^\dagger = A$, which implies $a_{xy}=a_{yx}$ for each $x,y\in\{0,1\}^n$. %
    \item Else if $\sum_j b_j$ is odd then %
    $A^\dagger = - A$, which implies $a_{xx}=0$ (i.e., $A$ has no diagonal part), and $a_{yx}=-a_{xy}$ for each $x,y\in\{0,1\}^n$. %
    \end{enumerate} 
    \item Each %
    tuple %
    $(b_\ell,\dots,b_2,b_1)\subset \{1,2,3,\dots\}^\ell$ 
    identifies a cost gradient operator%
    ~\eqrefp{eq:costGradOpGen}. However this labeling is not canonical as it is many-to-one. 
\end{enumerate}
\end{lem}
\noindent %
Here we have considered distinct $C,C'$ for generality, as we often but not always have $C'=C$ in application of the lemma.
\begin{proof}
The first point follows from expanding each of the commutators in $A$ in turn to give a linear combination of ordered products of $B$, $C$, and $A$. As each matrix multiplication consists of strictly real quantities, the resulting computational basis matrix elements $a_{xy}$ of $A$ are real and $A$ acts on $\ket{s}$ as %
$A\ket{s}=\tfrac1{\sqrt{2^n}}\sum_x (\sum_y a_{xy}) \ket{x}$. %
The second point follows from the definition of the commutator, and %
the cyclic property and linearity of the trace operation. Points 2.a and 2.b follow observing the commutator of a self-adjoint and a self- (skew-)adjoint operator is skew- (self-)adjoint, and applying the first point. 
For the third point, %
\eqref{eq:costJacobi} shows the mapping from positive integers %
is not injective %
in general. %
\end{proof}

Each $A=\nabla^{b_\ell}\nabla_C^{b_{\ell-1}}\dots\nabla_C^{b_2}\nabla^{b_1} C'$ may be uniquely decomposed into its diagonal and off-diagonal parts $A=A_{diag}+A_{non-diag}$ with respect to the computational basis, such that %
$\langle A\rangle_0=\langle A_{non-diag}\rangle_0$ and $\nabla_C A = \nabla_C A_{non-diag}$. In the skew-adjoint case $A_{diag}=0$. In the self-adjoint case, we have corresponding real functions $a(x)=a_{diag}(x)+a_{non-diag}(x)$, with $\sum_x a_{diag}(x)=0$. %
We will use these properties in the proofs of \appref{app:expecVals}.

\subsection{Initial state expectation values}
\label{sec:expecVals} 
We next show several results concerning the computation of initial state expectation values of cost gradient operators, which show their dependence on the cost difference functions; %
roughly speaking, \emph{powers of $\nabla_C$ %
correspond to powers of the cost function $c(x)$ and its coefficients, whereas powers of $\nabla$ %
correspond to cost differences over greater Hamming distances.} Looking forward to our application to QAOA, we will derive relevant expectation values as series expressions in the algorithm parameters and initial state expectation values of different cost gradients; these lemmas allow us identify a significant fraction of the resulting terms to be identically zero which greatly simplifies our results and proofs. 
Some %
additional results are given in \appref{app:expecVals}.

Observe that the expectation value of any $n$-qubit matrix $A$ with respect to  $\ket{s}=\frac{1}{\sqrt{2^n}}\sum_x\ket{x}$ relates to the trace of $A$ as 
\begin{equation} \label{eq:generalInitStateExpec}
   \langle A\rangle_0= \frac1{2^n} tr(A) + \frac1{2^n} \sum_{y\neq x} \bra{y}A\ket{x},
\end{equation}
where we recall our notation $\langle A\rangle_0:=\bra{s}A\ket{s}$. 
When $A$ is diagonal %
$A=a_0I + \sum_j a_j Z_j + \sum_{i<j}a_{ij}Z_iZ_j+\dots$  
this becomes $\langle A\rangle_0 = \frac1{2^n} tr(A)=a_0$. 
Hence, %
\lemref{lem:genGrad} implies a necessary condition for the initial state expectation value of a gradient operator \eqrefp{eq:costGradOpGen} to be nonzero is that its order (sum of %
the exponents $\sum_j b_j$) is even.

\begin{lem} \label{lem:genGrads}
Let $C'$ be diagonal. 
If $\sum_j b_j$ is odd then $\langle \nabla^{b_\ell}\nabla_C^{b_{\ell-1}}\dots\nabla_C^{b_2}\nabla^{b_1} C'\rangle_0=0$. 
\end{lem}
\begin{proof}
Follows from case 2.b of \lemref{lem:genGrad} using that the parity of $\sum_j b_j$ determines the adjointness of $\nabla^{b_\ell}\nabla_C^{b_{\ell-1}}\dots\nabla_C^{b_2}\nabla^{b_1} C'$, and applying \eqref{eq:generalInitStateExpec}. 
\end{proof}
We next show two results concerning more specific cases. 
\begin{lem}  \label{lem:expecGrad}
For any %
$n$-qubit linear operator $A$ and $\ell\in\integers_+$ we have
    $\langle \nabla^\ell A\rangle_0 =0$.
\end{lem}
\begin{proof}
For %
$\ell=1$ %
as $B\ket{s}=n\ket{s}$ we have
$\bra{s}\nabla A \ket{s}=\bra{s}(BA-AB)\ket{s}= (n-n)\bra{s}A\ket{s}=0$. 
The %
general result follows inductively %
considering %
$\nabla^{\ell}A=B(\nabla^{\ell-1}A)-(\nabla^{\ell-1}A)B$.%
\end{proof}
In particular, the lemma implies $\langle \nabla^\ell C\rangle_0 = 0$ for any cost Hamiltonian $C$ and $\ell\in\integers_+$.

\sh{
\begin{rem} \label{rem:GenLems}
The results of \lemssref{lem:genGrad}{lem:genGrads}{lem:expecGrad} %
extend to more general QAOA initial states and mixing Hamiltonians in a straightforward way. In particular, the result and proof of \lemref{lem:expecGrad} only requires that the initial state is some eigenvector of the mixing Hamiltonian. Similar considerations apply to a number of our results to follow. 
\end{rem}
}

\begin{lem} \label{lem:expecCDC}
For a cost Hamiltonian $C=a_0I + \sum_j a_jZ_j + \sum_{i<j}a_{ij}Z_iZ_j +\dots$ we have
\begin{eqnarray}  \label{eq:expecDCis0}
\langle\nabla_C \nabla C \rangle_0
&=&  \frac{2}{2^n} \sum_x c(x)dc(x)%
  = \frac{-1}{2^n}\sum_x \sum_j (\partial_j c(x))^2\\
 &=&%
  -4(\sum_j a_j^2 + \sum_{i<j} 2 a_{ij}^2+\sum_{i<j<\ell} 3 a_{ij\ell}^2 +\dots)\nonumber%
 \end{eqnarray}
which implies $\langle\nabla_C \nabla C \rangle_0 \leq 0$.  
Moreover,   %
 $\langle  \nabla_C \nabla^2 C \rangle_0 = \langle   \nabla^2_C \nabla C \rangle_0
 = \langle \nabla \nabla_C\nabla C \rangle_0
    =0$.
\end{lem}

\begin{proof}
For $\langle\nabla_C \nabla C \rangle_0$, the first two equalities follow from \eqsref{eq:DCDCx}{eq:DCDCs} %
which imply $\nabla_C \nabla C = -\sum_j (\partial_j C)^2 X_j$, and hence $\langle\nabla_C \nabla C \rangle_0 \leq 0$ always. The 
third equality follows expanding as a
Pauli sum $\nabla_C \nabla C  
=  -4\sum_\sigma a_\sigma^2  \sum_{i\in \sigma}X_i + \dots$ where $\sigma$ denotes tuples of strictly increasing qubit indices as in the definition of $C$, and any Pauli terms not shown to the the right each contain at least one $Z$ factor and so from~\eqref{eq:generalInitStateExpec} have $0$ expectation with respect to $\ket{s}$. 
The final statement follows directly from \lemref{lem:genGrads}. %
\end{proof}

Results for higher-order gradients may be similarly derived. %
In \appref{app:expecVals} we show
the next (order $4$) nonzero cost gradient expectations to be 
\begin{equation} \label{eq:expecDcD3C}
    \langle \nabla_C \nabla^3 C \rangle_0 = \frac2{2^n}\sum_x c(x)d^3c(x),
\end{equation}
\begin{equation} \label{eq:expecDc3DC}
    \langle \nabla^3_C \nabla C \rangle_0 = \frac1{2^n}\sum_x \sum_{j=1}^n (\partial_j c(x))^4,
\end{equation}
\begin{equation} \label{eq:expecDc2D2C}
    \langle \nabla^2_C \nabla^2 C \rangle_0
    =\langle \nabla_C \nabla \nabla_C \nabla C \rangle_0=\frac2{2^n}\sum_x \sum_{j<k} (\partial_{j,k}c(x))^2 \partial_j\partial_k c(x).
\end{equation}
The first equality of \eqref{eq:expecDc2D2C} follows from 
 \eqref{eq:Jacobi}. 
Generalizing \eqsref{eq:expecDcD3C}{eq:expecDc3DC} to $\ell=1,3,5\dots$ gives $ \langle \nabla_C \nabla^\ell C \rangle_0 = \frac2{2^n}\sum_x c(x)d^\ell c(x)$ and $\langle \nabla^\ell_C \nabla C \rangle_0 = \frac1{2^n}\sum_x \sum_{j} (\partial_j c(x))^{\ell+1}$, %
respectively, or $0$ for $\ell$ even.  
Such quantities may %
be efficiently computed from the Pauli basis coefficients of the cost Hamiltonian~$C$; %
cf. \lemref{lem:expecCDC}.  
We show such expressions for QUBO problems in \lemref{lem:higherOrderExpectations} of \appref{app:expecVals}. 
In general these expressions reveal dependence on the cost function structure. For example, an important term for analyzing QAOA for MaxCut is $\langle \nabla^2_C \nabla^2 C \rangle_0$, which depends on the number of triangles in the %
problem graph and is zero in triangle-free cases; cf. \secref{sec:MaxCutQUBO}.

\subsection{Problem and Hamiltonian locality considerations} \label{sec:probLocality}
For the sake of generality and notational clarity, 
so far we have considered arbitrary cost functions $c(x)$. For specific problems the cost function typically has additional structure we may directly incorporate into our framework when computing various quantities of interest such as expectation values. Here we elaborate on the generalization of our framework to incorporate locality, %
which we further consider in the context of QAOA$_1$ in \secref{sec:lightcones} and QAOA$_p$ in \secref{sec:lightconesp}; see in particular \tabref{tab:locality}. 

Recalling the discussion of \secref{sec:costHam}, 
consider %
an objective function  
$c(x)=\sum_{j=1}^m c_j(x)$, with corresponding cost Hamiltonian $C=\sum_{j=1}^m C_j$.\footnote{Clearly such decompositions are not unique. The two most useful for our purposes are when each $C_j$ corresponds to a problem clause $c_j$, or when each $C_j$ corresponds to a single Pauli $Z$ term in $C$. For the latter case it's sometimes convenient to include factors or the identity in $C_j$ which does not affect gradients (commutators) of $C_j$ nor QAOA dynamics.} 
Assume each %
$c_j$ acts on a subset $N_j \subset [n]$ of the variables $x_1,\dots,x_n$ (i.e., flipping variable values in $[n]\setminus N_j$ does not change $c_j(x)$). Then each $C_j$ %
acts nontrivially only on the qubits indexed by $N_j$ %
and so is %
$|N_j|$-local. %
For example, $|N_j|=2$ for MaxCut, and $|N_j| = k$ for Max-$k$-SAT where each $c_j$ is a $k$-variable %
Boolean clause. %
We identify $N_j$ similarly %
for non-diagonal Pauli terms and Hamiltonians.

To explicitly take advantage of %
Hamiltonian locality 
we will sometimes use $\cdot|_S$ to denote an operator %
restricted (in the obvious way) to a subset $S\subset [n]$ of the qubits or variables.\footnote{More precisely, for an operator expressed as a sum of Pauli terms $A=\sum_\alpha c_\alpha \sigma_\alpha$ and $S\subset [n]$, we define $A|_S$ to be the operator obtained by discarding any terms that act on qubits outside of $S$.} 
For example, $d|_{\{1,2\}}c:=\partial_1c+\partial_2c$ and $\nabla|_{\{1,2\}}C:=[X_1,C]+[X_2,C]$. 
Hence, for each $c_j$ and corresponding $C_j$, 
the cost divergence of $c_j$ reduces to a sum over $n_j$ bit flips 
$$ dc_j(x)=d|_{N_j}c_j(x)=\sum_{i\in N_j}\partial_ic_j(x), $$
and similarly %
$$DC_j =D|_{N_j}C_j = \sum_{i\in N_j} \partial_i C_j$$
When each $C_j$ consists of a single Pauli $Z$ term \eqsref{eq:DCa}{eq:DCb} imply $DC_j =-2|N_j| C_j$.

Hence, considering the product structure of $\ket{s}=\ket{+}^{\otimes n}$, %
quantities of interest
may often be evaluated by considering 
only a relevant subset of qubits or variables. 
In particular, the operator $\nabla C_j = \nabla|_{N_j}C_j$ acts on $\ket{s}$ as (cf. \eqref{eq:DCsuperpos0})
\begin{eqnarray} \label{eq:nablaCjs}
\nabla C_j \ket{s} = \frac1{\sqrt{2^n}} \sum_x dc_j(x)\ket{x} = \left(\frac1{\sqrt{2^{|N_j|}}} \sum_{x\in\{0,1\}^{N_j}} dc_j(x)\ket{x}\right)\otimes \ket{+}^{\otimes n-|N_j|}
\end{eqnarray}
where $C_j\ket{x}=c_j(x)$, and $\{0,1\}^{N_j}$ denotes bitstrings over the variables in $N_j$. 
The term in parenthesis %
is guaranteed to be efficiently computable classically when $|N_j|$ is a constant or scales as $O(\log n)$. 
Hence it is relatively straightforward to incorporate locality where applicable into our framework. %
In particular, for any %
operator $A$ and $S\subset [n]$ we have $\langle\, A|_S\rangle_0\,= \bra{+}^{\otimes S}\,A|_S\,\ket{+}^{\otimes S}\bra{+}^{\otimes [n]\setminus S}\ket{+}^{\otimes [n]\setminus S}= \frac{1}{2^{|S|}}\bra{+}^{\otimes S}\,A|_S\,\ket{+}^{\otimes S}$. 
For example,  $\langle \nabla^\ell C_j\rangle_0=0$ for $\ell \in \naturals$ from \lemref{lem:expecGrad},
and adapting the proof of \lemref{lem:expecCDC} gives 
\begin{equation}
\langle \nabla_C\nabla C_j\rangle_0 = \frac2{2^n}\sum_x c(x)dc_j(x) = \frac2{2^{n_j}}\sum_{x\in \{0,1\}^{N_j}}c(x)dc_j(x),  
\end{equation}
which becomes $\langle \nabla_C\nabla C_j\rangle_0=-4a_j^2$ when $C_j$ is a single Pauli term $C_j=a_jZ_{j_1}\dots Z_{j_{n_j}}$. %
By linearity this gives $\langle \nabla_C\nabla C\rangle_0=\sum_j \langle \nabla_C\nabla C_j\rangle_0$. %
Similar considerations will apply to higher order cost gradients as we further discuss in \secsref{sec:lightcones}{sec:lightconesp}. 

\section{Applications to QAOA$_1$} \label{sec:QAOA1}
As employed, %
for instance, in the QAOA performance 
analysis of \cite{wang2018quantum,hadfield2018thesis} %
for specific problems, it is often advantageous to consider QAOA circuits from the Heisenberg perspective %
i.e., as %
acting by conjugation on an observable~$C$ rather than acting on a state directly. 
The Heisenberg picture naturally has a number of applications in quantum computing, e.g. \cite{gottesman1998heisenberg}. 
Indeed, for a QAOA$_p$ circuit $Q=U_M(\beta_p)\dots U_P(\gamma_1)$ we have
$$ \langle C \rangle_p =\bra{s}\,(Q^\dagger C Q) \,\ket{s} ,$$%
so computing the QAOA$_p$ expectation value of the cost function, for any $p$, %
is equivalent to 
computing the initial state expectation value of the time evolved (conjugated) observable %
$U^\dagger C U $. Here, we consider application of our framework from this perspective to QAOA$_1$. In~\secref{sec:QAOAp} we %
generalize our result to QAOA$_p$. %

First consider a general $n$-qubit Hamiltonian $H$ and corresponding unitary $U=e^{-iHt}$ for $t\in\reals$. 
From the %
perspective of Lie groups and %
algebras~\cite{woit2017quantum,hall2013quantum}, 
the superoperator $it\nabla_H=it[H,\cdot]$ is the infinitesimal generator of conjugation by $U$, i.e., $e^{it\nabla_H}=U^\dagger \cdot U$. 
Explicitly, acting 
on an %
operator $A$ we have 
\begin{equation} \label{eq:infinitesimalConj}
    U^\dagger A U %
    = A + it\nabla_H A + \frac{(it)^2}{2} \nabla^2_H A
    + \frac{(it)^3}{3!} \nabla^3_H A \, + \dots
\end{equation}
Thus we can express conjugation by an operator $U=e^{-iHt}$ in terms of gradients $\nabla_H$.
Such superoperator expansions are also familiar from quantum information theory, %
e.g. the Liouvillian representation of unitary quantum channels~\cite{preskill1997lecture}. 
Iterating this formula, the QAOA$_1$ operator $Q=Q(\gamma,\beta)=U_M(\beta)U_P(\gamma)$
acts by conjugation as to map the cost Hamiltonian $C\rightarrow Q^\dagger C Q$ 
to a linear combination of cost gradient operators. We show several exact series expansions for QAOA$_1$ with arbitrary cost function (Hamiltonian).

\begin{lem} \label{lem:QAOA1conj}
Let $A$ be an $n$-qubit matrix. The QAOA operator 
 $ Q= U_M(\beta) U_P(\gamma)$
acts by conjugation on $A$ as %
\begin{equation}  \label{eq:costHamConj}
Q^\dagger A Q 
= \sum_{\ell=0}^\infty \sum_{k=0}^\infty \frac{(i\beta)^k(i\gamma)^\ell}{\ell!\,k!} \, \nabla_C^\ell %
\nabla^k A. 
\end{equation}
\end{lem}
\begin{proof}%
Identifying $Q^\dagger A Q$ as conjugation of $A$ by $U_M(\beta)$, followed by conjugation by~$U_P(\gamma)$,
the result then follows from two applications of \eqref{eq:infinitesimalConj}.
\end{proof}
Observe we may write \eqref{eq:costHamConj} 
succinctly as%
\begin{equation}  \label{eq:costHamConj2}
Q^\dagger A Q 
=  e^{i\gamma \nabla_C} e^{i\beta \nabla} A,
\end{equation}
where %
 the order-of-operations of the superoperators 
$e^{i\gamma \nabla_C}$ and $e^{i\beta \nabla}$ is implicitly understood. 
Hence all terms in $Q^\dagger C Q$ are of the form $\nabla_C^\ell \nabla^k C$, i.e., containing at most one alternation of $\nabla$ and $\nabla_C$.\footnote{The %
Baker–Campbell–Hausdorff formula~\cite{woit2017quantum} allows one to similarly express the effective Hamiltonian $H=H(\gamma,\beta)$ of each QAOA layer $Q=U_M(\beta)U_P(\gamma)=exp(-iH)$ as a series in commutators of $B$ and $C$, 
$$H= \beta B + \gamma C -i\tfrac{\gamma\beta}2\nabla C
-\tfrac{\gamma\beta^2}{12}\nabla^2 C +\tfrac{\gamma^2\beta}{12}\nabla_C\nabla C
+i\tfrac{\gamma^2\beta^2}{24}\nabla_C\nabla^2C+\dots,$$
but with %
more complicated %
expressions for higher-order terms 
as compared to \eqref{eq:costHamConj}.} 
In~\secref{sec:QAOAp} we apply this formula recursively to obtain analogous formulas for QAOA$_p$ %
consisting of terms containing at most $2p-1$ such alternations.

\begin{rem} \label{rem:normInvariance}
\eqsref{eq:costHamConj}{eq:costHamConj2} are invariant under the rescalings %
$(\gamma,C)\rightarrow (\gamma/a,aC)$ or 
$(\beta,B)\rightarrow (\beta/b,bB)$ for any $a,b>0$. Hence to apply these formulas perturbatively one must consider the quantities $\|\gamma C\|=|\gamma|\|C\|$ and $\|\beta B\|=|\beta|n$ as small parameters, not just $|\gamma|,|\beta|$. We consider such a small-angle regime in~\secref{sec:smallAngleQAOA1}. Note that rescaling $C$ without proportionately rescaling $c(x)$ preserves the solution structure but would violate the condition  \eqref{eq:costHam}; %
we assume %
$c(x)$ uniquely determines~$C$ as in~\eqref{eq:costHam} throughout.
\end{rem}

We apply \lemref{lem:QAOA1conj} to derive exact %
expressions for the cost expectation and probabilities of QAOA$_1$ using $\langle C\rangle_1 = \langle Q^\dagger C Q\rangle_0$ and $P_1(x)= \langle Q^\dagger H_x Q\rangle_0$ for $H_x:=\ket{x}\bra{x}$ from~\remref{rem:probProj}.

\begin{theorem}  \label{thm:QAOA1}
 Let $C$ be a %
 cost Hamiltonian. For QAOA$_1$ the cost expectation is 
\begin{eqnarray*}
\langle C \rangle_1 
&=& \langle C \rangle_0 
\,-\,\gamma \beta \langle \nabla_C \nabla C \rangle_0
\,\,+\, \sum_{ \ell + k \geq 4,\ell + k \text{ even}  }^\infty
 \frac{(i\gamma)^\ell(i\beta)^k}{\ell!\,k!} 
\langle\nabla_C^\ell %
\nabla^k C %
\rangle_0, \nonumber
\end{eqnarray*}
where $\langle \nabla_C \nabla C \rangle_0=\tfrac{2}{2^n}\sum_x c(x)dc(x)\leq 0$, 
and the probability of measuring $x\in\{0,1\}^n$ is 
\begin{eqnarray*} \label{eq:P1x}
P_1(x) 
&=& P_0(x)%
+\sum_{%
\ell + k \text{ even}  }^\infty
 \frac{(i\gamma)^\ell(i\beta)^k}{\ell!\,k!} 
\langle\nabla_C^\ell %
\nabla^k H_x %
\rangle_0
\end{eqnarray*}
where $H_x:=\ket{x}\bra{x}$, and both sums are over subsets of integers $\ell,k>0$. 
\end{theorem}
\begin{proof}
Plugging $A=C$ into \eqref{eq:costHamConj}, using $[C,C]=0$, and taking expectation values with respect to $\ket{s}$ gives the QAOA$_1$ cost expectation value %
\begin{equation}  \label{eq:expeccostHamConj0}
\langle C\rangle_1 \,=\, \bra{s}Q^\dagger C Q\ket{s}\,=\,\langle C\rangle_0 \, + \,
\sum_{\ell=0}^\infty
\sum_{k=1}^\infty \frac{(i\gamma)^\ell(i\beta)^k}{\ell!\,k!} 
\langle \nabla_C^\ell %
\nabla^k C \rangle_0.
\end{equation} 
From \lemsref{lem:expecGrad}{lem:expecCDC},
the lowest order nonzero contribution $\langle C \rangle$ is the second-order term $-\gamma \beta \langle \nabla_C \nabla C\rangle_0=
\tfrac{-2}{2^n}\gamma\beta\sum_x c(x)dc(x)$, %
and the remaining terms %
with $\ell+k$ odd are identically zero from \lemref{lem:genGrads}, which gives the first result. 

The probability result follows similarly. 
Applying \eqref{eq:costHamConj} to $H_y:=\ket{y}\bra{y}$ gives 
\begin{equation}
    P_1(y)=%
\langle Q^\dagger H_y Q\rangle_0=\sum_{\ell=0}^\infty
\sum_{k=0}^\infty \frac{(i\gamma)^\ell(i\beta)^k}{\ell!\,k!} 
\langle \nabla_C^\ell %
\nabla^k H_y \rangle_0,
\end{equation} 
to which using $[C,H_y]=0$ and again applying \lemref{lem:genGrads} gives the stated result. 
\end{proof}

We generalize this result to QAOA$_p$ in \thmref{thm:smallprecursed} 
in \secref{sec:QAOAp}. 
We now use the leading order terms  of \thmref{thm:QAOA1} to %
characterize 
QAOA$_1$ in small-parameter regimes.

\subsection{Leading-order QAOA$_1$}
\label{sec:smallAngleQAOA1}

This section provides the proofs of the results stated in  \secref{sec:smallAngleQAOA1initial}. 

\begin{proof}[Proof of \thmref{thm1:smallAngles}]
 The results follow applying \lemref{lem:expecCDC} to the leading order terms of \thmref{thm:QAOA1}, as $\langle\nabla_C\nabla C\rangle_0=\tfrac{2}{2^n}\sum_x c(x)dc(x)\leq 0$ and $\langle \nabla_C \nabla H_x \rangle_0 = \tfrac{2}{2^n} \sum_y c(y)d\delta_x(y)= \tfrac{2}{2^n} dc(x)$ for $H_x=\ket{x}\bra{x}$ of \remref{rem:probProj}. 
When  $c(x)$ is %
nonconstant then its Hamiltonian representation $C=a_0 + \sum_j a_j Z_j + \sum_{i<j}a_{ij} Z_i Z_j+\dots$  has at least one nonzero Pauli coefficient $a_\alpha$, $\alpha\neq0$, and so %
\lemref{lem:expecCDC} implies $\langle\nabla_C\nabla C\rangle_0 <  0$. 
\end{proof}
The proof of \thmref{thm:allanglessmall} is similar and given in \secref{sec:smallAngleQAOApClass}. 
Comparing to our results %
of \thmssref{thm1:smallAngles}{thm:allanglessmall}{thm:QAOA1}
we see that %
by selecting $\tau H=\sqrt{2}\beta B + \sqrt{2}\gamma C$ for a simple quantum quench, the lowest-order contribution to $\langle C\rangle$ 
is the same as for for QAOA. 
\begin{proof}[Proof of \propref{prop:quench}]
Observe that %
$\tau \nabla_H = \nabla_{\tau H}$ and hence 
$\tau \nabla_H C  =\sqrt{2}\beta\nabla C$
and $\tau^2 \nabla^2_H C = 2\gamma\beta \nabla_C\nabla C +2 \beta^2\nabla^2 C$.
Applying \eqref{eq:infinitesimalConj} to conjugation of $C$ by %
$e^{-iH\tau}$ gives
\begin{eqnarray}
    e^{i\tau H}Ce^{-i\tau H} &=& C + i\tau\nabla_H C - \frac{\tau^2}2 \nabla^2_H C \,\pm \dots\\
    &=& C +  i\sqrt{2}\beta\nabla C - \gamma\beta \nabla_C \nabla C - \beta^2 \nabla^2 C + \dots\nonumber 
\end{eqnarray}
Taking initial state expectation values of both sides and applying \lemref{lem:expecCDC} 
gives~\eqref{eq:quenchLeadingOrder}. The statement for QAOA$_p$ follows similarly.  
\end{proof}
 
\subsubsection{Small-angle error bound}
Here we derive rigorous error bounds %
to the leading-order approximations of \thmref{thm1:smallAngles}. %

Here, the spectral norm $\|C\|:=\|C\|_2$ is the maximal eigenvalue in absolute value. %
For a maximization problem with $c(x)\geq 0$, e.g., a constraint satisfaction problem, 
the norm %
gives the optimal cost $\|C\|%
=\max_x c(x)$. Here we write $C=a_0I+C_Z$ as in \eqref{eq:costHamZs}, with $C_Z$ traceless.

\begin{cor}%
\label{cor:thm1errorboundeps}
Let $\widetilde{P_1}(x)$, $\widetilde{\langle C\rangle_1}$ denote the %
probability and cost expectation 
estimates (\eqsref{eq:px1}{eq:expecC1}) of \thmref{thm1:smallAngles} for
QAOA$_1$ with a $k$-local cost Hamiltonian $C=a_0I + C_Z$, 
and 
let $0<\e<1$.  
\begin{itemize}
    \item If $|\gamma|\leq \frac{\e^{1/4}}{2\min(\|C_Z\|,\|C\|)}$ and $|\beta|\leq\frac{\sqrt \e}{2 k}$ 
then
\begin{equation} \label{eq:smallAnglesErrorBound}
    \left|\langle C \rangle_1 - \widetilde{\langle C\rangle}_1 \right| < %
   \min(\|C_Z\|,\|C\|)\;\e.
\end{equation}
\item If additionally or instead %
$|\beta|<\frac25\frac{\sqrt \e}{n}$ then for each $x\in\{0,1\}^n$ we have 
\begin{equation} \label{eq:smallAnglesErrorBoundProb}
    \left| P_1(x) - \widetilde{P_1}(x) \right| < \frac{\e}{2^{n}}.
\end{equation}
\end{itemize}
\end{cor}

The proof is shown in \appref{app:normAndErrorBounds} through a sequence of general lemmas. 
For example, for MaxCut we have $\|C_Z\|=m/2\leq \max_x c(x)=\|C\|$. Typically $\|C_Z\|< \|C\|$, 
but this is not true for arbitrary $C$, which is the reason for the minimum function of  \eqref{eq:smallAnglesErrorBound}.
Futhermore, the difference in the ranges for~$\beta$ between the two cases of \corref{cor:thm1errorboundeps} %
arises because measurement probabilities $P_1(x)$ correspond to $n$-local observables. 

As for optimization applications we typically are interested in the expected approximation ratio %
$\langle C \rangle/ c_{opt} \geq \langle C \rangle/ \|C\|$ (for maximization) achieved by a QAOA circuit, 
here we have selected the ranges of $\gamma,\beta$ for \eqref{eq:smallAnglesErrorBound} such that the resulting error bound contains a factor $\|C\|$ (or $\|C_Z\|$). 
Hence $O(\|C\|\e)$ error in the estimate for $\langle C\rangle_1$ corresponds to an $O(\e)$ error estimate for 
the expected approximation ratio. 
A general bound for the error of the cost expectation estimate of \thmref{thm1:smallAngles} as a polynomial in $\gamma,\beta$ is given in \lemref{lem:thm1errorbound}, which may be used to alternatively select the ranges of $\gamma,\beta$ and resulting error bound. The %
proof %
of \lemref{lem:thm1errorbound} is shown using a recursive formula that may be extended to obtain similar bounds when including higher-order terms beyond the approximations of \thmref{thm1:smallAngles}, as indicated in \thmref{thm:QAOA1}.

\subsubsection{Small-angle QAOA$_1$ behaves classically}
\label{sec:smallQAOA1class}
The above results imply that 
QAOA$_1$ can be classically emulated in the small-angle regime, 
in the following sense: we show  
a simple classical randomized algorithm for emulating QAOA$_1$ %
with sufficiently small $|\gamma|,|\beta|$ 
that reproduces the leading-order
probabilities $\widetilde{P_1}(x)$ of sampling %
each bitstring~$x$ from \thmref{thm1:smallAngles}. 
For $|\gamma|,|\beta|$ polynomially small in the problem parameters such that the conditions of \corref{cor:thm1errorboundeps} are satisfied this implies the same error bounds \eqsref{eq:smallAnglesErrorBound}{eq:smallAnglesErrorBoundProb} apply for the classical algorithm. %

Specifically, assume the cost function 
$c(x)$ %
can be efficiently evaluated classically. %
Let $K$ be a bound such that $|\partial_jc (x)|\leq K$ for all $j,x$; for example, for MaxCut, we may set $K$ to be $2$ times the maximum vertex degree in the graph. For an arbitrary cost Hamiltonian we may always take $K=2\|C\|$ (or, $K=2\|C_Z\|$ for $C=a_0I+C_Z$). Consider the following simple %
classical randomized algorithm. 
\vskip 0.5pc
\textbf{%
Algorithm 1: sample according to leading-order QAOA$_1$}
\begin{enumerate}
    \item %
   \underline{Input}: Parameters $\gamma,\beta\in [-1,1]$ such that $|\gamma\beta|\leq \frac1{2nK}$.
    \item %
    Select a bitstring $x_0\in\{0,1\}^n$ and an index $j\in [n]$ each uniformly at random.
    \item Compute $\partial_j c(x_0)$ and let $\delta (x_0,j) :=  2n\gamma\beta \partial_j c(x_0)$. 
    \item  Using a weighted coin, 
    flip the $j$th bit of $x_0$ with probability given by $\frac12 + \frac12 \delta (x_0,j)$.
    \item Return the resulting bitstring $x_1$. %
\end{enumerate}
 By the choice of algorithm parameters in Step 1 we have $-1\leq \delta (x,j)\leq 1$ which %
 ensures the probability distribution resulting from Step 4 remains normalized. For sufficiently small $\gamma,\beta$ this algorithm emulates QAOA$_1$ up to small error. 

\begin{theorem}%
\label{thm:smallAnglesClassical}
Consider a cost Hamiltonian $C=a_0I+C_Z$. 
For fixed $0<\e<1$, let %
$|\beta|\leq \frac25\frac{\sqrt \e}{n}$ and $|\gamma|\leq \e^{1/4}/(2\min(\|C_Z\|,\|C\|)$. 
Then there exists an efficient classical randomized algorithm producing bitstring $x$ with probability $P_{class}(x)$ which satisfies  
\begin{equation} \label{eq:thm2a}
|P_{class}(x)-P_1(x)| \leq %
\frac{\epsilon}{2^{n}}.
\end{equation}
and expected value of the cost function 
   $| \langle c \rangle_{class}  - \langle C \rangle_1 | <  \|C_Z\|\e$.
\end{theorem}
\noindent In this sense  %
we say that QAOA$_1$ behaves classically in the small-angle regime.

\begin{proof}%
Consider %
The conditions of the theorem imply $|\gamma\beta|\leq 1/(9n\min(\|C_Z\|,\|C\|))$ and hence $\gamma,\beta$ satisfy the first condition of Algorithm~$1$.
The probability of this procedure returning a particular %
bitstring $x$ is %
\begin{eqnarray*}
Pr(x_1=x) &=&\, Pr(x_0 = x)\sum_j Pr(j)Pr(coin=0) \,+\,\sum_j  Pr(x_0 = x^{(j)})Pr(j) Pr(coin=1)\\
&=& \frac1{2^n}\frac1{n}\sum_j(\frac12 - \frac12\delta (x,j)) \,+\,  \sum_j \frac1{n}\frac1{2^n}(\frac12 + \frac12\delta (x^{(j)},j))\\
&=& \frac1{2^n}( \frac1{n}\frac{n}{2} + \frac1{n}\frac{n}{2}) - \frac{1}{2^n}\frac1{n}\sum_j\frac12\delta (x,j)) \,+\, \frac1{2^n}\frac1{n}\sum_j \frac12\delta (x^{(j)},j)\\
&=& \frac1{2^n} -\frac1{2^n} \frac1{n}\sum_j\delta (x,j)\\
&=&Pr(x_0 = x) \,-\, \frac2{2^n} \gamma\beta \,dc(x)
\end{eqnarray*}
where we have used %
$\delta  (x^{(j)},j)=-\delta (x,j)$.
Observe that from $\sum_x dc(x) = 0$ %
the resulting distribution %
remains %
normalized.
Hence %
the expected value of the cost function $c(x)$ is
$$\langle c \rangle = \sum_x Pr(x_1=x)c(x) = \langle c(x) \rangle_0 -2\gamma \beta \langle c(x) dc(x)\rangle_0.$$
and so we see Algorithm 1 reproduces %
the leading-order %
probability and cost expectation 
estimates of \thmref{thm1:smallAngles}. 
The error bounds %
then follow directly from \corref{cor:thm1errorboundeps}.
\end{proof}

\thmref{thm:smallAnglesClassical} demonstrates that %
QAOA parameters must necessarily be not-too-small for potential quantum advantage. In practice, the parameter search space may be pruned in advance to avoid such regions. %

\subsection{Causal cones and locality considerations for QAOA$_1$}
\label{sec:lightcones}
Here we consider problem and Hamiltonian locality for QAOA$_1$ circuits, building off of %
\secref{sec:probLocality}.
 We extend these results to QAOA$_p$ in \secref{sec:lightconesp}. 
Our framework formalizes here similar observations made in the original QAOA paper \cite{farhi2014quantum} in the context of computing the cost expectation $\langle C\rangle_p$ and resulting bound to the approximation ratio, as well as subsequent work computing such quantities analytically or numerically for specific problems~\cite{wang2018quantum,hadfield2018thesis}, and approaches to computing or approximating quantum circuit observable expectation values more generally~\cite{evenbly2009algorithms,bravyi2019classical}. 

In particular we show how using locality
in our framework allows for a more direct method of dealing with the action of the QAOA phase operator. This is especially advantageous in cases where each variable appears in relatively few clauses (e.g., MaxCut on bounded degree graphs, as considered for QAOA in \cite{Farhi2014,wang2018quantum}). %
Such cases allow for straightforward %
evaluation in our framework by providing %
succinct efficiently computable expressions of 
relevant quantities such as $\bra{\gamma}\nabla C\ket{\gamma}$ and $\bra{\gamma}\nabla^2 C\ket{\gamma}$ %
which appear
for instance in both the leading-order and exact analysis of QAOA$_1$ of  \secsref{sec:smallMixingAngle}{sec:MaxCutQUBO}.

Consider a %
cost Hamiltonian $C=a_0I+\sum_{j=1}^m C_j$ such that each $C_j$ is a single Pauli term. As observed in \cite{Farhi2014}, 
in cases where the \textit{problem locality} is %
such that for the resulting cost Hamiltonian: i) each $C_j$ acts on at most $k$ qubits %
and ii) each qubit appears in at most $D$ terms, 
this information can be taken advantage of in computing each $\langle C_j\rangle$, especially when $k,D\ll n$, by a priori discarding operators that can be shown to not contribute.%
\footnote{On the other hand, while causal cone considerations are useful for computing quantities such as $\langle C\rangle=\sum_j\langle C_j\rangle$ in the case, e.g., of bounded-degree CSPs, for evaluating QAOA probabilities they may be less useful as $H_x=\ket{x}\bra{x}$ is $n$-local. This observation motivates the general result of \thmref{thm:QAOA1}.} 
This reduced number %
of necessary qubits and operators for computing expectation values is often called the (reverse) \textit{causal cone} or \textit{lightcone} in the literature  \cite{evenbly2009algorithms,bravyi2018quantum,bravyi2019classical,shehab2019noise,streif2020training}, and relates to Lieb-Robinson bounds concerning the speed of propagation of information~\cite{hastings2008observations}. See in particular \cite[Sec. IV.B]{childs2019theory} for applications to simulating local observables. We use the nomenclature \textit{QAOA}$_p$-\textit{cones}, which we define below and in~\tabref{tab:locality}.

\begin{table}[ht]
\centering
\begin{tabular}{ |c|c|c||c|} 
 \hline
 Label & Symb. & Definition & Example: Value for MaxCut\\ 
  \hline
    Hamiltonian term &$C_j$ & %
    $a_jZ_{{j_1}}\dots Z_{j_{|N_j|}}$ & $C_{uv}=-\frac12Z_uZ_v$, \,$(uv)\in E$\\ 
     Classical term & $c_j(x)$ & $a_j(-1)^{x_{j_1}\oplus x_{j_2} \oplus \dots \oplus x_{j_{|N_j|}}}$& $c_{uv}=-\tfrac12(-1)^{x_u\oplus x_v}$\\ 
   Qubits in $C_j$ ($c_j$) & $N_j$  & $\{i:Z_i\in C_j\}\subset [n]$  & $N_{uv}=\{u,v\}$ \\
   QAOA$_1$-cone of $C_j$ & $L_j$ & $ \cup_{i: N_i\cap N_j \neq \emptyset} N_i$ & %
   $\{u,v\}\cup \{w: (uw)\in E \vee (vw)\in E\}$\\
   QAOA$_p$-cone of $C_j$ & $L_{j,p}$ & $L_{j,p-1}\cup_{i: N_i\cap L_{j,p-1} \neq \emptyset} N_i$ & $L_{uv,p}=\{\ell:\, $dist$(\ell,N_{uv})\leq p\}$\\
 \hline
\end{tabular}
\caption{Notation concerning locality for cost Hamiltonian $C=a_0I + \sum_{j}C_j$, %
with corresponding decomposition of the classical cost function $c(x)=a_0+\sum_jc_j(x)$ such that $C_j\ket{x}=c_j(x)\ket{x}$.  
The QAOA$_1$-cone of $C_j$ %
corresponds
to the cost function neighborhood (with respect to the variables) of each $c_j$. %
For MaxCut, dist$(\ell,N_{uv})$ indicates the smaller of the edge distance in the graph of vertex $\ell$ from $u$, or $v$,  %
and so the QAOA$_1$-cone of $C_{uv}$ is vertices (qubits) in the graph %
adjacent to $u$ or $v$.} 
\label{tab:locality}
\end{table}

Recall from \secref{sec:costHam} 
we may 
decompose $C$ with respect to $C^{\{i\}}$ the (sum of) terms in $C$ containing a $Z_i$ factor, $i=1,\dots,n$, which from \eqref{eq:partialjC} satisfy $\partial_i C =-2 C^{\{i\}}$ (i.e., each $C^{\{i\}}$ is diagonal and represents the function $-\tfrac12 \partial_ic(x)$).
Consider the operator $\widetilde{B}:=U_P^\dagger B U_P$, which will appear in our analysis to follow.  
Using $B=\sum_{i=1}^n X_i$ and 
$XZ=-ZX$
we  have 
\begin{equation} \label{eq:Bconj}
   \widetilde{B} = U_P^\dagger B U_P= \sum_{i=1}^n e^{2i\gamma C^{\{i\}}}X_i
   =\sum_{i=1}^n  e^{-i\gamma \partial_i C} X_i
   =\sum_{i=1}^n  X_i e^{i\gamma \partial_i C},
\end{equation}
which implies $\widetilde{B}\ket{s}=\sum_{i=1}^n  e^{-i\gamma \partial_i C}\ket{s}=\tfrac1{\sqrt{2^n}}\sum_x(\sum_{i=1}^n e^{-i\gamma \partial_i c(x)})\ket{x}$. 
As each Hamiltonian $\partial_iC$ only acts on qubits adjacent to $i$ with respect to $C=\sum_{j=1}^mC_j$, \eqref{eq:Bconj} reflects the QAOA$_1$-cone structure of the particular problem.
Hence for each $C_j$, $j=1,\dots,m$, acting on qubits $N_j\subset[n]$, %
define the superset 
\begin{equation}
    L_{j}=N_j \cup \{\ell: \exists j'\neq j \; \text{ s.t. } \ell \in N_j \cap N_{j'}\},
\end{equation}
which we refer to as the \textit{QAOA}$_1$-\textit{cone} of $C_j$. In \secref{sec:lightconesp} we generalize this to $L_{j,p}$ for QAOA$_p$ with $L_{j,1}=L_j$ and $L_{j,0}
:=N_j$, 
as summarized in \tabref{tab:locality}. 

Turning to cost gradient operators, 
\eqref{eq:Bconj} implies the phase operator acts on $\nabla C$ as 
\begin{eqnarray} \label{eq:costGradConj}
e^{i\gamma C}\,\nabla C\,e^{-i\gamma C}= [\widetilde{B},C] \,= \nabla_{\widetilde{B}} C = \sum_j  \nabla_{\widetilde{B}} |_{L_j}C_j.
\end{eqnarray}
In particular terms in $C$ that do not act entirely within $L_j$ always commute through and cancel in each $e^{i\gamma C}\,\nabla C_j\,e^{-i\gamma C}$. 
Hence using $X\ket{s}=\ket{s}$ with \eqsref{eq:nablaCjs}{eq:Bconj} %
gives
\begin{eqnarray} \label{eq:expecgammanablaC}
\bra{\gamma}i\nabla C \ket{\gamma} &=& %
i\bra{s}\sum_{i=1}^n  (e^{i\gamma \partial_i C} - e^{-i\gamma \partial_i C})C \ket{s}%
=-\frac{2}{2^n} \sum_x c(x) \sum_{i=1}^n \sin(\gamma \partial_i c(x))%
\end{eqnarray}
where for each 
$C_j = a_j Z_{j_1}\dots Z_{j_{|N_j|}}$ %
from \eqref{eq:costGradConj} we have  
\begin{eqnarray} \label{eq:expecgammanablaCj}
\bra{\gamma}i\nabla C_j \ket{\gamma} &=&  
-\frac{2}{2^{|N_j|}} \sum_{x\in\{0,1\}^{L_j}} c_j(x) \sum_{i\in N_j} \sin(\gamma \partial_i c_j(x)),
\end{eqnarray}
and so reduces to a sum over at most $|L_j|$ variables, with $ \bra{\gamma}i\nabla C\ket{\gamma} =\sum_j\bra{\gamma}i\nabla C_j \ket{\gamma}$. %

Observe that $|L_j|$ gives an upper bound to the necessary number of variables to sum over, i.e., we can compute \eqref{eq:expecgammanablaCj} by summing over $2^{|L_j|}$ bitstrings. Alternatively as $|L_j|$ becomes large we may approximate quantities such as \eqsref{eq:expecgammanablaC}{eq:expecgammanablaCj} with, for example, Monte Carlo estimates. %

For $\nabla^2C$, a similar calculation, shown in \appref{app:smallMix}, yields 
\begin{eqnarray} \label{eq:expecgammanabla2C}
\bra{\gamma}\nabla^2 C \ket{\gamma} =
\frac{-8}{2^n} \sum_x c(x) \sum_{i<i'} \sin(\tfrac{\gamma}2 (\partial_i c(x)+\tfrac12\partial_i\partial_{i'}c(x)))
\sin(\tfrac{\gamma}2 (\partial_{i'} c(x)+\tfrac12\partial_i\partial_{i'}c(x))),
\end{eqnarray}
which for each $C_j = a Z_{j_1}\dots Z_{j_{N_j}}$, $C=\sum_j C_j$, reduces to 
\begin{eqnarray*}
\bra{\gamma}\nabla^2 C_j \ket{\gamma} =
\frac{-8}{2^{|L_j|}} \sum_{x\in\{0,1\}^{|L_j|}} c_j(x) \sum_{i<i'} \sin(\tfrac{\gamma}2 (\partial_i c(x)+\tfrac12\partial_i\partial_{i'}c(x)))
\sin(\tfrac{\gamma}2 (\partial_{i'} c(x)+\tfrac12\partial_i\partial_{i'}c(x))).
\end{eqnarray*}

We apply these results in \secref{sec:smallMixingAngle}.

\subsection{QAOA$_1$ with small mixing angle and arbitrary phase angle}
\label{sec:smallMixingAngle}

In this section we generalize the leading-order QAOA$_1$ results of  \secref{sec:smallAngleQAOA1} to the case of arbitrary phase angle $\gamma$, but $|\beta|$ taken as a small parameter. 
We show for %
this regime %
that QAOA$_1$ probabilities and expectation values remain easily expressible in terms of local cost differences, for arbitrary cost functions. %
Higher order terms %
may be similarly derived with our framework.
(We show %
complementary results for the setting of small $|\gamma|$ but arbitrary $\beta$ in \secref{sec:smallphase} for the case of QUBO cost functions.)

\begin{theorem}
\label{thm:smallBetap1}
The probability of measuring %
each string $x$ for QAOA$_1$
to first order in $\beta$ is
\begin{equation}
    P_1(x) \simeq P_0(x) \,-\, \frac{2\beta}{2^n} \sum_{j=1}^n  \sin({\gamma\, \partial_j c(x)}), %
\end{equation}
and the first-order expectation value
is then 
$\,\, %
    \langle C \rangle_1  \,%
   \,\simeq\, \langle C \rangle_0 
    \,-\, \frac{2\beta}{2^n} \sum_x c(x)
    \left(\sum_{j=1}^n  \sin({\gamma\, \partial_j c(x)})\right).%
$ %
\end{theorem}

Consider a string $x^*$ that maximizes $c(x)$. Then $\partial_j c(x^*) \leq 0$ for all $j$. Therefore, assuming $\gamma>0$ is small enough that each product $|\gamma \partial_j c(x^*)| < \pi $, then we see that the probability of such a state will increase, to lowest order in $\beta>0$.  Thus we again see in this regime that to lowest order probability will flow as to increase $\langle C \rangle$. 
Similar arguments apply to %
minimization. %

\begin{proof}
Expressing $\langle C \rangle_1$ as
the expectation value of
$U_M(\beta)^\dagger C U_M(\beta)$ as given in \eqref{eq:conj1}, taken with respect to $\ket{\gamma}:=U_P(\gamma)\ket{s}$,   gives 
to %
low order in $\beta$  %
\begin{eqnarray} \label{eq:expecCsmallbeta}
\langle C \rangle_1 
&=& \langle C \rangle_0 
+ \beta \bra{\gamma} i\nabla C \ket{\gamma}  
- \tfrac\beta{2} \bra{\gamma} \nabla^2 C \ket{\gamma}  
+ \dots%
\end{eqnarray}
The leading order contribution then follows from \eqref{eq:expecgammanablaC}. 

which plugging into \eqref{eq:expecCsmallbeta} gives the result for $\langle C \rangle_1$. 
 Similarly, repeating this derivation for the QAOA$_1$ expectation value of $H_x=\ket{x}\bra{x}$ shows the first-order %
correction to the %
probability of measuring each string $x$ %
is
$- \frac{2\beta}{2^n} \sum_{j=1}^n  \sin({\gamma \partial_j C(x)})$  %
which gives the %
result for $P_1(x)$. 
\end{proof}

Applying the small argument approximation $\sin(x)\simeq x$ to \thmref{thm:smallBetap1} reproduces the results of \thmref{thm1:smallAngles}. Furthermore, extending \thmref{thm:smallBetap1} to include the second-order contribution in $\beta$ to $\langle C\rangle_1$ follows readily from \eqref{eq:expecgammanabla2C} and \eqref{eq:expecCsmallbeta}.
From analysis similar to the case when both angles are small, and a simple modification to Algorithm 1 above 
it follows that QAOA$_1$ is classically emulatable %
for sufficiently small $|\beta|$ (up to small additive error), independent of the size of $|\gamma|$. 

\begin{cor}%
\label{cor:smallbeta}
There exists a constant $b$ such that 
for QAOA$_1$ with %
$|\beta|\leq b/n$
there is a simple classical randomized algorithm 
such that for each $x$ %
$$ |P_{class}(x)-P_1(x)|= O(1/2^n).$$
\end{cor}

\begin{proof}
The classical algorithm is constructed by adjusting the quantities of Algorithm $1$ to match those of \thmref{thm:smallBetap1}, which yields identical leading order terms for both algorithms. It remains to bound the error (cf. the results and proofs of \appref{app:errorBounds}). 

Let $\ket{\gamma } := U_P(\gamma)\ket{s}$. 
Applying \eqref{eq:infinitesimalConj} to %
$H_x:=\ket{x}\bra{x}$ for $P_1(x)=\bra{\gamma\beta}H_x \ket{\gamma\beta}$ gives
$$P_1(x)=
P_0(x) + \beta \bra{\gamma } i\nabla H_x \ket{\gamma } + \sum_{j=2}^\infty \frac{(i\beta)^j}{j!} \bra{\gamma } \nabla^j H_x \ket{\gamma }$$
Observe that using $|\bra{x}\ket{\gamma}|=\frac1{\sqrt{2^n}}$ we have  
$ |\bra{\gamma}\nabla^\ell H_x\ket{\gamma}|\leq (2n)^\ell |\bra{\gamma} H_x\ket{\gamma}| = (2n)^\ell \frac1{2^n}$,
so we bound the tail sum as 
\begin{eqnarray*}
\left| \sum_{j=2}^\infty \frac{(i\beta)^j}{j !} \bra{\gamma}\nabla^\ell H_x \ket{\gamma}\right|
&\leq& %
\sum_{j=2}^\infty
\frac{(2n\beta)^j}{j!} \frac1{2^n}%
\leq \frac1{2^n}(e^{2n\beta} -2n\beta -1)%
\leq \frac2{2^n}n^2\beta^2 e^{2n\beta}.
\end{eqnarray*}
Thus if $|\beta|=O(1/n)$ this quantity is $O(\frac1{2^n})$ as desired.
\end{proof}

Generally, for QAOA with larger parameter values, the leading-order %
approximations become less accurate as higher-order %
terms become significant. For large enough parameters, %
contributing terms will involve cost function differences over increasingly large neighborhoods, and 
the number of contributing terms 
may become %
super-polynomial. %
Hence the direct mapping to classical randomized algorithms fails to generalize to arbitrary angles as the probability updates will no longer be efficiently computable in general. 
Indeed, the results of~\cite{farhi2016quantum} imply that, %
under commonly believed computational complexity conjectures, there cannot exist an efficient classical algorithms emulating QAOA$_p$ with arbitrary angles in general, even for QAOA$_1$; %
see \remref{rem:samplingComplexity} below.  
  
In the remainder of \secref{sec:QAOA1} 
we illustrate the application of our framework by considering several examples including the Hamming ramp toy problem, a simple quench protocol related to QAOA, and QAOA$_1$ for MaxCut and QUBO problems.

\subsection{Example: QAOA$_1$ for the Hamming ramp} \label{sec:HammingRamp}
Here we consider %
the Hamming weight ramp problem (studied e.g. for QAOA in~\cite{bapat2018bang}). We apply our framework to determine the leading order terms of the QAOA$_1$ cost expectation, %
and show it matches the exact solution. 

Consider the Hamming-weight ramp cost function $c(x)=\alpha|x|$ with $\alpha\in\reals\setminus\{0\}$ which we may write 
\begin{equation}
    c(x)=\alpha|x|=\frac{\alpha n}2 +( \alpha |x|-\frac{ \alpha n}2) \, =:\frac{\alpha n}2 + c_z(x),
\end{equation}
with Hamiltonian
\begin{equation}
    C=\frac{\alpha n}2 I  - \frac{\alpha}2 \sum_j Z_j \, =:\frac{\alpha n}2I + C_Z,
\end{equation}
i.e., $C_Z := - \frac{\alpha}2 \sum_j Z_j$ represents the function $c_z(x)=\alpha |x|-\frac{ \alpha n}2$. %
\paragraph{Exact results:}
For QAOA$_1$ a simple calculation shows 
\begin{eqnarray} \label{eq:rampOptimal}
\langle C \rangle_1 %
= \frac{\alpha n}{2} + \frac{\alpha n}{2} \sin(\alpha \gamma) \sin(2\beta).
\end{eqnarray}
Thus QAOA$_1$ optimally solves this problem, with probability $1$, with angles $\gamma=\frac{\pi}{2\alpha},\beta=\pi/4$.
Expanding \eqref{eq:rampOptimal} %
using $\sin(x)=x-\frac{x^3}{3!}+\dots$ gives  
\begin{eqnarray} \label{eq:rampSmallAngle}
    \langle C\rangle_1 %
    &=&  \frac{\alpha n}2 + \alpha^2 n\gamma\beta  - \frac23 \alpha^2 n \gamma \beta^3
 - \frac16 \alpha^4 n \gamma^3 \beta
 +\dots
\end{eqnarray}
where the terms not show are order $6$ or higher in $\alpha$, and ${\gamma,\beta}$ combined. 

\begin{figure}[ht]
    \centering
    \begin{subfigure}[t]{0.5\textwidth}
        \centering
        \includegraphics[height=1.6in]{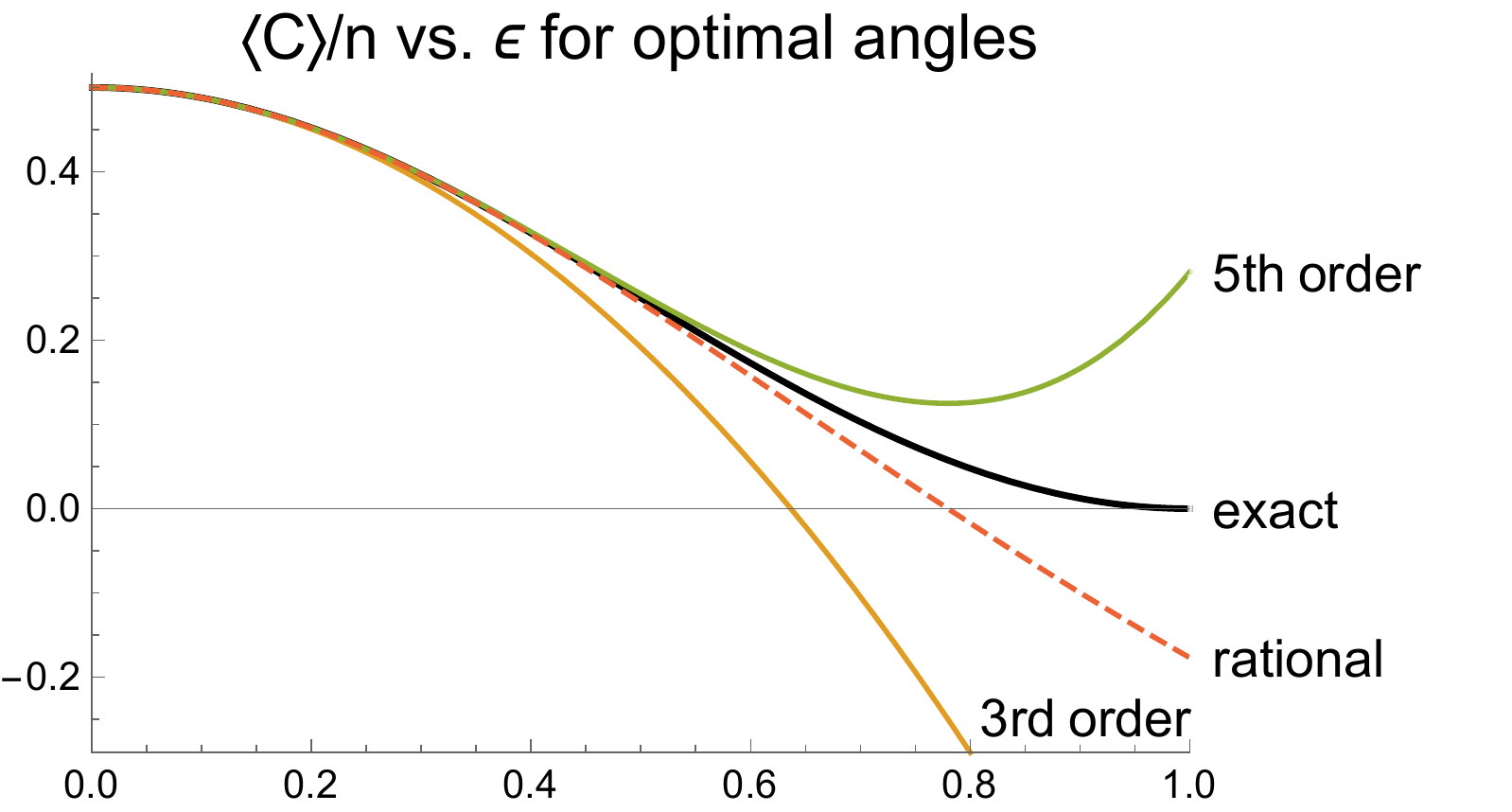}
        \caption{$\langle C\rangle_1$ vs $\epsilon$ with $\beta = -\epsilon (\pi/4)$ and $\gamma = \epsilon(\pi/2)$ linearly interpolating from zero to the angles giving the optimal solution with probability 1.}
    \end{subfigure}%
    \begin{subfigure}[t]{0.5\textwidth}
        \centering
        \includegraphics[height=1.6in]{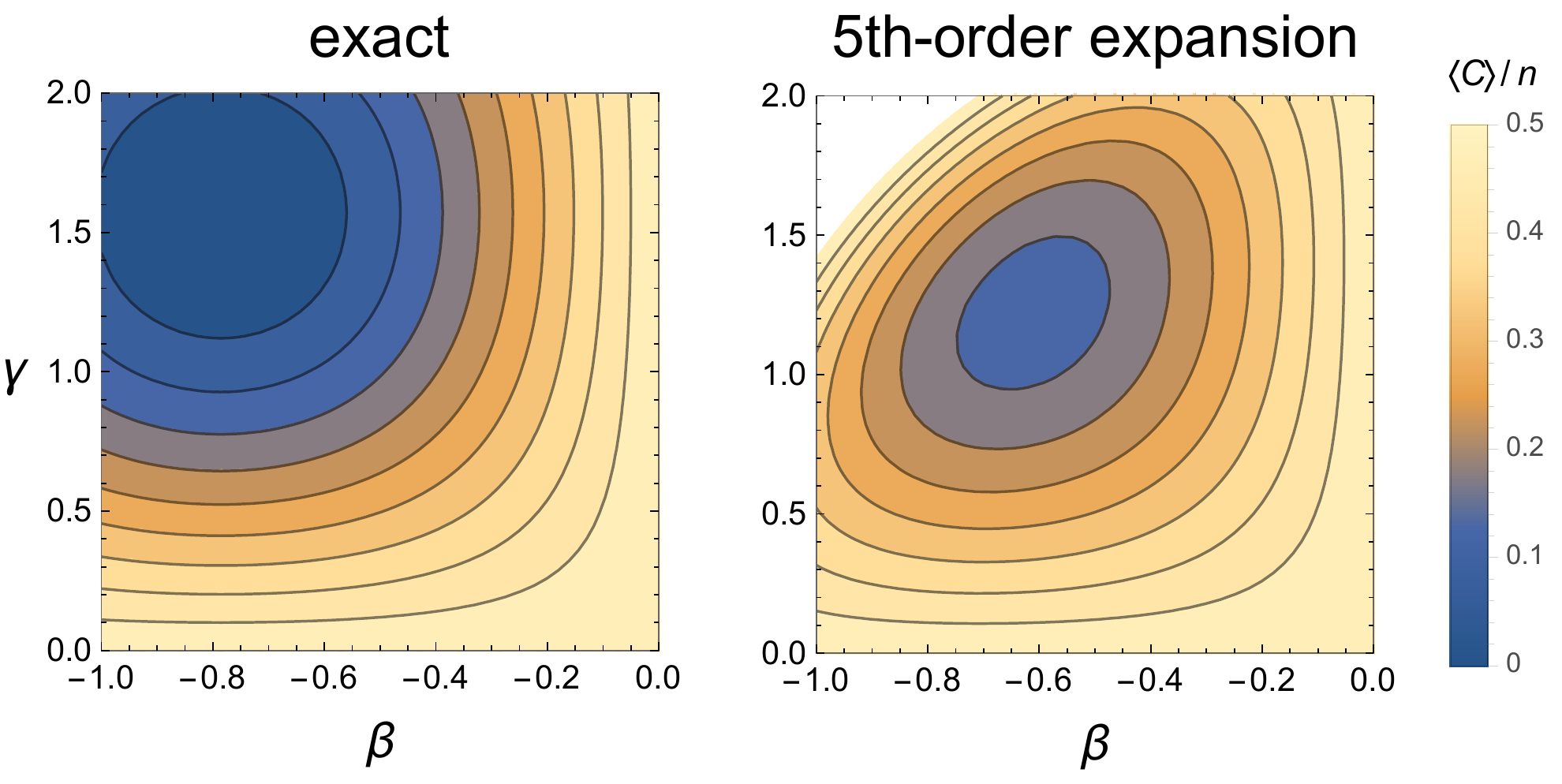}
        \caption{$\langle C\rangle_1$ vs $\beta$ and $\gamma$}
    \end{subfigure} 
    \caption{Comparison of approximate versus exact approximations of $\langle C\rangle_1$ for QAOA$_1$ for the Hamming weight ramp %
    problem $c(x)=|x|$. %
    Here we consider the minimization version optimized by the all $0$ string. Panel (a) shows the deviation between the exact formula \eqref{eq:rampOptimal} and the approximations obtained by restricting \eqref{eq:rampSmallAngle} to 3rd order and 5th order terms, respectively, plotted versus increasing parameter size. The figure also shows a 2,3-Pade approximation, which is a rational function whose coefficients match the expansion of that function match the 5th order truncated Taylor series~\cite{orszag1978advanced}. 
    The right panel compares contour plots of $\langle C\rangle_1$ for the indicated angle values. We observe significant overlap between the targeted blue regions of both plots.}
     \label{fig:ramp}
\end{figure}

\paragraph{Our framework:}
\thmsref{thm1:smallAngles}{thm:smallprecursed} produce the same leading-order expression \eqref{eq:rampSmallAngle} for the Hamming ramp. For this cost function it is easily seen %
    $d^\ell c(x) = (-2)^\ell c_z (x)$, 
and in particular $d c(x) = -2 c_z (x)=\alpha (n -2|x|)$ and $d^2 c(x) = 4 c_z (x)$. 
Thus we have
$$ \langle \nabla_C \nabla C \rangle_0 = \frac2{2^n}\sum_x c(x) dc(x) =-4\langle C\; C_Z \rangle_0= -4 \sum_j (\frac{\alpha^2}4)= -\alpha^2 n$$
which from \thmref{thm:smallprecursed} gives the correct leading order contribution $\alpha^2 n \gamma \beta$. Next we have 
\begin{eqnarray*}
\langle \nabla_C \nabla^3 C \rangle_0 = %
\frac2{2^n}\sum_x c(x) d^3c(x) 
= \frac8{2^n}\sum_x c(x) dc(x)
=-4\alpha^2 n
\end{eqnarray*}
which gives the $\gamma \beta^3$ term %
coefficient 
$-\frac23\alpha^2n$. 
Next, from \lemref{lem:quboGrads} we have $\nabla^2 C=4C_Z$, %
it follows %
$\langle \nabla^2_C \nabla^2 C \rangle_0 = 0$ 
as anticipated from \eqref{eq:rampSmallAngle}. 
Finally, for the term $\langle \nabla^3_C \nabla C \rangle_0$, %
we have $\nabla_C \nabla C = \alpha^2 B$ %
and so 
\begin{eqnarray*}
\langle \nabla^3_C \nabla C \rangle_0 &=& \langle \nabla^2_C \nabla_C \nabla C \rangle_0
=-\alpha^2 \langle  \nabla_C \nabla C \rangle = -\alpha^4 \langle B\rangle_0 = -\alpha^4 n
\end{eqnarray*}
Using these expressions in \thmref{thm:smallprecursed} (which generalizes \thmref{thm1:smallAngles}) for QAOA$_1$ then gives the correct terms up to 4th order. 
\figref{fig:ramp} compares these approximations with the exact behavior.

In~\cite{bapat2018bang}, variants of the ramp such as the Bush-of-implications and Ramp-with-spike are studied for QAOA, along with cost functions given by perturbations of the ramp. Our techniques similarly apply to these problems. More importantly, \thmsref{thm1:smallAngles}{thm:smallprecursed} may be applied to problems where the exact results are not generally obtainable.

\subsection{Example: QAOA$_1$ for random Max-3-SAT} \label{sec:3SAT}

When applying QAOA to instances drawn from a class of problems, it may be convenient to select a single set of parameters rather than optimizing parameters individually for each instance. For instance, this selection could be QAOA parameters that optimize the average cost for the class of problems. Our framework can address this situation by averaging the cost operators over the class of problems with a fixed choice of parameters.

As an example, %
here we apply this idea to random Max-$3$-SAT. %
A SAT problem  in conjunctive normal form consists of $m$ disjunctive clauses
in $n$ Boolean literals (i.e., variables or their negations). %
For $k$-SAT, each clause involves %
$k$~distinct literals.  
The cost function associated with a Max-$k$-SAT  instance
is the number of clauses an assignment $x$ satisfies. 
For Max-3-SAT, the corresponding cost Hamiltonian as in \eqref{eq:costHamZs} and used in  \lemref{lem:expecCDC} is easily obtained \cite{hadfield2018representation}:
each %
clause contributes $7/8$ to $a_0$, and $\pm 1/8$ to each corresponding $a_i$, $a_{i,j}$ and $a_{i,j,k}$, %
with the signs depending on which variables %
are negated within each clause.

Random $k$-SAT instances are constructed by randomly selecting $k$ distinct variables for each clause and negating each variable with probability one-half. %
Averaging the squares of the cost coefficients that appear in \lemref{lem:expecCDC} over random 3-SAT instances gives
\begin{equation}
   \overline{ \langle \nabla_C \nabla C \rangle_0 } = -\frac{3m}{4}
\end{equation}
where here the overline denotes the average over random %
instances.
Applying this in \thmref{thm:QAOA1} gives the averaged leading-order change in QAOA$_1$ expected cost %
for such random Max-3-SAT instances to be %
\begin{eqnarray}
\overline{\langle C\rangle_1} \,=\, \langle C\rangle_0 + \frac34 m \gamma \beta + \dots \,=\, \frac78 m+ \frac34 m \gamma \beta + \dots
\end{eqnarray}
where $\langle C\rangle_0 = \tfrac78 m$ corresponds to the expected solution value obtained from uniform random guessing. Similar considerations apply to the higher-order terms of \thmref{thm:QAOA1}.

An important consideration when using averages over a class of problems is how well they represent the characteristics of typical instances. The variance with respect to instances in the class is one measure of this. As an example, consider random Max-3-SAT with a  fixed clause to variable ratio $m = \alpha n$ as $n$ increases. This scaling provides a high concentration of hard instances when $\alpha$ is close to the transition between mostly satisfiable and mostly unsatisfiable instances~\cite{kirkpatrick94}. In this case, the variance scales as
\begin{equation}
   \mbox{var}( \langle \nabla_C \nabla C \rangle_0  ) \sim \frac{9 n \alpha^2}{128}.
\end{equation}
Thus, the relative deviation in $\langle \nabla_C \nabla C \rangle_0$ is proportional to $1/\sqrt{n}$.
Such concentration indicates that the average leading-order change %
 $\langle C\rangle_1$ is representative of the behavior of individual instances. Thus averaging over the class of problems gives a simple representative expression for how QAOA performs for this class of problems when the angles are small. Similar considerations apply to higher-order terms.

For individual 3-SAT instances, the QAOA$_1$ cost expectation value can be 
efficiently evaluated, for %
example, by extending the prescription of \secref{sec:MaxCutQUBO} (cf. \secref{sec:classicalAlgGeneral2}).
However, for other problem classes, or higher-order terms in the expansion, or QAOA$_p$ more generally, evaluating %
such quantities 
can be challenging. In such cases, averaging over a class of problems can be a simpler proxy for the behavior of typical instances than per-instance evaluation, and, for example, applied toward obtaining good algorithm parameters.

\subsection{%
Analysis of QAOA$_1$ for QUBO problems and MaxCut}\label{sec:MaxCutQUBO}

Our framework is useful for \textit{exact} %
analysis of QAOA with \emph{arbitrary} angles, not just to leading-order contributions or in small-angle regimes. Such analysis is %
challenging in general and typically requires some degree of problem specialization. %
Here we show how our framework may be used to succinctly reproduce analytic performance results previously obtained for MaxCut~\cite{wang2018quantum}, and extend such analysis to the wider class of QUBO problems. For these problems we have the following. 

\begin{lem} \label{lem:quboGrads}
For a QUBO cost Hamiltonian $C = a_0 I + \sum_{j}a_{j}Z_j + \sum_{i<j}a_{ij}Z_iZ_j =: a_0 I + C_{(1)} + C_{(2)}$ we have for $k\in\integers_+$
\begin{eqnarray}
    \nabla^{2k}C = 4^k C_{(1)} \, + \, 16^{k-1}\nabla^2 C_{(2)},\\
    \nabla^{2k+1}C = 4^k \nabla C_{(1)} \, + \,16^{k} \nabla C_{(2)}.
\end{eqnarray}
\end{lem}

We prove the lemma in \appref{app:klocal}. 
Similar %
results %
follow for higher degree cost Hamiltonians (towards analysis of QAOA for problems beyond QUBOs). For example, %
for a strictly $3$-local cost Hamiltonian $C=\sum_{i<j<k}a_{ijk}Z_iZ_jZ_k$ it can be shown that any $\nabla^\ell C$ for $\ell\geq 4$ can be expressed as a real linear combination of $C$, $\nabla C$, $\nabla^2 C$, and $\nabla^3 C$.  

We use \lemref{lem:quboGrads} %
to sum the series $U_M^\dagger CU_M=
C+i\beta\nabla C -\tfrac12\beta^2\nabla^2C+\dots$. 
Consider an instance of MaxCut, i.e., a graph $G=(V,E)$ with $|V|=n$ and $|E|=m$, 
with cost %
Hamiltonian
$ C=\frac{m}{2}I-\frac12\sum_{(ij)\in E}Z_iZ_j$, and the other related operators shown in \tabref{tab:summary}. 
From \lemref{lem:quboGrads},
for a strictly quadratic cost Hamiltonian $C_Z=\sum_{uv}c_{uv}Z_uZ_v$ we have %
$\nabla^{2k+1} C_Z = 16^k \nabla C_Z$ and 
 $\nabla^{2k} C_Z = 16^{k-1} \nabla^2 C_Z$, where $\nabla C_Z$ and $\nabla^2 C_Z$ are easily derived as Pauli operator expansions given in \tabref{tab:summary}. 
Hence, %
it follows $U_M^\dagger(\beta) C_ZU_M(\beta)=C_Z-\tfrac14\sin(4\beta)i\nabla C_Z-\tfrac18 \sin^2(2\beta)\nabla^2 C_Z$, and so 
applying %
\thmref{thm:QAOA1} 
for QAOA$_1$ %
yields %
\begin{eqnarray}   \label{eq:expecC1maxcut}
      \langle C\rangle_1 &=& 
\langle C\rangle_0 \,\,+\,%
\frac{\sin(4\beta)}{4}\, \bra{\gamma}i\nabla C \ket{\gamma}
     -\frac{\sin^2(2\beta)}{8} \, \bra{\gamma} \nabla^2 C\ket{\gamma},
\end{eqnarray}
for MaxCut, where $\ket{\gamma}=U_P(\gamma)\ket{s}$ and the angles $\gamma,\beta$ are arbitrary. (We %
give a more detailed derivation of \eqref{eq:expecC1maxcut} 
in \appref{sec:QUBOs}.) 
The right-hand side expectation values depend on the structure of the graph $G$, as reflected in \eqsref{eq:expecgammanablaC}{eq:expecgammanabla2C}, and may be further reduced. %
Indeed, 
we previously computed the quantities of \eqref{eq:expecC1maxcut} for the case of MaxCut  
in \cite{wang2018quantum,hadfield2018thesis} in terms of the proble graph parameters; comparing to %
\cite[Thm. 1]{wang2018quantum}, we have %
\begin{equation} \label{eq:maxcutnablaCexpec}
    \bra{\gamma}i\nabla C \ket{\gamma} 
    =  \sin(\gamma) \sum_{u\in V} d_u\, \cos^{d_u-1}(\gamma),
\end{equation}
where $d_u=\deg(u)$ is the graph degree of vertex $u$, 
and
\begin{equation}\label{eq:maxcutnabla2Cexpec}
    \bra{\gamma}\nabla^2 C \ket{\gamma} =
   2\sum_{(uv)\in E} \cos^{d_u+d_v-2f_{uv}-2}(\gamma)\, (1-\cos^{f_{uv}}(2\gamma)),
\end{equation}
where $f_{uv}$ gives the number of triangles ($3$-cycles) containing the edge $(uv)$ in $G$. 
In particular, if $G$ is triangle free then
$\bra{\gamma}\nabla^2 C \ket{\gamma}=0$. Hence we may always efficiently compute $\langle C \rangle_1$ for a given instance of MaxCut. Moreover, %
\eqssref{eq:expecC1maxcut}{eq:maxcutnablaCexpec}{eq:maxcutnabla2Cexpec} %
lead to QAOA$_1$ performance bounds (i.e., the parameter-optimized expected approximation ratio achieved) for MaxCut for particular classes of graphs; see~\cite{wang2018quantum,hadfield2018thesis}.

Likewise, for general QUBO cost Hamiltonians, applying \lemref{lem:quboGrads}  %
yields
\begin{eqnarray} \label{eq:QAOA1QUBO}
\langle C \rangle_1 = \langle C \rangle_0
+\frac{\sin (2\beta)}{2} \bra{\gamma}i\nabla C_{(1)} \ket{\gamma} %
+ \frac{\sin(4\beta)}{4}   \bra{\gamma}i\nabla C_{(2)} \ket{\gamma} - \frac{\sin^2(2\beta)}{8}   \bra{\gamma}\nabla^2 C_{(2)} \ket{\gamma} .%
\end{eqnarray}
The right-hand side  %
expectation values are independent of $\beta$ and 
may each be similarly computed as in the case of MaxCut as 
polynomials of~$\gamma$ %
that reflect the cost Hamiltonian coefficients~$a_\alpha$ and underlying adjacency graph, see e.g. \cite[App. E]{hadfield2018thesis}. %
Alternatively, they may be computed using \eqsref{eq:expecgammanablaC}{eq:expecgammanabla2C}. %
Practical computation of these quantities typically takes advantage of linearity of expectation $\langle C\rangle_p = \sum_j \langle C_j\rangle_p$ for cost Hamiltonians $C=\sum C_j$ where each $C_j$ is a single weighted Pauli term, as well as problem-locality considerations, as discussed in \secsref{sec:probLocality}{sec:lightcones}. 
Similar formulas
for more general problems %
may be obtained with our framework. As an example we  consider a variant of Max-2-SAT.

\subsubsection{Example: Analysis of QAOA$_1$ for Balanced-Max-$2$-SAT} \label{sec:balancedMax2Sat}
Here we apply our framework to QAOA$_1$ for the Balanced-Max-$2$-SAT problem, defined as instances of Max-$2$-SAT such that each variable appears negated or unnegated in an equal number of clauses. 
Assuming the Unique Games Conjecture in computational complexity theory, there is a sharp threshold $\theta\simeq 0.943$
such that it is NP-hard to hard to approximate this problem to $\theta$ or better whereas a $\theta-\epsilon$ approximation can be achieved any polynomial time \cite{khot2007optimal}[Thms. 3 and 4]; see also \cite{austrin2007balanced}.\footnote{A related variant of Max-$2$-SAT has been previously explored for quantum annealing~\cite{santra2014max}. In that case, random instances were constructed using equal probabilities of a variable appearing as negated or unnegated in each clause; cf. also \secref{sec:3SAT}.} 

Here for simplicity we consider problem instance with the additional assumption that each pair of variables $x_i,x_j$ appears in at most one clause (i.e., one of $x_i \vee x_j$,  $\overline{x}_i \vee x_j$, $x_i \vee \overline{x}_j$, or $\overline{x}_i \vee \overline{x}_j$).\footnote{We may relax this assumption at the expense of a more complicated proof and presentation of \eqref{eq:expecC1balmax2sat} due to bookkeeping required, %
in which case $G$ generalizes to a multigraph.} We use $(-1)^{i\oplus j}$ to denote the parity of a given clause, which is $-1$ when only one of $x_j$, $x_j$ is negated.
Then 
for a Balanced-Max-$2$-SAT instance with $m$ clauses and $n$ variables, the cost Hamiltonian takes a particularly simple form 
\begin{equation} \label{eq:costHamBalancedMaxSat}
    C = \frac34 I - \frac14 \sum_{(ij)\in E}(-1)^{i\oplus j} Z_iZ_j,
\end{equation}
where %
we have identified the graph $G=([n],E)$ with $|E|=m$ edges  corresponding to the problem clauses. %

For computing $\langle C\rangle_1=\tfrac34{m}+\sum_{(ij)\in E}\langle C_{ij}\rangle_1$, the QAOA$_1$-cone of each $C_{ij}$ consists of terms corresponding to edges adjacent to $(ij)$ in $E$. A neighbor $k$ of both $i$ and $j$ defines a triangle in $G$, with parity defined to be the product of its edge parities. 
Let $F^{\pm}$ denote the number of triangles in $G$ with parity $\pm1$.%
\footnote{A balanced instance need not have balanced triangles. 
E.g, the single-triangle instance 
$x_1\vee \overline{x}_2+x_2\vee \overline{x}_3+x_3\vee \overline{x}_1$ has $(f^+,f^-)=(0,1)$ whereas the two-triangle instance 
$x_1\vee x_2+x_2\vee \overline{x}_3+x_3\vee \overline{x}_1 +  \overline{x}_2 \vee  \overline{x}_4 +  \overline{x}_2  \vee \overline{x}_5 + x_4\vee x_5$
has $(f^+,f^-)=(2,0)$.} 
Similarly, 
for each %
edge %
define $f^+=f^+_{ij},f^-=f^-_{ij}$ with respect to the triangles containing $(ij)$. %

Applying our framework and results above \eqrefp{eq:QAOA1QUBO}, in \appref{sec:QUBOs} we show the exact QAOA$_1$ cost expectation is given as a function of the angles and problem instance by 
\begin{eqnarray} \label{eq:expecC1balmax2sat}
\langle C\rangle_1&=&\langle C\rangle_0 \,+\ \frac{\sin4\beta\sin(\gamma/2)}8\sum_{(ij)\in E}(\cos^{d_i}(\gamma/2)+\cos^{d_j}(\gamma/2))\nonumber\\
&-&\frac{\sin^22\beta}{4}\sum_{(ij)\in E} \cos^{d+e-2f_{ij}^+-2f_{ij}^-}(\gamma/2)\,g(f_{ij}^+,f_{ij}^-; \gamma,\beta), 
\end{eqnarray}
where $\langle C\rangle_0=\frac34m$, $d_i+1$ is the degree of vertex $i$ in $G$,  %
and we have defined 
$$g(f^+,f^-;\gamma,\beta)=\sum_{\ell=1,3,5,\dots}^{f^++f^-} \cos^{2(f^++f^--\ell)}(\gamma/2)\sin^{2\ell}(\gamma/2) \binom{f^+}{\ell} \prescript{}{2}{\mathbf{F}}_1(-f^-,-\ell;f^+-\ell+1;-1).$$
Here $\prescript{}{2}{\mathbf{F}}_1$ is the Gaussian (ordinary) hypergeometric function; see e.g.~\cite{slater1966generalized}. In particular, the function $\binom{f^+}{\ell} \prescript{}{2}{\mathbf{F}}_1(-f^-,-\ell;f^+-\ell+1;-1)$ occurring in $g(\cdot)$ evaluates to $f^+-f^-$ for $\ell=1$, and to $\binom{f^+}{3}-\binom{f^-}3-\tfrac12 f^+f^-(f^+-f^-)$ for $\ell=3$.  

Hence for a given Balance Max-2-SAT instance $\langle C \rangle_1$ may be efficiently computed. 
In some cases we can obtain further simplified expressions. For example, in the case that %
all triangles in the instance are of the same parity, %
$\langle C \rangle_1$ is easily seen to reduce to the same formula for MaxCut given in \secref{sec:MaxCutQUBO} and previously obtained in~\cite{wang2018quantum,hadfield2018thesis} (up to constant factors and the shift $\gamma\rightarrow\gamma/2$ due to the affine mapping between the respective cost Hamiltonians). Indeed, along the way to proving \eqref{eq:expecC1balmax2sat} in \appref{sec:QUBOs}, we rederive \eqssref{eq:expecC1maxcut}{eq:maxcutnablaCexpec}{eq:maxcutnabla2Cexpec} using our framework.   

Results such as \eqref{eq:expecC1balmax2sat} may be used to bound the expected QAOA approximation ratio %
$\langle R\rangle_1 = \langle C\rangle_1 / c^* \geq  \langle C\rangle_1 / m$, where $c^*$ is the optimal cost value.

\subsection{QAOA$_1$ with small phase angle and arbitrary mixing angle}
\label{sec:smallphase}
We conclude \secref{sec:QAOA1} by considering QAOA$_1$ with small phase angle $\gamma$ and arbitrary mixing angle~$\beta$ (i.e., the converse case of  
\secref{sec:smallMixingAngle}). %
For simplicity, here we consider QUBO problems as in \secref{sec:MaxCutQUBO};   
similar but more complicated formulas follow for more general %
cost Hamiltonians.  
We show how %
the leading order contribution to the cost expectation value reflects cost Hamiltonian structure (Fourier coefficients) of \eqref{eq:costHamZs}. %

\begin{theorem}%
\label{thm:smallgamma}
Consider %
a QUBO cost Hamiltonian  
$C=a_0I+\sum_j a_j Z_j + \sum_{j<k} a_{jk}Z_jZ_k$. %
The QAOA$_1$ expectation value of $C$ to second order in $\gamma$ is
\begin{eqnarray*} \label{eq:smallGammaApprox}
\langle C \rangle_1  %
&\simeq& \langle C\rangle_0 + 2\gamma \left( \sin(2\beta)\sum_j a_j^2  + \sin(4\beta)\sum_{j<k} a_{jk}^2 \right)%
\,-\, 4\gamma^2 \sin^2(2\beta)
\sum_{i<j}a_{ij}\left( a_ia_j +\sum_{k} a_{ik}a_{jk}\right)
\end{eqnarray*}
\end{theorem}
\begin{proof}
Recall $C_{(k)}$ denotes the Hamiltonian defined as the (sum of the) strictly $k$-local terms of~$C=a_0 I + C_{(1)} + C_{(2)}$.
Expanding each expectation value in \eqref{eq:QAOA1QUBO} to second order in~$\gamma$ and %
applying \lemsref{lem:genGrads}{lem:expecGrad}
gives 
\begin{eqnarray*}
\bra{\gamma} i\nabla C_{(j)} \ket{\gamma}&=& %
\cancel{\langle i\nabla C_{(j)}\rangle_0} - \gamma \langle \nabla_C \nabla C_{(j)}\rangle_0 - \frac{\gamma^2}2 \cancel{i\langle \nabla^2_C \nabla C_{(j)}\rangle_0} + \dots,
\end{eqnarray*}
for $j=1,2$, 
where $\langle \nabla_C \nabla C_{(1)}\rangle_0 %
=-4\sum_{j}a_j^2$ and    
$\langle \nabla_C \nabla C_{(2)}\rangle_0 %
=-8\sum_{j<k}a_{jk}^2$ from \lemref{lem:expecCDC}, which together %
give the result in the first parenthesis.  %
For the last term in \eqref{eq:QAOA1QUBO} we have 
\begin{eqnarray*}
\bra{\gamma} \nabla^2 C_{(2)} \ket{\gamma}&=& %
\cancel{\langle \nabla^2 C_{(2)}\rangle_0} - \cancel{\gamma i \langle \nabla_C \nabla^2 C_{(2)}\rangle_0} - \frac{\gamma^2}2\langle \nabla^2_C \nabla^2 C_{(2)}\rangle_0 + \dots,
\end{eqnarray*}
with %
$\langle \nabla_C^2 \nabla^2 C_{(2)}\rangle_0
=64 \sum_{i<j}a_ia_ja_{ij}+64 \sum_{i<j}\sum_{k\in nbd(i,j)} a_{ki}a_{ij}a_{jk}$
which we show in \lemref{lem:higherOrderExpectations} of \appref{app:expecVals}. 
\end{proof}

\begin{rem} \label{rem:samplingComplexity}
The quantities of \thmref{thm:smallgamma} may be alternatively expressed in terms of expectation values of classical functions (as in \thmref{thm1:smallAngles}), 
see \lemref{lem:higherOrderExpectations}. 
However, in contrast to the case of QAOA$_1$ with small mixing angle
considered in \thmref{thm:smallBetap1}, here we do not obtain a simple expression for the the leading order probability of measuring each given~$x$ (i.e., the result of %
 \thmref{thm:smallgamma} 
is not of the form $\langle C\rangle_1 \simeq \sum_x c(x) P'(x)$ for some explicit probability distribution $P'(x)$).
Hence, \thmref{thm:smallgamma} doesn't directly yield a simple classical algorithm approximately emulating QAOA$_1$ in this parameter regime in the same way as from  \thmref{thm1:smallAngles}. 

Indeed, a general efficient classical procedure for sampling from QAOA$_1$ circuits for arbitrary parameters would imply the collapse of the polynomial hierarchy in computational complexity~\cite{farhi2016quantum}, and hence such procedures are believed impossible. Our results are consistent with this viewpoint and indicate that sampling remains hard even as the phase angle becomes small (for constant but sufficiently large mixing angle), but as expected becomes relatively easy when both angles, or only the mixing angle, are sufficiently small.  %
Additional intuition 
may be gained from the sum-of-paths viewpoint of 
\sh{\appref{app:sumOfPaths}}. %
\end{rem}

\section{Application to QAOA$_p$} \label{sec:QAOAp}

The gradient operator framework %
extends to %
QAOA$_p$  
with arbitrary level~$p$ and %
parameters~$\gamma_i,\beta_j$. 
In the Heisenberg picture, 
the $j$th QAOA level  
acts %
by conjugating the observable resulting from the preceding ($(j+1)$th) level, in  iteration for $j=p,p-1,\dots,1$.  
This approach yields  
formulas similar to but more general than~\lemref{lem:QAOA1conj}.

\begin{lem} \label{lem:QAOAconjGen}
The QAOA$_p$ operator $Q=Q_pQ_{p-1}\dots Q_2Q_1$, where
 $ Q_j= U_M(\beta_j) U_P(\gamma_j)$, 
acts by conjugation on the cost Hamiltonian $C$ as
\begin{equation}  \label{eq:costHamConjGen}
Q^\dagger C Q 
= %
\sum_{\ell_{1}=0}^\infty \sum_{k_{1}=0}^\infty \cdots \sum_{\ell_{p-1}=0}^\infty \sum_{k_{p-1}=0}^\infty
\sum_{\ell_p=0}^\infty %
\sum_{k_p=0}^\infty
\frac{(i\gamma_1 \nabla_C)^{\ell_1} 
(i\beta_1\nabla)^{k_1} \dots
(i\gamma_p \nabla_C)^{\ell_p} 
(i\beta_p\nabla)^{k_p} }{\ell_1!k_1!\ell_2!k_2!\dots\ell_p!k_p!} C . 
\end{equation}
\end{lem}

The QAOA$_p$ cost expectation then follows as $\langle C\rangle_p=\langle Q^\dagger C Q \rangle_0$. 

\begin{proof}
Follows recursively applying \eqref{eq:infinitesimalConj} for each $Q_j$ as in the proof of \lemref{lem:QAOA1conj}.
\end{proof}

\subsection{Effect of %
$p$th level of QAOA$_p$} 
\label{sec:smallAngleQAOAp1}
We next consider the change in cost expectation between a level-($p-1$) QAOA circuit (with fixed parameters), %
and the level-$p$ circuit constructed from adding an additional QAOA level.
This case concerns the last layer applied of the QAOA circuit, or, alternatively the %
effect of an additional layer.
The following theorem generalizes \thmref{thm:QAOA1}. 

\begin{theorem} \label{thm:costExpecHeis}
For QAOA$_p$ the cost expectation satisfies
\begin{equation}
\langle C \rangle_p 
= \langle C \rangle_{p-1} 
+ \, \sum_{\ell=0}^\infty \sum_{k=1}^\infty
 \frac{(i\gamma_p)^\ell(i\beta_p)^k}{\ell!\,k!} 
 \,\langle\nabla_C^\ell %
\nabla^k C\rangle_{p-1}.
\end{equation}
\end{theorem}
\noindent Here, with respect to the fixed QAOA$_{p}$ state $\ket{\boldsymbol{\gamma}\boldsymbol{\beta}}_{p}$ state, %
$\langle \cdot \rangle_{p-1}$ indicates the expectation value with respect to the corresponding QAOA$_{p-1}$ state resulting from  application of only the first $p-1$ QAOA stages (or, equivalently, $\langle \cdot \rangle_p$ with $\gamma_{p},\beta_{p}$ set to $0$).

\begin{proof}
The result follows applying \lemref{lem:QAOAconjGen}
to the righthand side %
of 
$\bra{\boldsymbol{\gamma}\boldsymbol{\beta}}_{p}C \ket{\boldsymbol{\gamma}\boldsymbol{\beta}}_{p}=
\bra{\boldsymbol{\gamma}\boldsymbol{\beta}}_{p-1}\left(Q_p^\dagger%
CQ_p%
\right) \ket{\boldsymbol{\gamma}\boldsymbol{\beta}}_{p-1}, %
$ 
where $Q_p:=U_M(\beta_p)U_P(\gamma_p)$.
\end{proof}

\subsubsection{Example:
QAOA$_p$ with small phase and mixing angles at level $p$}
When the parameters of the final $p$th level of the QAOA$_p$ are small, we obtain a generalization of \thmref{thm1:smallAngles}.

\begin{cor}%
[Small angles at level $p$]\label{cor:smallAnglesp}
For QAOA$_{p}$ with $p$th level angles $\gamma_p,\beta_p$, to second order %
we have %
\begin{equation} \label{eq:thmexpecCpsmall}
    \langle C \rangle_{p} - \langle C \rangle_{p-1} \simeq    \beta_p \langle i\nabla C \rangle_{p-1} - \frac{\beta_p^2}{2} \langle \nabla^2 C \rangle_{p-1} - \gamma_p \beta_p \langle \nabla_C \nabla C\rangle_{p-1}, %
\end{equation}
where the neglected terms %
cubic or higher in $\gamma_p,\beta_p$.
\end{cor}
\begin{proof}
The result follows collecting the quadratic leading-order terms %
of \thmref{thm:costExpecHeis}.
\end{proof}

\subsubsection{Example:
QAOA$_p$ with small mixing angle at level $p$} 
\label{sec:smallMixingAnglep}
We also consider the case of QAOA$_p$ with small %
mixing angle $\beta_p$, but arbitrary phase angle $\gamma_p$. 
This case applies, for instance, in parameter schedules inspired by quantum annealing, where $\beta$ approaches zero at the end of the anneal. The following result generalizes \thmref{thm:smallBetap1}; the proof is similar %
and is given in \appref{app:smallMix}. 

\begin{cor} %
\label{cor:smallBeta}
Consider a QAOA$_{p-1}$ state $\ket{\boldsymbol{\gamma \beta}}_{p-1} =\sum_{x}q_{x} \ket{x}$ to which another ($p$th) level of QAOA with angles $\beta:=\beta_p,\gamma:=\gamma_p$ is applied. 
The change in probability %
to first order in $\beta$ is 
\begin{eqnarray}\label{eq:probsmallbetap}
P_{p}(x)-P_{p-1}(x) =  -2\beta \sum_{j=1}^n  r_{xj}  \sin(\alpha_{xj}+\gamma \partial_j c(x)) +\dots %
\end{eqnarray}
for each $x\in\{0,1\}^n$, 
where we have defined the real polar variables $r,\alpha$ as 
$r_{xj}=|q^\dagger_{x^{(j)}}q_{x}|$ and 
$ q^\dagger_{x^{(j)}}q_{x} = r_{xj}e^{-i\alpha_{xj}}$ 
which reflect the degree of which the coefficients $q_{x}$ are non-real. 
The expectation value then changes as 
\begin{equation} \label{eq:expecsmallbetap}
    \langle C \rangle_{p} -
   \langle C \rangle_{p-1} = 
  -2\beta \sum_x c(x) \sum_{j=1}^n  r_{xj}  \sin(\alpha_{xj}+\gamma \partial_j c(x)) +\dots.%
\end{equation}
\end{cor}

\subsection{Leading-order  QAOA$_p$}\label{sec:smallAngleQAOAp}
When all QAOA angles are small or viewed as expansion parameters, 
\lemref{lem:QAOAconjGen}
leads to the following generalization of \thmref{thm:QAOA1}.

\begin{theorem} \label{thm:smallprecursed}
For QAOA$_p$, to fifth order in the angles $\gamma_1,\beta_1,\dots, \gamma_p,\beta_p$ we have
\begin{eqnarray} \label{eq:expecCp0}
\langle C \rangle_p \, %
\,=\, %
\, \langle C \rangle_{0}
&-&\langle \nabla_C \nabla C \rangle_0
\sum_{1\leq i\leq j}^p \gamma_i \beta_j%
\nonumber\\%\, +\,
&+& \langle \nabla_C \nabla^3 C \rangle_0 %
\left( \sum_{i\leq j< k<l} \gamma_i \beta_j \beta_k \beta_l 
+ \frac12 \sum_{i\leq j<k}  \gamma_i (\beta^2_j \beta_k+\beta_j \beta^2_k) 
+ \frac1{3!}\sum_{i\leq j} \gamma_i \beta_j^3
  \right)%
  \nonumber\\
 &+&   \langle \nabla_C^3 \nabla C \rangle_0 %
\left( \sum_{i<j< k\leq l} \gamma_i \gamma_j \gamma_k \beta_l 
+ \frac12 \sum_{j<k \leq l}   (\gamma^2_j \gamma_k+\gamma_j \gamma^2_k) \beta_l 
+ \frac1{3!}\sum_{i\leq j} \gamma^3_i \beta_j
  \right)%
  \nonumber\\
&+&\langle \nabla^2_C \nabla^2 C \rangle_0 %
\Bigg( \sum_{i<j\leq k< l} \gamma_i \gamma_j \beta_k \beta_l
+ \frac12 \sum_{i\leq k<l} \gamma^2_i \beta_k \beta_l 
 + \frac12 \sum_{i<j\leq k} \gamma_i \gamma_j \beta_k^2 \nonumber\\
&& \;\;\;\;\;\;\;\;\;\;\;\;\;\;\;\;\;\;\;\;\;\;\;\;  \,\,+\,\, \frac14 \sum_{i\leq j} \gamma^2_i \beta_j^2 
\,+\,\sum_{i\leq j < k \leq l} \gamma_i \beta_j \gamma_k \beta_l \Bigg)%
\;\;\;+\;\;\;\dots 
\end{eqnarray}
where the terms not shown to the right are each  \textbf{order six} or higher in angles and the gradient operator expectations.
\end{theorem}
The gradient operator initial state expectations of %
\eqref{eq:expecCp0} are expressed in terms of cost difference functions in \eqsssref{eq:expecDCis0}{eq:expecDcD3C}{eq:expecDc3DC}{eq:expecDc2D2C}, and \lemref{lem:higherOrderExpectations} of \appref{app:expecVals}. 
These factors depend only on the cost function. Higher order contributions may be derived similarly. 

\sh{
\begin{rem}
Observe that the $k$-th order %
coefficient of %
\eqref{eq:expecCp0} consists of sums over products of angles, with each product of degree $k$. At least one of these sums %
is indexed by $k$ distinct angles, and the others involve $k$ distinct angles or fewer. %
Sums with $k$ distinct angles contain a number of terms that scales as $O(p^k)$ and so these factors dominate when $p$ becomes large. This observation extends to more general initial states or mixing operators (recalling \remref{rem:GenLems}).
\end{rem}
}

\begin{proof}
We compute $\langle C\rangle_p= \langle Q^\dagger C Q\rangle_0$ by 
taking the expectation value of \eqref{eq:costHamConjGen} with respect to $\ket{s}$ and %
collecting the terms proportional to each cost gradient expectation value.  
From %
\lemref{lem:genGrads}, only terms for which %
$\ell_1+k_1+\dots \ell_p + k_p$ is even can give nonzero contributions. Moreover, from \lemref{lem:expecGrad} any expectation value with a leftmost $\nabla$ or a rightmost $\nabla_C$ is zero. Applying these rules %
the only terms which can contribute up to fifth order are %
\begin{eqnarray*}
\langle C \rangle_p \,=\,
\langle C \rangle_{0} &+& a_0 \langle \nabla_C \nabla C \rangle_0 
+ a_1\langle \nabla_C \nabla^3 C \rangle_0 
+ a_2\langle \nabla^2_C \nabla^2 C \rangle_0\\&+& \, a_3\langle \nabla^3_C \nabla C \rangle_0
+ a_4\langle \nabla_C \nabla \nabla_C\nabla C \rangle_0\,+\,\dots 
\end{eqnarray*}
where the $a_i$ are real %
homogeneous polynomials in the angles which correspond to  the possible ways of generating its associated 
gradient operator in the sum \eqref{eq:costHamConjGen}. %
From \lemref{lem:expecCDC} the second order term is $\langle \nabla_C \nabla C \rangle_0 = \frac2{2^n}\sum_x c(x) dc(x)$.
The coefficient $a_0$ follows from each possible selection of the ordered pair $(i\gamma\nabla_C)(i\beta\nabla)C$ in \eqref{eq:costHamConjGen} and is easily seen to be $a_0=\sum_{1\leq i\leq j}^p \gamma_i \beta_j$. 
At fourth order, similarly,  
the coefficients $a_1,\dots,a_4$ are degree $4$ polynomials and easily calculated as sums corresponding to all possible ways of generating the associated gradient, when the order of each (super)operator product in \eqref{eq:costHamConjGen} is accounted for. In particular, $\langle\nabla_C \nabla \nabla_C \nabla C\rangle_0 = \langle \nabla^2_C \nabla^2 C \rangle_0$ from \eqref{eq:costJacobi}, which gives $a_4=\sum_{i\leq j < k \leq l} \gamma_i \beta_j \gamma_k \beta_l$ and corresponds to the final term of the last parenthesis. %
\end{proof}

Repeating the above proof for the QAOA$_p$ probability $P_p(x)=\langle H_x\rangle_p =\langle \ket{x}\bra{x}\rangle_p$ as in the proof of \thmref{thm:QAOA1} produces %
the leading-order expression for~$P_p(x)$ given in \tabref{tab:tab1smallangles}, and similarly for higher order terms. 

\eqref{eq:expecCp0} shows that for any parameters values the cost function expectation value achieved by QAOA$_p$ can be a expressed  as a linear combination of the initial expectation values of $C$ and its gradients. %
These expectation values %
depend only on %
the structure of the cost function. %
The dependence on QAOA parameters %
occurs through the coefficients of each expectation value, which %
are polynomials of the algorithm parameters. %
These polynomials 
are homogeneous but not symmetric. For example, the leading-order coefficient polynomial 
$\sum_{1\leq i\leq j}^p \gamma_i \beta_j$ depends more strongly on early $\gamma$ values and late $\beta$ values, as opposed to late $\gamma$ values or early $\beta$ values.  
Hence \eqref{eq:expecCp0} and its higher order generalizations may help select parameters. Indeed, for a few-parameter schedule, such as e.g., 
$\gamma_j =  \gamma_0 + aj$ and $\beta_j =  \beta_0 + bj$, the coefficients of \eqref{eq:expecCp0} can be computed in terms of the reduced parameter set. %
\sh{In particular, for such linear schedules satisfying $|\gamma_j|,|\beta_j|\leq\Delta $ the leading-order coefficient grows with $p$ as  $|\sum_{1\leq i\leq j}^p \gamma_i \beta_j|=O(\Delta p^2)$.}

\sh{
As an aside we show how the leading-order term of \eqref{eq:expecCp0} %
yields a simple argument showing that that QAOA can %
also obtain the quadratic quantum speedup for Grover's problem, %
reproducing results rigorously obtained in~\cite{jiang2017near}. For convenience consider the case of strings of $n$ variables with a single marked solution $x^*$, and corresponding cost Hamiltonian $C=\ket{x^*}\bra{x^*}$. Following~\cite{jiang2017near}, we apply QAOA with the standard initial state and mixer, and fixed parameters $\gamma_j=\gamma=\pi$ and $\beta_j=\beta=\pi/n$, which gives $\sum_{1\leq i\leq j}^p \gamma_i \beta_j=O(p^2/n)$. In this case the mixing angles become very small as $n$ becomes large, and hence all products of phase and mixing angles also become small. Hence, we can reasonably assume that the leading-order term of \eqref{eq:expecCp0} dominates for large $n$, and neglect the higher-order terms. Since $\langle C \rangle_0 = 1/2^n$ and $\langle \nabla_C \nabla C \rangle_0 = -2n/2^n$ for this problem, to leading order we have 
$\langle C\rangle_p \simeq 1/2^n + O(p^2/2^n)$, where the factors of $n$ have canceled. Rearranging terms then implies that the success probability $P_p(x^*)=\langle C\rangle_p$ is $\Theta(1)$ when $p=\Theta(\sqrt{2^n})$, reproducing Grover's famous speedup for QAOA as obtained in~\cite{jiang2017near}. While this approach does not immediately reveal constant factors or other details, it is particularly useful for obtaining qualitative insights. Beyond Grover's problem We will %
apply %
related ideas using the low-order terms of \thmref{thm:smallprecursed} to obtain further insights %
in future work. 

We conclude this subsection with two useful remarks.}

\begin{rem} %
The %
included 
terms of the leading-order coefficient polynomial 
$\sum_{1\leq i\leq j}^p \gamma_i \beta_j$ of \eqref{eq:expecCp0} are also %
illuminated from the following general formulas. 

Observe that using $U^\dagger_MU_M=I$ we have  %
\begin{eqnarray}
U_M(\beta)U_P(\gamma)=e^{-i\gamma (U_M(\beta)CU^\dagger_M(\beta))}U_M(\beta).
\end{eqnarray}
Hence for QAOA$_p$ with $Q:=U_M(\beta_p)U_P(\gamma_p)\dots U_M(\beta_1)U_P(\gamma_1)$, %
we similarly have 
\begin{eqnarray}
 Q=%
e^{-i\gamma_p (U_M(\beta_p)CU^\dagger_M(\beta_p))}
\, e^{-i\gamma_{p-1} (U_M(\beta_p+\beta_{p-1})CU^\dagger_M(\beta_p+\beta_{p-1}))}
\dots %
e^{-i\gamma_1 (U_M(\overline{\beta})CU^\dagger_M(\overline{\beta}))}U_M(\overline{\beta})
\end{eqnarray}
for $\overline{\beta}:=\sum_{j=1}^p\beta_j$, %
where we have used the property $U_M(\alpha)U_M(\beta)=U_M(\alpha+\beta)$. 
\end{rem}

\begin{rem}
To consider QAOA with arbitrary initial states,  
\lemref{lem:QAOAconjGen} may be similarly applied %
taking expectation values of %
\eqref{eq:costHamConjGen} to obtain results analogous to those of \thmssref{thm1:smallAngles}{thm:allanglessmall}{thm:smallprecursed}. 
\sh{Similar results %
follow for more general QAOA mixing Hamiltonians under mild conditions. In particular, one only requires the initial state to be an eigenstate of the mixing Hamiltonian for the leading-order term of \thmref{thm:smallprecursed} to be provably of the same form (as indicated in \remref{rem:GenLems}).}
\end{rem}

\subsubsection{Leading-order QAOA$_p$ behaves like an effective QAOA$_1$}
\label{sec:smallAngleQAOApClass}

\sh{\lemref{lem:QAOAconjGen} and 
\thmref{thm:smallprecursed}}
give 
perturbative expansions %
in the parameters $\gamma_1,\beta_1,\dots,\gamma_p,\beta_p$, with respect to conjugation by the Identity. These results suggest that when all QAOA angles are small in magnitude, only the relatively few nonzero low-order terms will 
effectively contribute. %
Hence when all parameters are sufficiently small, QAOA$_p$ can again be efficiently classically emulated by a similar argument to the $p=1$ case.

\begin{proof}[Proof of \thmref{thm:allanglessmall}]
The result follows %
from \thmref{thm:smallprecursed} %
and \eqref{eq:expecDCis0}. 
\end{proof}

\begin{rem}
When all angles are bounded such that  $|\gamma_1|,|\beta_1|,\dots,|\beta_p|\leq L$, comparing 
\thmref{thm:allanglessmall}
as $L\rightarrow 0^+$ 
to \thmref{thm1:smallAngles} shows that the cost expectation of QAOA$_p$ approaches that of an equivalent QAOA$_1$ with effective angles $\gamma',\beta'$ satisfying
$$ \gamma'\beta' = \sum_{1\leq i\leq j}^p \gamma_i \beta_j .$$
\sh{Hence as $L$ becomes small the leading-order term of \thmsref{thm:allanglessmall}{thm:smallprecursed} becomes insensitive to the particular bounded parameter schedule the angles may be derived from.}
\end{rem}

Hence, when all angles are sufficiently small we may again classically emulate QAOA$_p$. 
\begin{rem}
Repeating the analysis of \secref{sec:smallQAOA1class} %
we see that for $p=O(1)$ there exists a sufficiently large polynomial  $q(n)$ such that if %
$|\gamma_1|,|\beta_1|,\dots,|\beta_p|=1/q(n)$ then QAOA$_p$ can be efficiently classically sampled from with absolute error of each probability $O(1/2^n)$. 
\end{rem}

\subsubsection{QAOA beats random guessing} \label{sec:randomGuessing}
Here we %
show that QAOA always beats random guessing, in the sense that there always exist polynomially small angles (as opposed to exponentially small) such that $\langle C\rangle_p$ beats $\langle C\rangle_0$ (by a factor at worst polynomially small).
The QAOA$_1$ case implies the same for QAOA$_p$. Whereas it is known QAOA beats random guessing for specific problems (cf. \secref{sec:MaxCutQUBO}), the following result holds generally.

\begin{cor}[QAOA beats random guessing]  \label{cor:qaoa1randomguessing}
Let $c(x)$ be a cost function on $n$ bits that we seek to maximize or minimize, with corresponding Hamiltonian $C$. Assume $c$ is nonconstant and $|c(x)|$ %
is bounded by a polynomial in $n$. 
Then %
there %
exists a polynomial $q(n)>0$ and angles $\gamma_1,\beta_1,\dots$  that achieve $\langle C \rangle_p = \langle C \rangle_0 +\Omega(1/q(n))$ for maximization, or similarly, $\langle C \rangle_p = \langle C \rangle_0-\Omega(1/q(n))$ for minimization. %
\end{cor}
\begin{proof}
Observe as QAOA$_p$ subsumes QAOA$_1$ (i.e., by setting $\gamma_2=\beta_2=\dots=\beta_p=0$), it suffices to show the claim for $p=1$.  
Recall from \lemref{lem:expecCDC} %
that $\langle\nabla_C\nabla C\rangle_0 < 0$ for nonconstant $c(x)$, such that 
the sign of the leading-order contribution in the approximation  $\widetilde{\langle C\rangle}_1$ 
of \thmref{thm1:smallAngles} (cf. \thmref{thm:QAOA1}) is equal to the sign of $\gamma\beta$, and hence the sign of 
$\widetilde{\langle C\rangle}_1-\langle C\rangle_0$ can be controlled by selecting the quadrants of $\gamma,\beta$ appropriately (i.e., for maximization take $\gamma,\beta>0$ or $\gamma,\beta<0$). The result then follows observing that we can select sufficiently small $|\gamma|,|\beta|$ and $\e$ such that the error bound of \corref{cor:thm1errorboundeps} together with the triangular inequality give the result.

Specifically, for simplicity assume we seek to maximize a cost function $c(x)>0$; the general case is similar. 
Let $a:=-\langle \nabla_C\nabla C \rangle_0$ which satisfies $0<a<4n\|C\|^2$ from 
\lemref{lem:expecCDC}, and \lemsref{lem:normGrad}{lem:normCGrad} in \appref{app:normBounds}, and let  $\epsilon=(\frac1{8}\frac{a}{n\|C\|^2})^{4}$, $\gamma=\frac{\epsilon^{1/4}}{2\|C\|}$ and $\beta=\frac{\sqrt{\epsilon}}{2n}$. Here $\e$ is selected to satisfy %
$\|C\|\epsilon=\tfrac12\gamma\beta a$. %
These choices satisfy the conditions of \corref{cor:thm1errorboundeps} to give $|\langle C\rangle_1-\widetilde{\langle C\rangle}_1|<\|C\|\epsilon$. Hence %
$\langle C\rangle_1 > \widetilde{\langle C\rangle}_1 - \|C\|\epsilon$
and so the claim follows as
\begin{eqnarray*}
\langle C\rangle_1 -\langle C\rangle_0 &=&
\langle C\rangle_1 -\widetilde{\langle C\rangle}_1+\widetilde{\langle C\rangle}_1-\langle C\rangle_0 %
\\
&>&- \|C\|\epsilon + \gamma\beta a 
\,=\, \tfrac12\gamma\beta a \\%=\|C\|\epsilon=
&=&\Omega(1/\textrm{poly}(n))
\end{eqnarray*}
where the assumption $\|C\|=O(\textrm{poly}(n))$ implies the %
inverse-polynomial lower bound. %
\end{proof}
\noindent Clearly, the quantities of %
\corref{cor:qaoa1randomguessing} may be tightened significantly when considering specific cost functions, or by considering higher-order terms and more involved error bounds.

\subsection{Causal cones and locality considerations for QAOA$_p$}
\label{sec:lightconesp}

Here we build off of the definitions and discussion 
of \secref{sec:lightcones}.  %
Suppose again %
we are given a cost Hamiltonian $C=a_0I+\sum_jC_j$ such that each Pauli term $C_j=a_jZ_{j_1}\dots Z_{j_{N_j}}$ acts on at most $|N_j|\leq k$ qubits. Assume each qubit appears in at most $D$ terms; the causal cone perspective is especially useful when $D$ is bounded. We let $L_{j,\ell}\subset [n]$ denote the lightcone of $C_j$ corresponding to the $\ell=1,\dots,p$ levels of a 
QAOA$_p$ circuit applied in turn, with $L_{j,\ell-1}\subset L_{j,\ell}$, as defined in~\tabref{tab:locality}.   
Note that as the QAOA levels act in reverse order with respect to conjugation, each $L_{j,\ell}$ corresponds to QAOA angles indexed from $p-\ell+1,\dots,p-1,p$; cf. \eqref{eq:costHamConjGen}. For example, %
$L_{j,1}$ corresponds to %
the single-level QAOA conjugation of $C_j$ with angles $\gamma_p,\beta_p$. %

For computing QAOA$_p$ expectation values such as %
$\langle C\rangle_p = \sum_j \langle C_j\rangle_p$, 
each $\langle C_j\rangle_p$ may again be computed independently and taking advantage of the QAOA-cones of the $C_j$. %
Each $\langle C_j\rangle_p$ can be computed by restricting the initial state and resulting operators to at most $k((D-1)(k-1))^p$ qubits, which for particular values may be significantly fewer qubits than $n$. Moreover, we may apply this idea for each QAOA layer in turn.

\begin{theorem}  \label{thm:lightconep}
Let $C=a_0I+\sum_j C_j$, $C_j=a_j Z_{j_1}\dots Z_{i_{N_j}}$, be a $k$-local 
cost Hamiltonian (i.e., $|N_j|\leq k$), such each qubit appears in at most $D$ of the $C_j$.  
Then for QAOA$_p$ with operator $Q=Q_p\dots Q_2Q_1$, $Q_i=U_M(\beta_i)U_P(\gamma_i)$, we have %
\begin{eqnarray} \label{eq:expecCplightcones}
\langle C\rangle_p &=& \sum_j\, \bra{+}^{\otimes L_{j,p}}
 \,Q^\dagger|_{L_{j,p}}\, C_j\,Q|_{L_{j,p}}\, \ket{+}^{\otimes L_{j,p}}\\
 &=& \sum_j \,\bra{+}^{\otimes L_{j,p}}
 \,Q_1^\dagger|_{L_{j,p}}\,Q_2^\dagger|_{L_{j,p-1}}\dots Q_p^\dagger|_{L_{j,1}}\, C_j\,Q_p|_{L_{j,1}}\dots Q_2|_{L_{j,p-1}}\,Q_1|_{L_{j,p}}\, \ket{+}^{\otimes L_{j,p}},\nonumber
\end{eqnarray}
with $|L_{j,\ell}|\leq \min( k((D-1)(k-1))^\ell,
\, n)$ for $\ell=1,\dots,p$. 
\end{theorem}

The first %
equality of \eqref{eq:expecCplightcones} shows that each $\langle C_j\rangle_p = \langle Q^\dagger C_j Q\rangle_0$ can be equivalently computed by discarding from the QAOA$_p$ circuit %
any qubits or operators outside of $L_{j,p}$. %
The second equality shows that this idea may be applied layerwise, with successive (with respect to conjugation) QAOA layers depending on increasing numbers of qubits.
Both properties are relatively straightforward to take advantage of in numerical computations.

Similar results apply for computing QAOA expectation values of a general %
cost gradient operator as in \eqref{eq:costGradOpGen} (with respect to the resulting QAOA-cones of its Pauli terms). 

\begin{proof}
Consider a fixed observable  $C_j=aZ_{j_1}Z_{j_2}\dots Z_{j_k}$ for a QAOA$_p$ circuit. %
At each $\ell$th layer, operators outside of $L_{j,\ell}$ commute through and cancel with respect to conjugation. Hence 
applying causality considerations to \lemref{lem:QAOAconjGen}, in particular the property that mixing operator conjugations cannot increase Pauli term locality, 
we have  
\begin{eqnarray*}\label{eq:costHamConjGen2}
Q^\dagger C_j Q 
&=& %
\sum_{\ell_{1}=0}^\infty \sum_{k_{1}=0}^\infty \cdots \sum_{\ell_{p-1}=0}^\infty \sum_{k_{p-1}=0}^\infty
\sum_{\ell_p=0}^\infty \sum_{k_p=0}^\infty \frac{(i\gamma_1 \nabla_C)^{\ell_1} 
(i\beta_1\nabla)^{k_1} \dots
(i\gamma_p \nabla_C)^{\ell_p} 
(i\beta_p\nabla)^{k_p} }{\ell_1!k_1!\ell_2!k_2!\dots\ell_p!k_p!} C_j \\
&=& \sum_{\ell_{1}=0}^\infty \dots \sum_{k_p=0}^\infty \frac{(i\gamma_1 \nabla_C|_{L_{j,p}})^{\ell_1} 
(i\beta_1\nabla |_{L_{j,p-1}})^{k_1} \dots
(%
(i\gamma_p \nabla_C |_{L_{j,1}})^{\ell_p} 
(i\beta_p\nabla |_{N_{j}})^{k_p} }{\ell_1!k_1!\ell_2!k_2!\dots\ell_p!k_p!} C_j,
\end{eqnarray*}
where we recall $N_j=L_{j,0}$. %
Taking initial state expectations and summing over $j$ then gives a slighter tighter result than  \eqref{eq:expecCplightcones}. %
Observing that we may increase each $L_{j,\ell}$ without affecting expectation values implies the two equalities of \eqref{eq:expecCplightcones}. Finally, the bound on $|L_{j,\ell}|$ follows similarly observing that each conjugation by $U_P$ can increase the degree of a given Pauli term by a factor of at most $(D-1)(k-1)$, and that %
by definition $|L_{j,\ell}|\leq n$. 
\end{proof}

\subsection{%
Algorithms %
computing or approximating $\langle C\rangle_p$}
\label{sec:classicalAlgGeneral2}

Here we %
show two general algorithms for computing QAOA cost expectation values $\langle C\rangle_p$. The first approach is exact, though its may scale exponentially in $p$; hence the second approach considers a family of approximations for $\langle C \rangle_p$ obtained by truncating the series expression of \thmref{thm:costExpecHeis} at a given order. The first approach encompasses %
much of the existing techniques in the literature.  We emphasize that while such classical procedures may be useful, e.g., to help find good algorithm parameters, for optimization applications in general a quantum computer is required to efficiently obtain the corresponding solution samples.

\vskip 1pc
\textbf{%
Algorithm 2: compute  $\langle C\rangle_p$ for QAOA$_p$}
\begin{enumerate}
    \item \underline{Input}: Parameters $n,p\in\naturals$, angles  $(\gamma_1,\beta_1,\dots,\gamma_p,\beta_p)\in \reals^{2p}$, cost Hamiltonian $C=a_o I+\sum_{j=1}^m C_j$ with $m=$poly$(n)$ and $C_j=a_jZ_{j_1}\dots Z_{j_\ell}$.
    \item For $j=1,\dots,m$ and $\ell=1,\dots,p$ compute lightcones $L_{j,\ell}\subset [n]$ (or upper bounds). 
    \item For $j=1$ %
    compute $Q^\dagger C_j Q$ as a Pauli sum %
    by restricting the operators in each QAOA level to those in $L_{j,\ell}$ as in \thmref{thm:lightconep} and its proof. %
    \item Discard all terms in the sum containing a Y or a Z factor, %
and set $\langle C_j\rangle_p$ as the sum of the remaining coefficients.
    \item Apply Steps 3 and 4 for each %
    $j= 1, . . . , m$.  
    \item Return the overall sum $\langle C\rangle_p = a_0+\sum_j \langle C_j\rangle_p$. 
\end{enumerate}
Algorithm 2 generalizes the proof given in \cite{wang2018quantum,hadfield2018thesis} of $\langle C \rangle_1$ for MaxCut \eqrefp{eq:expecC1maxcut}, and is used, for instance, in \cite{hadfield2018thesis} to show a similar result for for the Directed-MaxCut problem variant. As observed in \cite{Farhi2014}, for $C=\sum_j C_j$ the quantities $\langle C_i \rangle_p$ and $\langle C_j \rangle_p$ are equal whenever the terms $C_i$ and $C_j$ have the same local neighborhood structure with respect to $p$ and the underlying cost function
(more precisely, whenever there exists a permutation of qubits taking $C|_{L_{i,p}}$ to $C|_{L_{j,p}}$). This property significantly reduces the number of unique computations of $\langle C_j\rangle_p$ required in Steps 3 and 4 %
in order to compute $\langle C\rangle_p$, as a recent paper further elaborates~\cite{shaydulin2021exploiting}.
For example, for MaxCut on the complete graph it suffices to compute the expectation value for a single edge. %
Further results concerning symmetry in QAOA are given in \cite{shaydulin2020classical}.

Even for bounded-degree problems, the number of terms involved in computing each $\langle C_j\rangle_p$ in Algorithm 2 typically grows exponentially with $p$ which makes exact QAOA performance results difficult to obtain beyond quite small $p$ in general. 
\thmssref{thm:costExpecHeis}{thm:smallprecursed}{thm:lightconep} may be  similarly used to obtain faster classical algorithms for approximating $\langle C\rangle_p$. %
Given a parameter~$\ell\in\naturals$, Algorithm~3 returns an approximation of $\langle C \rangle_p$ accurate up to order~$\ell$ in the QAOA angles, such that
the accuracy of the returned approximation improves with increased $\ell$ and converges to the exact value returned by Algorithm~2. %
We describe how worst-case error bounds may be obtained in for fixed truncation $\ell$ in \appref{app:errorBounds}, though our approximations often perform much better in practice than such bounds may indicate (cf. \figref{fig:ramp}).
We emphasize that the condition $|\gamma_i|,|\beta_j|<1$ in optimal parameter sets or restricted parameter schedules often appears in the literature. 

\vskip 1pc
\textbf{%
Algorithm 3: approximate  $\langle C\rangle_p$ to order $\ell$ in the QAOA angles}
\begin{enumerate}
    \item \underline{Input}: Parameters $p,\ell\in\naturals$, angles  $(\gamma_1,\beta_1,\dots,\gamma_p,\beta_p)\in \reals^{2p}$, cost Hamiltonian $C$ %
    \item If $\ell$ is odd then $\ell \leftarrow \ell -1$.
    \item If $\ell=0$ or $p=0$ then return $\langle C \rangle_0$.
    \item Let $\mathcal{A}:=\{(a_1,b_1\dots,a_p,b_p)\in \{0,1,\dots,\ell -1\}^{2p}: \sum_j(a_j+b_j)\in \{0,2,\dots,\ell \}\}.$
    \item Let $G_\alpha =\nabla_C^{a_1}\nabla_{b_1}\dots \nabla_C^{a_p}\nabla^{b_p}C$ for each $\alpha \in \mathcal{A}$
    \item For each $\alpha\in \mathcal{A}$, compute $\langle G_\alpha \rangle_0$ using the lemmas of \secref{sec:expecVals}
    and %
    the coefficient polynomials $f_\alpha(\gamma_1,\dots,\beta_p)$ as in the proof of  \thmref{thm:smallprecursed}. 
    \item Return the estimate $\widetilde{\langle C\rangle_p} = \sum_{\alpha \in \mathcal{A}} f_\alpha(\gamma_1,\dots,\beta_p) \langle G_\alpha \rangle_0 $
\end{enumerate}

Clearly, Algorithm $3$ may also take advantage of lightcone considerations as in Algorithm $2$ in computing the quantities $\langle G_\alpha \rangle_0=\sum_j\langle G_{\alpha,j} \rangle_0$ for $C=\sum_j C_j$. 
For implementation, the set $\mathcal{A}$ contains at most $\tfrac12 \ell^{2p}$ elements and so the operators $G_\alpha$ can be enumerated and efficiently computed as Pauli sums when $p=O(1)$ and $\ell=\log n$. If we further restrict to $\ell=O(1)$, then %
we can efficiently compute each $\langle G_\alpha \rangle_0$, and the polynomials $f_\alpha(\gamma_1,\dots,\beta_p)$ %
each have poly$(n)$ terms. %
Hence, %
when $p,\ell=O(1)$ and the number of terms in $C$ is poly(n), the algorithm can always be implemented efficiently. %

Algorithms 2 and 3 may be further combined to yield hybrid approaches to estimating $\langle C\rangle_p$, where some parameters are treated perturbatively and other exactly, though we do not %
enumerate these approaches here. We conclude \secref{sec:QAOAp} by %
considering an alternative %
approach to deriving series expressions for QAOA quantities.

\section{Generalized mixers, initial states, and applications} 
\label{sec:generalizedCalculus} 

We proposed designs for generalized %
QAOA mixing Hamiltonians and 
unitaries in~\cite{hadfield2019quantum,Hadfield17_QApprox} %
beyond the transverse-field Hamiltonian mixer originally proposed in~\cite{Farhi2014}. %
Such constructions are particularly appealing for optimization problems or encodings with hard constraints. %
Here we illustrate generalizations of our framework to quantum alternating operator ans\"atze constructions for Max Independent Set and a Graph Coloring optimization problems. Similar ideas apply more generally such as to the other optimization problems and corresponding constructions of~\cite{hadfield2019quantum}. 
\sh{We conclude the section with an outline of how the framework may be applied to electronic structure problems of quantum chemistry based on the QAOA construction of~\cite{kremenetski2021quantum}.} 
\sh{We provide some additional discussion of}  applications beyond classical optimization in \secref{sec:discussion}.

For the purpose of this section, consider %
a general mixing Hamiltonian $\widetilde{B}=\sum_j \widetilde{B}_j$ acting on $n$ qubits with each $[\widetilde{B}_j,C]\neq 0$ (and possibly  $[\widetilde{B}_j,\widetilde{B}_k]\neq 0$ for some $k\neq j$). %
We consider here problems such that each $x\in\{0,1\}^n$ is either feasible, meaning it encodes a valid candidate problem solution, or else infeasible. We assume that the $\widetilde{B}_j$ each preserve feasibility, i.e., map the subspace of feasible computational basis states to itself. As the cost Hamiltonian~$C$ is diagonal, this immediately implies that the gradients of all orders $\nabla^{a_\ell}_{\widetilde{B}} \dots \nabla_{\widetilde{B}}^{a_3} \nabla^{a_2}_C \nabla^{a_1}_{\widetilde{B}} C$ also preserve the feasible subspace. 
We show through two %
examples the relationship between the resulting cost  gradients and (generalized) classical cost difference functions. In particular, 
for Max-Independent-Set we consider the transverse-field mixer augmented with control, and for Graph Coloring problems we consider a type of XY mixer~\cite{hadfield2019quantum,wang2019xy}. 
We refer the reader to~\cite{hadfield2019quantum} for additional details and variations.

\subsection{Maximum Independent Set} \label{sec:MaxIndepSet}
Consider the NP-hard optimization problem Max-Independent-Set: given a graph~$G=(V,E)$ with $|V|=n$ vertices, we seek to find (the size of) the largest subset of independent vertices. %
Different QAOA variants for this problem were considered in~\cite{farhi2014quantum,Hadfield17_QApprox,pichler2018quantum,hadfield2019quantum,farhi2020quantum}; %
see e.g.~\cite{ausiello2012complexity} for problem details, classical algorithms, and complexity.  
Proceeding as in \cite{hadfield2019quantum}, we encode subsets of $V$ with $n$ indicator variables $x_j$, which map to $n$-qubit computational basis states $\ket{x}$. 
The feasible subspace is spanned by the subset of basis states $\ket{x}$, $x\in\{0,1\}^n$, that encode independent sets (which depend on the particular problem instance). Let $\widetilde{B}=\sum_{j=1}^n \widetilde{B}_j$, where each $\widetilde{B}_j$ is such that $e^{-i\beta \widetilde{B}_j}=\Lambda_{f_j}(e^{-i\beta X_j})$ is an  $X_j$-rotation controlled on all of the neighbors of vertex $j$ not belonging to an independent set $x$, with control function $f_j=\wedge_{\ell \in nbd(j)} \overline{x}_\ell$, which when $f_j(x)=1$ implies that flipping the $j$th bit of $x$ cannot cause the independent set property to be violated. In \cite{hadfield2019quantum,hadfield2018representation} we show that 
$$\widetilde{B}_j = X_j  \prod_{\ell \in nbd(j)} \overline{x}_\ell $$
suffices, where %
$\overline{x}_\ell=\frac12 I + \frac12Z_\ell$ 
represents the (negation of) the Boolean function returning the $j$th bit of~$x$, such when its control function is false $\widetilde{B}_j$ acts as the Identity.

Expanding $\widetilde{B}_j$ %
as a Pauli sum, the number of terms in this representation of $\widetilde{B}_j$ may be exponential in $n$, for example if $deg(j)=\Theta(n)$. Nevertheless, each multi-controlled $X_j$-rotation $e^{-i\widetilde{B}_j}$ can be efficiently implemented with basic quantum gates \cite{barenco1995elementary,nielsen2010quantum, hadfield2018thesis}. 
As products of $B_j$ operators connect every feasible state $\ket{x}$ to the empty set state $\ket{00\dots0}$, it follows~\cite{hadfield2019quantum} that sufficiently many applications of $e^{-i\beta \widetilde{B}_j}$ for different~$j$ can connect any two feasible solution states $\ket{x}$ and $\ket{y}$.  (Note that~\cite{hadfield2019quantum} considers a variety of possible mixing unitaries, constructed, for example, as $e^{-\beta\widetilde{B}}$ or $\prod_j e^{-\beta\widetilde{B}_j}$, which are inequivalent; we do not deal with this distinction here and focus instead on the classical calculus derived from the~$\widetilde{B}_j$ generally.)   

The mixing Hamiltonians $\widetilde{B}_j$ induce a neighborhood structure on the space of feasible solutions, in this case the unit Hamming distance strings that are also feasible.  %
For each $j=1,\dots n$ and independent set~$x$ we define the local differences 
$$  \widetilde{\partial}_j c (x) :=\left\{ \begin{array}{ll}
                \partial_j c(x)  \qquad \qquad  \text{ if } f_j(x) = 1\\
                 \qquad 0\qquad \qquad \text{ else,}
                \end{array}
                \right.
               $$
i.e., our usual notion of partial difference but restricted to independent sets (feasible strings) connected by single bit flips. 
Clearly, this structure may be applied to other cost functions over independent sets. For our case of maximizing cardinality, assuming the restriction to feasible states the cost function $c(x)=|x|$ is the Hamming weight of each string, and maps to the Hamiltonian $C=\frac{n}2 I - \frac{1}{2}\sum_j Z_j$. Hence, observe that each $\widetilde{\partial}_j c (x)\in\{-1,0,1\}$. 
The cost divergence then becomes $$\widetilde{dc}(x)=\sum_j \widetilde{\partial}_j c(x),$$ %
which is equal to the number of independent sets obtainable from $x$ by flipping a bit from $0$ to $1$, minus the number obtained flipping a $1$ to $0$. 

The cost gradient becomes  $$\widetilde{\nabla}C=[\widetilde{B},C]=\sum_j\nabla_{\widetilde{B}_j}C,$$
where
the operators $\nabla_{\widetilde{B}_j}$ act for each feasible $\ket{x}$ as 
$$\nabla_{\widetilde{B_j}}C\ket{x}= -\widetilde{\partial}_j c (x)\ket{x^{(j)}}.
$$

For example, consider the possible choice of initial state $\ket{00\dots 0}$ which encodes the empty set. As all single-vertex sets are also independent, $\widetilde{\nabla}C$ acts to give an equal superposition of them as  $$\widetilde{\nabla}C\ket{00\dots 0}=-\ket{100\dots}-\ket{010\dots}-\dots-\ket{0\dots 01}.$$ %
Note that different choices of initial state will lead to different such relations.

Mixed gradients follow similarly. For example, 
as the cost Hamiltonian is $1$-local, it follows
$\widetilde{\nabla}C= i Y_j \otimes (\prod_{\ell \in nbd(j)} \overline{x}_\ell ),$ and so %
we have
$$ \nabla_C \widetilde{\nabla}C = i \widetilde{B}.$$
Further results concerning higher order gradients and
initial state expectation values %
may be similarly derived as %
for the %
transverse-field mixer case. %

\subsection{Graph Coloring}
\label{sec:graphColoring}
Here we consider optimization problems over colorings of a graph. 
Given a graph $G=(V,E)$, $|V|=n$,  and $k$ colors, 
we say a \textit{valid coloring} is assignment of a unique color to each vertex, and a \textit{proper coloring} is a valid assignment such that each edge connect vertices of different colors. 
Here we assume an arbitrary cost %
function $c(y)$ we seek to optimize over %
valid colorings $y$. For example, the Max-$k$-Cut generalization of MaxCut is the maximization problem where $c(y)$ gives the number of edges with different colors. %
Several other NP-hard variants of optimization problems over valid or proper graph colorings %
are considered in~\cite{hadfield2019quantum}.

Generally, depending on $k$ and the qubit encoding used, some bitstrings may represent %
invalid colorings. 
Following~\cite{hadfield2019quantum,wang2019xy} we consider encoding the coloring of each vertex using the Hamming weight-$1$ subspace of $k$ qubits, for $kn$ qubits total. (This encoding may be equivalently viewed as $n$ $k$-dits~\cite{hadfield2019quantum}, which we utilize below.) 
Valid vertex colorings then span a $k^n$-dimensional subspace of the $2^{kn}$-dimensional Hilbert space, %
with proper colorings corresponding to a subspace %
that depends on $G$. %
Hence we may write computational basis states encoding valid (or proper) colorings as $\ket{y}=\ket{y_1y_2\cdots y_n}$, $y_j \in [k]$.

A natural mixer for this encoding is the \textit{ring $XY$-mixer} \cite{hadfield2019quantum,wang2019xy} derived from the Hamiltonian 
$ \widetilde{B}:=\sum_{j=1}^n \widetilde{B}_j$
with 
\begin{equation} \label{eq:BGC}
    \widetilde{B}_j:=\sum_{\ell=1}^k \frac{X_{j,\ell}X_{j,{\ell+1}}+Y_{j,\ell}Y_{j,{\ell+1}}}2, 
\end{equation}
where $(j,\ell)$ labels the qubit encoding the $\ell$th color variable for the $j$th vertex and the label additions are modulo $k$. 
(Here the factor of $1/2$ is included for convenience.) 
We may %
identify left- and right-shift operators $L_j$,$R_j$ which act as $$L_j\ket{y}=\ket{y_1\dots y_{j-1}(y_j-1)y_{j+1}\dots y_n}\,=:\ket{y^{(j,-1)}},$$ 
with $L_j^\ell\ket{y}=L_jL_j\dots L_j\ket{y}=\ket{y^{(j,-\ell)}}$, and
$$R_j\ket{y}=\ket{y_1\dots y_{j-1}(y_j+1)y_{j+1}\dots y_n}\,=:\ket{y^{(j,+1)}},$$ 
with $R_j^\ell\ket{y}=\ket{y^{(j,\ell)}}$. 
As in %
\cite[App. C]{hadfield2019quantum} we may write
$$\widetilde{B}_j = L_j + R_j. $$
Hence in this case, from the structure of the mixing Hamiltonians, for a generic cost function $c(y)$ we define \textit{left and right jth partial differences} %
\begin{eqnarray*}
    \partial_{j,L}c(y)&:=&c(y^{(j,-1)})-c(y),\\
     \partial_{j,R}c(y)&:=&c(y^{(j,1)})-c(y).
\end{eqnarray*}
These functions relate to the action of the %
generalized partial gradients 
$$ \nabla_{\widetilde{B}_j}C\ket{y}=-\left(\partial_{j,L}c(y)L_j+\partial_{j,R}c(y)R_j\right)\ket{y}$$
and hence we have
$$ \widetilde{\nabla}C=-\sum_j (
L_j\partial_{j,L}C+R_j\partial_{j,R}C)
.$$
Thus in this case we see that the gradient corresponds to cost differences with respect to left and right shift of each $k$-dit variable.

The Hamiltonians $\widetilde{B}_j,\widetilde{B}$ together with the diagonal cost Hamiltonian~$C$ and their derived gradients have the important property of mapping valid colorings to valid colorings. 
This property is preserved if we replace the ring structure of the sum in \eqref{eq:BGC} with a fully connected one, i.e., 
\begin{equation} \label{eq:BGCb}
    \widetilde{B}_j=\sum_{\ell<m} \frac{X_{j,\ell}X_{j,m}+Y_{j,\ell}Y_{j,m}}2.
\end{equation}
in which case the above expressions generalize to shifts in $k$-many directions. This property is maintained if we instead consider \textit{sequential mixers} $U(\beta)=\prod U_i(\beta)$ built from \textit{partial mixers} $U_i(\beta)=e^{-i\beta\frac{X_{j,\ell}X_{j,m}+Y_{j,\ell}Y_{j,m}}2}$, in which case the order of the $U_i$ becomes important due to noncommutativity~\cite{hadfield2019quantum} and this will be reflected when applying our framework.

Finally, these QAOA constructions, and our framework, extend to problems where we seek to restrict to proper graph colorings by combining the XY mixers here with control,  analogous to the controlled transverse-field mixers for Max-Independent-Set~\cite{hadfield2019quantum}.  

\sh{
\subsection{Electronic Structure} %
\label{sec:elecStruct}
Here we outline how our framework may be applied to problems beyond combinatorial optimization, in particular the prototypical electronic structure problem of quantum chemistry. A recent work~\cite{kremenetski2021quantum} proposed and studied a generalization of QAOA suitable for this problem, and we follow this approach here. 
This application demonstrates the breadth of our framework and how it may be applied to diverse problems and ans\"atze.
Indeed a wide variety of different ans\"atze have been proposed for %
this problem, see for example \cite{kremenetski2021quantum,romero2018strategies}.

Consider an instance of the electronic structure problem corresponding to an electronic Hamiltonian $C$ given in second quantized form, specified as a quartic polynomial of fermionic creation and anihilation operators with respect to a truncated single-particle basis, for which we seek to find its ground state and energy $E_0$; see \cite{kremenetski2021quantum,romero2018strategies} and references therein for the technical details. For such problems it is standard to assume that one has classically obtained the Hartree-Fock approximate ground state and energy~$E_{0,HF}$, and has suitably rotated the operators (orbitals) such that the Hartree-Fock ground state is encoded on the quantum computer by a single computational basis state $\ket{HF}$, which we take as the QAOA initial state. 

Identifying $C$ as the cost Hamiltonian, a natural choice of mixer is then derived from the Hartree-Fock approximation as $$\widetilde{B}=E_{0,HF} \ket{HF}\bra{HF} + \dots$$ where the terms not shown to the right similarly correspond to the excited states and energies of the Hartree-Fock approximation (i.e., are generated from the possible excitations of the Hartree orbitals). With this choice, the resulting QAOA circuits retain the desirable property of reproducing adiabatic quantum annealing behavior in particular limits as the QAOA depth $p$ becomes large. However, we note that, in contrast to the usual QAOA operators, here the mixing Hamiltonian is diagonal with respect to the computational basis while the cost Hamiltonian is not. Also observe that this QAOA construction satisfies the properties of \remref{rem:GenLems} and so many of our results apply immediately in this setting. 

In particular, the main focus of~\cite{kremenetski2021quantum} is regimes where $p$ is relatively large, and phase and mixing angles are drawn from a fixed schedule such that their magnitudes are relatively small. Hence we jump straight to the generalization of \thmref{thm:smallprecursed}. 
For QAOA$_p$ it then follows that to leading-order we have
\begin{eqnarray*}
     \langle C \rangle_p 
&=& 
 \langle C \rangle_{0}
\,-\, \langle \nabla_C \nabla C \rangle_0
\sum_{1\leq i\leq j}^p \gamma_i \beta_j \,+\, \dots\\
&=& E_{0,HF} + 2(E_{0,HF}\braket{\phi}{\phi} - \bra{\phi}\widetilde{B}\ket{\phi}) \sum_{1\leq i\leq j}^p \gamma_i \beta_j + \dots,
\end{eqnarray*}
where the terms not shown the the right are again proportional to higher-order polynomials in the angles. The second line follows using $\ket{\phi}:=C\ket{HF}$, and using standard properties of Hartree-Fock approximations it %
follows that $E_{0,HF}\braket{\phi}{\phi} - \bra{\phi}\widetilde{B}\ket{\phi} \leq 0$. 
This last property shows that, for positive angles,  to leading order the QAOA ground state energy estimate is guaranteed to improve upon the Hartree-Fock estimate.   

Further analysis may be applied in the spirit of our results obtained for combinatorial optimization. In particular, in future work we will %
leverage this generalization of the framework to help explain universal behavior empirically observed in QAOA performance (phase) diagrams across different problems and settings in~\cite{kremenetski2021quantum}. Finally, we remark that the underlying approach and analysis of our framework is relatively straightforward to extend to particular ansatz beyond QAOA, such as the UCC-SD (unitary coupled cluster singles and doubles) ansatz considered in~\cite{romero2018strategies}. 

}

\section{Discussion and Future Directions}
\label{sec:discussion}

Our framework provides tools for, and its application provides new insights into, 
layered quantum algorithms. 
We focused on our framework's application to quantum alternating operator ans\"atze circuits.
Specifically, we used it to obtain both
new results and more succinct ways of expressing and seeing existing results
for transverse field QAOA$_p$ for general $p$, not just $p=1$. We also discussed two examples of our framework applied to QAOA with more complicated mixing operators.

Our framework applies more broadly than just the cases discussed in this paper, providing a promising approach to obtain both additional results for QAOA and insights into other types of problems and quantum algorithms. 

For example, the framework can be applied to eigenvalue problems related
to physical systems. For instance, the general conjugation formula \eqref{eq:infinitesimalConj}
has been previously invoked \cite[Sec. I]{romero2018strategies} in the context of particular 
parameterized quantum circuits for the electronic structure problem of quantum chemistry. %
 \sh{Similar to but more general than the quantum alternating operator ansatz approach of~\cite{kremenetski2021quantum} considered in \secref{sec:elecStruct},} circuit ans\"atze for this problem %
are often built from 
fermionic creation and annihilation operators and simple reference (e.g., Hartree-Fock) initial 
states. These operators then have %
natural 
interpretations in terms of transitions between occupied orbitals, analogous to the operator-function correspondence of 
\secref{sec:frameworkOverview}, %
suggesting a route for future work to apply and extend our framework and results in detail to %
different ans\"atze for this setting. %

\sh{For QAOA and beyond, questions related to parameter setting may be amenable to the series approaches of our framework, in particular toward a better understanding of parameter concentration and landscape independence phenomena, as well as implications for algorithm performance limitations~\cite{brandao2018fixed,bravyi2019obstacles,marwaha2021bounds,chou2021limitations}. Moreover our framework may be applied in}
both in ideal
and realistic (noisy) settings. %
For example, in cases when noise largely flattens the cost expectation
landscape \cite{marshall2020characterizing,wang2020noise},
our series expressions with only a few terms may be sufficient to
capture key aspects of the behavior.

In addition to further applications to QAOA, our framework can give insights into recently proposed variants such as adaptive QAOA \cite{zhu2020adaptive}
and recursive QAOA (RQAOA) \cite{bravyi2019obstacles}.
The framework and obtained results could also provide insight into quantum annealing (QA) and adiabatic quantum computing (AQC) given the close ties between parameters for QAOA and QA or AQC schedules \cite{brady2021optimal,brady2021behavior,zhou2018quantum,kremenetski2021quantum,wurtz2022counterdiabaticity}, including to cases with advanced drivers \cite{hen2016quantum,hen2016driver,leipold2020constructing}.
One such application of our formalism is designing
more effective mixing operators and ans\"atze, and facilitating quantitative
means for comparing them. Our framework
suggests direct approaches to incorporating cost function information into
mixing operators, such as, for example, using $i\nabla C$ as a mixing
Hamiltonian. %
Indeed, the ansatz
$U(\alpha)\ket{s}=e^{-i\alpha\, i\nabla C}\ket{s}$ has the same
leading-order contribution to $\langle C\rangle$ as QAOA$_1$ for
$\alpha=\gamma\beta$ when using the transverse-field mixer.  Another
possibility is to incorporate measurements or expectation
values of cost gradient observables, such as $i\nabla C$ or $\nabla_C \nabla
C$, directly into the algorithm or parameter search procedure.
\thmref{thm:costExpecHeis}
shows that the cost expectation of a QAOA$_p$ circuit can be computed or
approximated in terms of expectation values taken for a corresponding circuit
with $p-1$ levels of the cost Hamiltonian and its
cost gradient operators, which could be estimated classically or via a
quantum computer.
In this vein, a recent paper \cite{magann2021feedback} proposes an adaptive
parameter setting strategy that involves repeatedly obtaining estimates of
$\langle i\nabla C \rangle$.

As a practical matter, evaluating series approximations derived from our
framework is straightforward but can become tedious when one
desires to include many terms.
In these situations, computer algebra systems or numerical software can readily
compute the required expressions. Moreover, manipulating the resulting series
can improve the accuracy for a fixed number of terms. For example, instead of
truncating the series expansion at a finite number of terms, giving a
polynomial approximation, often better numerical approximations arise from Pade
approximants, which are rational functions whose coefficients match the
truncated Taylor series~\cite{orszag1978advanced} (cf.~\figref{fig:ramp}).

These potential applications of our formalism highlight its general
applicability to understanding algorithms based on layered quantum circuits,
including identifying new operators and ans\"atze matched to the
structure of specific problems. 

\section*{Acknowledgments}
We thank the members of the Quantum Artificial Intelligence Lab (QuAIL) for valuable feedback and discussions. 
We are grateful for support from NASA Ames Research Center and the Defense Advanced Research Projects Agency (DARPA) via IAA8839 annex 114. S.H. and T.H. are supported by the NASA Academic Mission Services, Contract
No. NNA16BD14C. 

\addcontentsline{toc}{section}{References}
\bibliographystyle{ieeetr}
\bibliography{bib}%

\newcommand{\changelocaltocdepth}[1]{%
  \addtocontents{toc}{\protect\setcounter{tocdepth}{#1}}%
  \setcounter{tocdepth}{#1}%
}
\changelocaltocdepth{1}

\begin{appendix}
\section{Sum-of-paths perspective for QAOA} \label{app:sumOfPaths} %
\sh{ 

Here we briefly show how QAOA 
probabilities 
(and similarly cost expectation values)
can be %
expressed as sums over sequences of bitstrings (\lq\lq paths\rq\rq), each weighted by functions of the Hamming distances and cost function values along the path. 
This ``sum-of-paths" perspective complements the main results %
of our framework and provides additional %
intuition, though the resulting expressions %
are no longer power series in the algorithm parameters. Indeed several results of the paper were originally derived using a sum-of-paths approach, but with significantly more complicated proofs than those obtained with our QAOA framework; on the other hand, some tasks are convenient in the sum-of-paths formulation such as working directly with QAOA state amplitudes. 

Consider the QAOA mixing operator matrix elements $u_{d(x,y)}(\beta)=\bra{x}U_M(\beta)\ket{y}=\cos^n(\beta)(-i\tan\beta)^{d(x,y)}$ given in \eqref{eq:mixingmatrixelements} for Hamming distance $d(x,y)$, which satisfy $u_d^\dagger = (-1)^d u_d$. 
For QAOA$_1$, using the computational basis decomposition %
$I=%
\sum_y\ket{y}\bra{y}$
the probability amplitude corresponding to each string 
$x\in\{0,1\}^n$ is 
\begin{equation} \label{eq:amp1sumpaths}
    \bra{x}\ket{\gamma \beta}
= \sum_y\bra{x}U_M(\beta)\ket{y}\bra{y}U_P(\gamma)\ket{s}=
\frac{1}{\sqrt{2^n}}\sum_y %
u_{d(x,y)}(\beta)
e^{-i\gamma c(y)},
\end{equation}
i.e., the amplitude 
$\bra{x}\ket{\gamma \beta}$ may be %
expressed as sum of \lq\lq single-step paths\rq\rq (from each possible $y\in\{0,1\}^n$, to $x$) with %
each path contributing $u_{d(x,y)}(\beta)
e^{-i\gamma c(y)}$ times its initial amplitude $\bra{y}\ket{s}=1/\sqrt{2^n}$.
The probability $P_1(x)=|\bra{x}\ket{\gamma \beta}|^2$ of measuring $x$ is then %
\begin{eqnarray} \label{eq:prob1sumpaths}
P_1(x)&=&%
\frac1{2^n} \sum_{y,z} u^\dagger_{d(x,z)}u_{d(x,y)}
e^{-i\gamma (c(y)-c(z))}\nonumber\\
&=&\frac1{2^n} \sum_{y,z} \cos^{2n}\beta (\tan\beta)^{d(x,y)+d(x,z)}i^{d(x,z)-d(x,y)}
e^{-i\gamma (c(y)-c(z))}\nonumber\\
&=&\frac{\cos^{2n}\beta}{2^n} \sum_{y,z} (\tan\beta)^{d(x,y)+d(x,z)} i^{d(x,z)-d(x,y)} \left(
\cos(\gamma(c(y)-c(z))) - i \sin(\gamma(c(y)-c(z)))\right)\nonumber\\
&=&\frac{\cos^{2n}\beta}{2^n}\, \bigg( \sum_{\substack{y,z \\ d(x,z)-d(x,y)\textrm{ even}}} (\tan\beta)^{d(x,y)+d(x,z)} (-1)^{\frac{d(x,z)-d(x,y)}2}
\cos(\gamma(c(y)-c(z))) \nonumber\\ 
& &\,\,\,+  \sum_{\substack{y,z \\ d(x,z)-d(x,y)\textrm{ odd}}} (\tan\beta)^{d(x,y)+d(x,z)} (-1)^{\frac{d(x,z)-d(x,y)-1}2} \sin(\gamma(c(y)-c(z)))\bigg)
\end{eqnarray}
where the last equation follows because  $\cos$ and $\sin$ are even and odd functions, respectively, and shows explicitly that all contributing terms are real valued, but with differing signs that may lead to term cancellation (i.e., paths may \lq\lq interfere\rq\rq). The signs depend on the %
algorithm parameters, the difference in cost function between $y$ and $z$, and their 
Hamming distances of from $x$. 

Thus we see that the probability $P_1(x)$ corresponds to a sum of two-step paths ($z$ to $x$ to $y$) %
with weights as shown in \eqref{eq:prob1sumpaths}. The cost expectation $\langle C\rangle_1=\sum_x c(x)P_1(x)$ then follows from additionally weighting all paths %
in $P_1(x)$ by $c(x)$, and then summing over all $x\in\{0,1\}$. Furthermore, Taylor expansions of \eqref{eq:prob1sumpaths} %
reproduce the leading-order contributions of \thmref{thm1:smallAngles}, though with considerably more effort than required for our proof.

It is straightforward %
to extend the sum-of-paths perspective to QAOA$_p$ with $p>1$, at the expense of additional notation. In this case probabilities %
correspond to paths of length~$2p$.  
Let $u^{(j)}_{d(x,z)}:=u_{d(x,z)}(\gamma_j,\beta_j)$, and let %
$Q_j(x)$ denote the amplitude of computational basis state $\ket{x}$ after $j$ layers of QAOA have been applied. Clearly, %
$Q_0(x)= 1/\sqrt{2^n}$, and %
$Q_1(x)=\bra{x}\ket{\gamma\beta}$ is given in \eqref{eq:amp1sumpaths}. Generalizing \eqref{eq:amp1sumpaths} 
for a QAOA$_p$ state gives %
\begin{equation}
  Q_p(x) = \sum_z u^{(p)}_{d(x,z)}  e^{-i\gamma_p c(z)}Q_{p-1}(z),
\end{equation}
which may be expanded recursively to give 
\begin{equation} \label{eq:ampgenp}
  Q_p(x) = \frac1{\sqrt{2^n}}\sum_{z_1,z_2,\dots,z_p} 
u^{(1)}_{d(z_1,z_2)} u^{(2)}_{d(z_2,z_3)} \dots u^{(p)}_{d(z_{p},x)}
e^{-i (\gamma_1c(z_1) + \gamma_2c(z_2) + \dots \gamma_p c(z_p))}.
\end{equation}
Hence computational basis amplitudes after $p$ levels of QAOA are likewise given by paths consisting of $p$-tuples of intermediate bitstrings, weighted by 
phase factors corresponding to cost function evaluations %
and mixing matrix elements corresponding to Hamming weight differences %
along the path.

Expressions for the probability and expected cost may then be obtained from $P_p(x)=|Q_p(x)|^2=Q_p^\dagger(x) Q_p(x)$ as above. While in principle this approach can also lead to results such as \thmref{thm:allanglessmall}, our %
proof provides a more succinct route. 
Nevertheless, the sum-of-paths perspective gives a different but complementary perspective into QAOA, and the general approach outlined here may be extended to  
a wider variety of ans\"atze, %
such as QAOA with %
generalized mixers,
encodings, and initial states. %
}

\section{%
Initial state expectation values}
\label{app:expecVals}
Here we extend the results of \secref{sec:expecVals}. 
For a linear operator $A$ on $n$ qubits we write its computational basis representation $A=\sum_{x,y}a_{xy}\ket{x}\bra{y}$ with matrix elements
$ a_{xy}:= \bra{x}A\ket{y}\,\in \complex$
for $x,y\in\{0,1\}^n$. The operator $A$ is traceless if and only if $\sum_x a_{xx} =0$. %

\begin{lem}  \label{lem:expec1}
Let $C$ be a cost Hamiltonian and $A$ a cost gradient as in (\ref{eq:costGradOpGen})
such that $[A,C]\neq 0$ and 
$A\ket{s} = \frac{1}{\sqrt{2^n}}\sum_x a(x) \ket{x}$ for a real function $a(x)$. %

Then there exists a real 
function $g(x)$ satisfying $\nabla_C A \ket{s} = \frac{1}{\sqrt{2^n}}\sum_x g(x) \ket{x}$, and:
\begin{itemize}
\item If %
$A$ is skew-adjoint $A^\dagger = -A$ then $\;\sum_x a(x) = 0$, 
\begin{equation*}
   \langle \nabla_C A \rangle_0  = \frac{2}{2^n}\sum_x c(x)a(x)= \frac{1}{2^n}\sum_x g(x),
    \;\;\;\;\; \text{ and } \;\;\;\;\; \langle \nabla^2_C A \rangle_0  =0. 
 \end{equation*}

\item Else if instead $A$ is self-adjoint $A^\dagger =A$ then  $\; \langle \nabla_C A \rangle_0 = \sum_x g(x) =0,\;$ and 
\begin{eqnarray*}
    \langle \nabla^2_C A \rangle_0 &=&  \frac1{2^n}\sum_{x\neq y} (c(x)-c(y))^2 a_{xy} \,=\, \frac{2}{2^n}\sum_x c(x)g(x)\\
    &=&\frac{2}{2^n}\sum_x c^2(x) a(x) -\frac2{2^n}\sum_{x,y}c(x)c(y)a_{xy}.
\end{eqnarray*}
\end{itemize}
\end{lem}

The condition $[A,C]\neq 0$ is included because %
otherwise trivially $\nabla_C A = \nabla^2_C A=0$ and so the quantities of the lemma become identically zero. Recall a cost gradient $A$ may be uniquely decomposed into its diagonal and off-diagonal parts $A=A_{diag}+A'$ with respect to the computational basis, with $\nabla_C A = \nabla_C A'$. 
As the results of the self-adjoint case are invariant under replacing $A$ by $A+E$ for a traceless real diagonal matrix $E$ (such that $\nabla_C E=0$), we can replace $a(x)$ above by any function $a(x)+e(x)$ where $\sum_x e(x)=0$; hence we can often ignore any diagonal component of cost gradient operators in computing initial state expectation values.

\begin{proof}
From the definition of $A$ we have $a(x)=\sum_y a_{xy}$, and $a_{yx}=\pm a_{xy}$ when $A^\dagger = \pm A$, which implies $\sum_x a(x)= 0$ when $A^\dagger = - A$.
Using $A=\sum_{x,y}a_{xy}\ket{x}\bra{y}$, 
the function $g$ is defined as $\tfrac1{\sqrt{2^n}}g(x):=\bra{x} \nabla_C A \ket{s}$ and satisfies
$$\tfrac1{\sqrt{2^n}}g(x) %
=\bra{x} CA-AC\ket{s}=\tfrac1{\sqrt{2^n}}c(x)a(x)-\tfrac1{\sqrt{2^n}}\sum_y c(y)a_{xy}=\tfrac1{\sqrt{2^n}}\sum_y a_{xy}(c(x)-c(y)),$$
which implies $\sum_x g(x) = 0$ when $A^\dagger = A$. %
Indeed, from $\bra{s}A=\pm \tfrac{1}{2^n}\sum_x a(x)\bra{x}$ we have
$$\langle g\rangle_0=\bra{s}\nabla_C  A\ket{s} = \bra{s}CA\ket{s}- \bra{s}AC\ket{s}
=\langle ca\rangle_0\mp \langle ca\rangle_0$$
when $A^\dagger = \pm A$, which shows the results for $\langle\nabla_C A\rangle_0$. 
Similarly,
$\bra{s}\nabla_CA=\mp \tfrac{1}{2^n}\sum_x g(x)\bra{x}$ 
and so 
\begin{eqnarray*}
     \bra{s} \nabla^2_C A \ket{s}&=&\bra{s}C\,\nabla_CA\ket{s}- \bra{s}\nabla_CA\,C\ket{s}
=\langle cg\rangle_0\pm \langle cg\rangle_0\\
&=&\bra{s}(C^2A-2CAC+AC^2)\ket{s}=\langle c^2a\rangle_0\pm \langle c^2a\rangle_0 - 2\langle CAC\rangle_0.
\end{eqnarray*}
Finally, we use   
$\langle c^2a\rangle_0=\tfrac1{2^n}\sum_x c^2(x)a(x)=\tfrac1{2^n}\sum_{x,y} c^2(x)a_{xy}$
and 
$\langle CAC\rangle_0=\tfrac1{2^n}\sum_{x,y}c(x)c(y)a_{xy}$ %
for the case $A=A^\dagger$ to give 
$$ \langle \nabla^2_C A\rangle_0 =\frac2{2^n}\sum_{x,y} (c^2(x)-c(x)c(y))a_{xy}=\frac1{2^n}\sum_{x,y} (c^2(x)+c^2(y)-2c(x)c(y))a_{xy}$$
which shows the terms in the sum with $x=y$ are zero and 
implies that last equation of the lemma. 
\end{proof}

We use the lemma to prove Equations (\ref{eq:expecDcD3C}), (\ref{eq:expecDc3DC}), and (\ref{eq:expecDc2D2C}) concerning the %
nonzero initial state expectation values of cost gradient at fourth order. 
The technique of the proof may be used to similarly compute higher order expectation values. 
\begin{lem} \label{lem:higherOrderExpectations}%
For a cost Hamiltonian $C$ representing real function $c(x)$ 
we have:
\begin{itemize}
    \item $\langle \nabla_C \nabla^3 C \rangle_0 = \frac2{2^n}\sum_x c(x)d^3c(x)$
               \item $\langle \nabla^2_C \nabla^2 C \rangle_0 =\langle \nabla_C \nabla \nabla_C \nabla C \rangle_0 =\frac2{2^n}\sum_x \sum_{j<k} (\partial_{j,k}c(x))^2 \partial_j\partial_k c(x)$
                     \item $\langle \nabla^3_C \nabla C \rangle_0 = \frac2{2^n}\sum_x c(x)\sum_{j} (\partial_j c(x))^3 =- \frac1{2^n}\sum_x \sum_{j} (\partial_j c(x))^4$
\end{itemize}
where $\partial_{j,k}c(x) := c(x^{(j,k)})-c(x)=\partial_j\partial_k c(x) + \partial_j c(x)+\partial_k c(x)$. %

In particular, for a QUBO cost Hamiltonian $C=a_0 +\sum_i a_i Z_i + \sum_{i<j}a_{ij}Z_iZ_j$ these quantities can be further expressed (defining $a_{ji}:=a_{ij}$) as:
\begin{itemize}
    \item $\langle \nabla_C \nabla^3 C \rangle_0 = -16 \sum_ia_i^2 - 128 \sum_{i<j}a_{ij}^2$
\item $\langle \nabla^2_C \nabla^2 C \rangle_0 = 64 \sum_{i<j}a_ia_ja_{ij}+  192 \sum_{i<j<k} a_{ij}a_{jk}a_{ik}$
\item $\langle \nabla^3_C \nabla C \rangle_0  = -16\sum_i a_i^4 - 32\sum_{i<j} a_{ij}^4 - 96 \sum_{i<j}a_{ij}^2 (a_i^2+a_j^2 + \frac12 %
\sum_{k\neq i,j}(a_{ik}^2+a_{jk}^2))$
\end{itemize}
\end{lem}
The second part of the lemma shows that we may always efficiently compute the quantities of the first part for cost functions such as QUBOs. Moreover, each expectation values reflects the topology of the underlying QUBO problem graph; for example, for MaxCut, $\langle \nabla^2_C \nabla^2 C \rangle_0$ is seen to be identically zero on triangle-free graphs.  
\begin{proof}
Recall %
that any diagonal part of $A$ in Lemma~\ref{lem:expec1} can be ignored (such as, for example, in $\nabla^2 C$ for MaxCut from Table~\ref{tab:summary}) as it will not affect the initial state expectation values. 

The first result follows from the first case of Lemma~\ref{lem:expec1} with $A=\nabla^3 C=-A^\dagger$ which from (\ref{eq:gradCs}) satisfies %
$a(x)=d^3c(x)$ to give $\langle \nabla_C A \rangle_0 = \frac2{2^n}\sum_x c(x)d^3c(x)$. 
The second result follows applying the second case of Lemma~\ref{lem:expec1} to $A$ defined as the non-diagonal part of $\nabla^2 C$, which from (\ref{eq:grad2Cx}) has matrix elements $a_{xy}=2\partial_j\partial_k c(x)=2\partial_j\partial_k c(x^{(j,k)})$ when $x$ and $y=x^{(j,k)}$ differ by two bit flips $j,k$, and $a_{xy}=0$ otherwise, and so 
$$\langle \nabla^2_C \nabla^2 C \rangle_0 = \frac2{2^n}\sum_x \sum_{j<k} (c(x)-c(x^{(j,k)}))^2\partial_j\partial_k c(x),$$
where $c(x^{(j,k)})-c(x)=\partial_j\partial_k c(x) + \partial_j c(x)+\partial_k c(x)$ follow from the definitions of Sec.~\ref{sec:costDiffs}.  
Additionally, from the Jacobi identity and (\ref{eq:costJacobi}) we have $\langle \nabla_C \nabla \nabla_C \nabla_C \rangle_0 =\langle  \nabla_C^2\nabla^2 C\rangle_0$. 
For the third result, %
let $A=\nabla_C \nabla C=A^\dagger$, for which from (\ref{eq:DCDCx}) %
we have  
$a_{xy}=-(\partial_jc(x))^2$ when $y=x^{(j)}$ and $0$ otherwise, %
and hence  the first equality of the second case of Lemma~\ref{lem:expec1} gives 
$\langle \nabla^2_C A \rangle_0=\tfrac1{2^n}\sum_x\sum_j (c(x)-c(x^{(j)})^2(-(\partial_jc(x))^2)=-\tfrac1{2^n}\sum_x\sum_j (\partial_jc(x))^4$. This sum can be easily rearranged %
to become $\tfrac2{2^n}\sum_x c(x) \sum_j (\partial_jc(x))^3$, or alternaively this may be obtained from the second equality of the second case of Lemma~\ref{lem:expec1} 
by defining the function $g(x)=(\partial_j c(x))^3$ from (\ref{eq:DlCDCs}) that corresponds to $\nabla_C A$.

Now further suppose we have a QUBO cost Hamiltonian~$C:=a_0I + C_1 + C_2$. 
We proceed by expanding each cost gradient as a linear combination of Pauli operators, and then computing the initial state expectation value using the observation that only terms containing stricilty $I$ or $X$ factors can contribute (see \cite[Sec. YY]{hadfield2018thesis} for details). Note that this approach allows us to avoid dealing with exactly computing the coefficients of various terms which can't contribute. 
 
 Let $C_{2,Y}:=\sum_{i<j}a_{ij}Y_iY_j$ and here we write $a_{ji}=a_{ij}$, $a_{ii}:=0$, and $j\in nbd(i)$ if $a_{ij}\neq 0$. Applying Lemma~\ref{lem:quboGrads} and the Pauli matrix commutation relations ($[X,Y]=2iZ$ and cyclic permutations) we have $\nabla C = \nabla C_1 + \nabla C_2$, $\nabla^2 C = 4C_1 + 8C_2 - 8C_{2,Y}$, and $\nabla^3 C=4\nabla C_1 + 16\nabla C_2$. %
 Hence $\nabla_C \nabla^3 C = -16\sum_i a_i^2 X_i -64\sum_{i<j}a_{ij}^2(X_i+X_j) +\dots$ where the terms not shown to the right each contain $Z$ factors (for example, terms proportional to $a_ia_{ij}X_iZ_j$), which implies
 $\langle \nabla_C \nabla^3 C \rangle_0=-16\sum_i a_i^2  -128\sum_{i<j}a_{ij}^2$ as stated. 
 
 Similarly,  $\nabla_C \nabla^2 C = -8 \nabla_C C_{2,Y}=16i\sum_{i<j}a_{ij}(a_iX_iY_j + a_jY_i X_j+\sum_{k\neq i,j} (a_{ik}Z_k X_iY_j+a_{jk}Y_iX_jZ_k))$, and so 
 $\nabla_C^2 \nabla^2 C = 64 \sum_{i<j}a_{ij}a_ia_jX_iX_j+64
 \sum_{i<j}\sum_{k} a_{ij}a_{jk}a_{ik}X_iX_j +\dots,$
with the terms to the right each containing $Y$ or $Z$ factors. As $a_{ii}=0$ the $a_{ij}a_{jk}a_{ik}$ terms are nonzero only for triangles in the graph formed from nonzero $a_{ij}$ which implies  $\langle \nabla_C^2 \nabla^2 C\rangle_0 = 64 \sum_{i<j}a_{ij}a_ia_j+3*64
 \sum_{i<j<k}%
 a_{ij}a_{jk}a_{ik}$. %
 
 Finally, a similar calculation for gives the result for $\langle\nabla_C^3\nabla C\rangle_0$, as can readily be seen expanding the identity $\nabla_C \nabla C = \nabla_C^4 B$. %
\end{proof}

\subsection{QAOA with small mixing angle}\label{app:smallMix}
Here we show two proofs related to %
results in \secssref{sec:lightcones}{sec:smallMixingAngle}{sec:smallMixingAnglep}.

\begin{proof}[Proof of \eqref{eq:expecgammanabla2C}]
Recal the notation $\nabla = \sum_i \nabla_i$ with $\nabla_i = [X_i,\cdot]$. 
As for any $C$ 
$$\nabla^2 C=\sum_i\nabla_i^2 C + 2\sum_{i<j}\nabla_i\nabla_jC= -2DC +2\sum_{i<j}\nabla_i\nabla_jC$$ and $\langle DC\rangle_0 =0$ we have
\begin{eqnarray}
\bra{\gamma}\nabla^2 C \ket{\gamma} = 2\sum_{i<j} \langle \nabla_i\nabla_j C\rangle_0.
\end{eqnarray}
Next observe $\partial_iC + \partial_kC + \partial_i\partial_kC= -2(C^{\{i\setminus k\}}+C^{\{k\setminus i\}})$, where $C^{\{j\setminus k\}}$ denotes the terms in $C$ that contain a $Z_j$ but not a $Z_k$,  
 which gives 
 $$U_P^\dagger X_iX_kU_p=e^{i\gamma(C^{\{i\setminus k\}}+C^{\{k\setminus i\}})}X_iX_k.$$
Hence \eqref{eq:expecgammanabla2C} follows by observing  
\begin{eqnarray*}
\langle \nabla_i\nabla_j C\rangle_0 &=& \langle X_iX_j C-X_iCX_j-X_jCX_i+CX_iX_j\rangle_0\\
&=& \langle Ce^{-2i\gamma(C^{\{i\setminus j\}}+C^{\{j\setminus i\}})}  + e^{2i\gamma(C^{\{i\setminus j\}}+C^{\{j\setminus i\}})} C\rangle_0\\
&-& \langle Ce^{-2i\gamma(C^{\{i\setminus j\}}-C^{\{j\setminus i\}})}  + e^{2i\gamma(C^{\{i\setminus j\}}-C^{\{j\setminus i\}})} C\rangle_0\\
&=& 2\langle (\cos(C^{\{i\setminus j\}}+C^{\{j\setminus i\}})-2\cos((\gamma C^{\{i\setminus j\}}-C^{\{j\setminus i\}})))C\rangle\\
&=& -4\langle \sin(\gamma C^{\{i\setminus j\}})\sin(\gamma C^{\{i\setminus j\}})C\rangle_0\\
&=& -4\langle \sin(\gamma(\tfrac12\partial_iC +\tfrac14 \partial_i\partial_jC )) \sin(\gamma(\tfrac12\partial_jC +\tfrac14 \partial_i\partial_jC )) C\rangle_0
\end{eqnarray*}
where we used $C^{\{i\setminus j\}}=-\tfrac12\partial_iC -\tfrac14\partial_i\partial_jC$ and $\cos(a)-\cos(b)=-2\sin\tfrac{a+b}2\sin\tfrac{a-b}{2}$. 
\end{proof}

\begin{proof}[Proof of \corref{cor:smallBeta}]
Consider a general QAOA$_p$ state which we write shorthand as 
 $\ket{p}:=\ket{\boldsymbol{\gamma \beta}}_p = \sum_x q_{x,p} \ket{x},$
where the coefficients $q_{x,p}$ may be complex. 
Proceeding as in the proof of Theorem~\ref{thm:smallBetap1} %
gives to first order in $\beta$ 
\begin{eqnarray}
\langle C \rangle_{p+1} &\simeq& \bra{p}C\ket{p} + \beta \bra{p\gamma_{p+1}} i\nabla C \ket{p\gamma_{p+1}}=\langle C \rangle_p +\beta \bra{p} \nabla_{\widetilde{B}}C \ket{p}
\end{eqnarray}
where $\ket{p\gamma_{p+1}} := e^{-i\gamma_{p+1}C}\ket{p}$ and $\widetilde{B}=e^{i\gamma C}B\e^{-i\gamma C}$. 
Observing that 
\begin{eqnarray}
\nabla_{\widetilde{B}}C \ket{p} = -\sum_x \left( \sum_{j=1}^n q_{x^{(j)},p}\partial_j c(x) e^{i\gamma_{p+1} \partial_j c(x)}\right) \ket{x}, 
\end{eqnarray}
we have %
    $\langle C \rangle_{p+1} \simeq 
   \langle C \rangle_{p} - i\beta \sum_x q^\dagger_{x,p}
    \left(\sum_{j=1}^n q_{x^{(j)},p}\partial_j c(x) e^{i\gamma_{p+1} \partial_j c(x)}\right)$.
Next observe that
\begin{eqnarray*}
 \sum_x \sum_{j=1}^n  q^\dagger_{x,p} q_{x^{(j)},p}\partial_j c(x) e^{i\gamma_{p+1} \partial_j c(x)} &=& 
\sum_{j=1}^n  \sum_x q^\dagger_{x,p} q_{x^{(j)},p} e^{i\gamma_{p+1} \partial_j c(x)}
(c(x^{(j)})-c(x))\\
&=&
\sum_{j=1}^n  \sum_x c(x) (q^\dagger_{x^{(j)},p}q_{x,p} e^{-i\gamma_{p+1} \partial_j c(x)}-q^\dagger_{x,p} q_{x^{(j)},p} e^{i\gamma_{p+1} \partial_j c(x)})\\
&=&
\sum_{j=1}^n  \sum_x c(x) |q^\dagger_{x,p} q_{x^{(j)},p}| (e^{-i\alpha_{xjp}-i\gamma_{p+1} \partial_j c(x)} -  e^{i\alpha_{xjp} + i\gamma_{p+1} \partial_j c(x)})\\
&=&
\sum_{j=1}^n  \sum_x c(x) r_{xjp} (-2i\sin(\alpha_{xjp}+\gamma_{p+1} \partial_j c(x)))\\
&=&
 -\sum_x c(x) \sum_{j=1}^n  2i r_{xjp} \sin(\alpha_{xjp}+\gamma_{p+1} \partial_j c(x))
\end{eqnarray*}
where we have defined the real polar variables $r,\alpha$ as 
$r_{xjp}=|q^\dagger_{x^{(j)},p}q_{x,p}|$ and 
$ q^\dagger_{x^{(j)},p}q_{x,p} = r_{xjp}e^{-i\alpha_{xjp}}$ 
which reflect the degree of which the coefficients $q_{x,p}$ are non-real. Plugging in above implies (\ref{eq:expecsmallbetap}). 
Further, rearranging as
\begin{equation}
      \langle C \rangle_{p+1} \simeq 
    \sum_x c(x)  \left(|q_{x,p}|^2 -2\beta \sum_{j=1}^n r_{xjp}  \sin(\alpha_{xjp}+\gamma_{p+1} \partial_j c(x))  \right)
\end{equation}
shows (\ref{eq:probsmallbetap}) gives the change in probability for each state $x$ to first order. %
\end{proof}

\section{Norm and error bounds} \label{app:normAndErrorBounds}
Here we show bounds to the norms of cost gradient operators, which we subsequently apply to bound the error of our leading-order formulas. 

The spectral norm $\|A\|:=\|A\|_2$ is the maximal eigenvalue in absolute value of $A$. When $C$ represents a function $c(x)\geq 0$ to be maximized, for example a constraint satisfaction problem such as MaxCut or MaxSAT, then $\|C\|=\max_x c(x)$. %
Hence computing $\|C\|$ exactly %
can be NP-hard, %
in which case upper bounds to $\|C\|$ (such as the number of clauses) typically suffice. 

For commutator norms, commuting components should not affect the norm estimates, in particular the Identity operator component of each operator. To reflect this in our norm and error estimates, 
here we define %
\begin{equation} \label{eq:normstar}
    \|A\|_* := \min_{b\in\complex}\|A + bI\|,
\end{equation}
which in particular satisfies $\|A\|_*\leq \|A\|$, and similarly for the traceless part of $A$
\begin{equation}\label{eq:normstar2}
    \|A\|_*\leq \|A-\tfrac1{2^n}tr(A)I\|.
\end{equation}
As $\|aI\|_*=0$ for $a\in\complex$, \eqref{eq:normstar} gives a seminorm. 

For example, for MaxCut we have $C=\tfrac{|E|}2I-\tfrac12\sum_{(ij)}Z_iZ_j=:\tfrac{|E|}2I+C_Z$ which satisfies $\|C_Z\|=\tfrac{|E|}2$ %
and $\|C\|_*=\tfrac{c^*}2=\frac{\|C\|}2$. 
When $c^*\neq |E|$ we have %
$\|C\|_* < \|C-\tfrac{|E|}2I\|$ %
which illustrates that the inequality in 
\eqref{eq:normstar2} can be strict.

\subsection{Norm bounds} \label{app:normBounds}

We first consider general gradients (commutators) $\nabla_G:=[G,\cdot]$. 

\begin{lem}  \label{lem:normCGrad}
For $n$-qubit operators $A,G$ and $\ell \in \naturals$ we have
\begin{equation}  \label{eq:lemnormCGrad}
    \|\nabla_G^\ell A\|  \, \leq \, (2\|G\|_*)^\ell \|A\|_* \,\leq\, (2\|G\|)^\ell \|A\|.
\end{equation}

\noindent %
If we know a decomposition $A=A_{com}+A_{nc}$ such that $[G,A_{com}]=0$ then this improves to 
\begin{equation}  \label{eq:lemnormCGrad2}
    \|\nabla_G^\ell A\|   \,=\,     \|\nabla_G^\ell A_{nc}\| \,\leq\, (2\|G\|_*)^\ell \|A_{nc}\|_*. 
\end{equation}
\end{lem}
\begin{proof}
Let $G'=G+gI$ and $A'=A+aI$. 
As identity components always commute,
trivially $\nabla^\ell_{G'} A'=\nabla^\ell_G A$ for all $g,a$.
Selecting $g,a$ such that $\|G\|_*=\|G'\|$ and $\|A\|_*=\|A'\|$ and 
applying the triangle and  Cauchy-Schwarz inequalities to %
$\nabla_{G'} A' = G'A'-A'G'$ gives $\|\nabla_G A\|\leq 2\|G'\|\|A'\|=2\|G\|_*\|A\|_*\leq 2\|G\|\|A\|$. %
\eqref{eq:lemnormCGrad} then follows by induction on $\nabla_G^{\ell-1}A=\nabla_{G'}^{\ell-1}A'$. \eqref{eq:lemnormCGrad2} follows by the same arguments observing $\nabla_G A = \nabla_G A_{nc}$. 
\end{proof}

In particular, for the cost Hamiltonian $C=a_0I + C_Z$, \lemref{lem:normCGrad} implies \begin{equation}
    \|\nabla^\ell_C A\|\leq (2\|C_Z\|)^\ell \|A_{non-diag}\|
\end{equation}
for $A_{non-diag}$ the Pauli basis components of $A$ that contain an $X$ or $Y$ factor.

For gradients $\nabla^\ell$ with respect to the transverse-field Hamiltonian $B$ we show a %
tighter bound %
applicable to $k$-local cost Hamiltonians. (Recall an operator is $k$-local if its Pauli expansion contains terms acting nontrivially on $k$ or fewer qubits.)

\begin{lem} \label{lem:normGrad}
For a %
$k$-local operator $A$ and $\ell \in \naturals$ we have
\begin{equation} \label{eq:lem:normGrad}
 \|\nabla^\ell A\|\,\leq\, (2k)^\ell \|A\|_*\,\leq\, (2k)^\ell \|A\|,
 \end{equation}
 and this %
 bound is tight (i.e., there exists an $A$ such that $\|\nabla^\ell A\|= (2k)^\ell \|A\|_*=(2k)^\ell \|A\|$).
\end{lem}
\begin{proof}
Let $A'=aI+A$ such that $\|A\|_*=\|A\|.$
In general $A'=\sum_j w_j \sigma^{(j)}$, where $w_j\in \complex$ and each $\sigma^{(j)}$ gives a string of at most $k$ Pauli matrices.
Consider the action of $\nabla=[B,\cdot]=\sum_i[X_i,\cdot]$ on a fixed 
$\sigma^{(j)}$; clearly, of the $n$ commutators $[X_i,\sigma^{(j)}]$ for $i=1,\dots,n$, 
at most $k$ can be nonzero. 
Moreover, each $[X_i,\cdot]$ either acts unitarily on a $\sigma^{(j)}$ (up to a factor of $2$) if $\sigma^{(j)}$ contains a $Y_i$ or $Z_i$ factor, or else annihilates it. Hence $\|[X_i,\sigma^{(j)}]\|\leq 2\|\sigma^{(j)}\|=2$. 
and so Cauchy-Schwarz gives %
$$\|\nabla A\|=\|\nabla A'\|=\|  \sum_j w_j \sum_{j_{1}}^{j_{k}}  [X_{j_i},\sigma^{(j)}]\| 
\leq 2k \|A\|_*\leq 2k\|A\|.
$$
Repeating this argument for $\|\nabla(\nabla^{\ell-1}A')\|$
gives the claim %
by induction.

Next, let $A=Z_1Z_2\dots Z_k$ for some $k\leq n$; %
a straightforward calculation shows $\|\nabla^\ell A\|= (2k)^\ell \|A\|_*=(2k)^\ell \|A\|=(2k)^\ell$ for each $\ell=1,2,\dots$ as desired. 
\end{proof}

\lemsref{lem:normCGrad}{lem:normGrad} may be applied recursively to bound $\|A\|$ for $A$ an arbitrary cost gradient operator as in \lemref{lem:genGrad}. 
Norm bounds are useful for upper bounding expectation values, as for any %
Hamiltonian~$H$
and normalized quantum state $\ket{\psi}$ we have $\bra{\psi}H\ket{\psi}\leq \|H\|$.  
For example, %
for MaxCut, \lemref{lem:normGrad} gives
$\langle (i\nabla)^\ell C \rangle_p  \leq 4^\ell \frac{c^*}2$. %

\subsection{Error bounds}   \label{app:errorBounds}
We now show how to obtain error bounds for commutator expansions, which we use to below to prove  \corref{cor:thm1errorboundeps} concerning the error of the estimates of \thmref{thm1:smallAngles}. %
The same approach may be applied to bounding the error when truncating our series expressions at higher order,  %
or to deriving  
corresponding estimates for QAOA$_p$.

Here we utilize 
an integral representation of operator conjugation that follows from the fundamental theorem of calculus 
(see e.g. \cite[Eq. 23]{childs2019theory})
\begin{equation}\label{eq:error1}
    e^{i\beta \nabla_H}A=A+i\int_0^\beta d\alpha\, e^{i\alpha \nabla_H}\nabla_H A,
\end{equation}
which may be recursively applied to give 
\begin{equation} \label{eq:error2}
    e^{i\beta \nabla_H}A=A
    +i\beta\nabla_HA -\int_0^\beta d\alpha \int_0^\alpha  d\tau\, e^{i\tau \nabla_H}\nabla^2_H A,
\end{equation}
and %
similarly to pull out higher-order leading terms as desired. 
\begin{lem} \label{lem:smallmixC}
For a $k$-local matrix $A$ acting on $n$ qubits 
we have
\begin{equation*}
    \|U_M^\dagger(\beta) \,A\, U_M(\beta)-(A+i\beta \nabla A)\| \leq 2k^2 \beta^2 \|A\|_*.
\end{equation*}
\end{lem}
\begin{proof}
Applying \eqref{eq:error2} with  $U_M(\beta)$ and using Cauchy-Schwarz gives  
\begin{eqnarray*}
 \|U_M^\dagger A U_M - A -i\beta\nabla A\|&=&\|\int_0^\beta d\alpha \int_0^\alpha  d\tau e^{i\tau \nabla}\nabla^2 A\|
\,\leq\, \left|\int_0^\beta d\alpha \int_0^\alpha d\tau\right|\cdot \|e^{i\tau \nabla}\nabla^2 A\|\\
&\leq& \frac{\beta^2}2\|\nabla^2 A\|\leq 2k^2\beta^2\|A\|_*\leq 2k^2\beta^2\|A\|.
\end{eqnarray*}
where we have used $\|e^{i\tau \nabla}\nabla^2 A\|=\|\nabla^2 A\|$ from unitarity, and \lemref{lem:normGrad}. 
\end{proof}

A similar argument yields an $O((2k|\beta|)^{\ell+1}\|A\|_*)$ error bound for truncating the series for $U_M^\dagger A U_M$ at order $\ell$ (cf. \cite[Thm. 5]{childs2019theory}). 

We next bound the error of the quantities of Theorem~\ref{thm1:smallAngles} as stated in Cor.~\ref{cor:thm1errorboundeps}. 

\begin{lem} \label{lem:thm1errorbound}
For QAOA$_1$ with $k$-local cost Hamiltonian $C$ 
the error in the estimate (\ref{eq:expecC1}) of \thmref{thm1:smallAngles} is bounded as
\begin{equation} \label{eq:errorBoundC}
\left|
\langle C\rangle_1 - \,\left( \langle C \rangle_0 - \gamma \beta \langle \nabla_C \nabla C \rangle_0 \right) \right|
\,\leq\, 
4|\beta|\gamma^2 k\|C\|_*^3 + 4\beta^2|\gamma|  k^2\|C\|_*^2 .
\end{equation}
\end{lem}

\begin{proof}%
Applying \eqsref{eq:error1}{eq:error2} to QAOA$_1$ we have  
\begin{eqnarray*} 
    e^{i\gamma \nabla_C} e^{i\beta \nabla}C
    &=& C +i\beta e^{i\gamma \nabla_C} \nabla C -\int_0^\beta d\alpha \int_0^\alpha  d\tau e^{i\gamma \nabla_C} e^{i\tau \nabla}\nabla^2 C\\
    &=& C + i\beta\nabla C - \gamma\beta \nabla_C\nabla C -\beta \int_0^\gamma d\alpha \int_0^\alpha  d\tau e^{i\tau \nabla_C}\nabla_C^2 \nabla C\\
    & &-\, \int_0^\beta d\alpha \int_0^\alpha  d\tau  \left(e^{i\tau \nabla}\nabla^2 C + i\int_0^\gamma da e^{ia\nabla_C} \nabla_C (e^{i\tau \nabla}\nabla^2 C)\right), 
\end{eqnarray*}
so taking the initial state expectation value of both sides and using $\bra{s}\nabla^\ell C\ket{s}=0$ gives
\begin{eqnarray*}
|\langle C\rangle_1 - \langle C\rangle_0 +\gamma\beta \langle \nabla_C \nabla C\rangle_0| &\leq& |\beta| \| \int_0^\gamma d\alpha \int_0^\alpha  d\tau e^{i\tau \nabla_C}\nabla_C^2 \nabla C\|+\frac{\beta^2}2\|\int_0^\gamma da e^{ia\nabla_C} \nabla_C (e^{i\tau \nabla}\nabla^2 C)\|\\
&\leq& |\beta| (\gamma^2/2) \|\nabla_C^2\nabla C\| + 
\frac{\beta^2\gamma}2 \|\nabla_C \nabla^2 C\|\\
&\leq& 
4|\beta|\gamma^2 k\|C\|_*^3 + 4\beta^2|\gamma|  k^2\|C\|_*^2, 
\end{eqnarray*}
where we have used \lemsref{lem:normCGrad}{lem:normGrad}. Observing the left-hand side is invariant under shifts $C\leftarrow C +aI$ implies the tighter bound \eqref{eq:errorBoundC}. 
\end{proof}
Naively apply the same approach to the error in the probability estimate of \thmref{thm1:smallAngles} leads to a bound proportional to $\|H_x\|=\|\ket{x}\bra{x}\|=1$. Here we apply a refined analysis to obtain a %
more useful bound %
that reflects the exponentially small initial probabilities in the computational basis, at the expense of an exponential factor in the bound.  

\begin{lem} \label{lem:thm1errorboundp}
For QAOA$_1$ %
the error in the estimate (\ref{eq:px1}) of \thmref{thm1:smallAngles} satisfies
\begin{equation} \label{eq:probBound}
  \left|  P_1(x) - (P_0(x)- \frac{2\gamma \beta}{2^n}dc(x))\right|  %
 \leq \frac{2}{2^n}\bigg( n^2\beta^2 e^{2n|\beta|} + \frac43n|\beta||\gamma|^3\|C\|_*^3  %
 \cosh(2|\gamma|\|C\|_*)\bigg).
  \end{equation}
\end{lem}
\begin{proof}

Turning to the probability, for $H_x=\ket{x}\bra{x}$ we have
\begin{eqnarray*} 
    e^{i\gamma \nabla_C} e^{i\beta \nabla}H_x
    &=& H_x + i\beta e^{i\gamma \nabla_C}\nabla H_x + e^{i\gamma \nabla_C}\underbrace{(-\frac{\beta^2}2\nabla^2 H_x -i\frac{\beta^2}{3!}\nabla^3 H_x + \dots)}_{R_1(\beta)}\\
    &=&  H_x + i\beta\nabla H_x -\gamma\beta\nabla_C\nabla C +i\beta \underbrace{(-\frac{\gamma^2}2\nabla^2_C\nabla H_x+\dots)}_{R_2(\gamma)} + e^{i\gamma \nabla_C}R_1(\beta)
\end{eqnarray*}
and so taking initial state expectations and using $\langle \nabla^\ell H_x \rangle_0=0$ from \lemref{lem:expecGrad} gives 
$$  %
P_1(x)-P_0(x)
+\langle \nabla_c \nabla H_x \rangle_0=%
\bra{\gamma} R_1(\beta) \ket{\gamma}
+i\beta \langle R_2(\gamma)\rangle_0$$
for $\ket{\gamma}:=U_P(\gamma)\ket{s}$.
For the first term on the right we have 
\begin{eqnarray*}
|\bra{\gamma}R_1(\beta)\ket{\gamma}|&=& %
\left|\frac1{2^n}\sum_{yz}e^{-i\gamma(c(z)-c(y))} \bra{y}R_1(\beta)\ket{z}\right| %
\,\leq\, \frac1{2^n}\sum_{yz}| \bra{y}R_1(\beta)\ket{z}|\\ 
&\leq& \frac1{2^n}\sum_{\ell=2}^\infty  \frac{|\beta|^\ell}{\ell!} \sum_{yz} | \bra{y}\nabla^\ell H_x \ket{z}| \,\leq\, \frac1{2^n}\sum_{\ell=2}^\infty \frac{|\beta|^\ell}{\ell!}(2n)^\ell \\
&\leq&  \frac2{2^n}n^2\beta^2 e^{2n|\beta| }
\end{eqnarray*}
where we have used $\sum_{y,z}|\bra{y}\nabla^\ell H_x \ket{z}|\leq (2n)^\ell$, which follows 
observing that $H_x$ has a single nonzero matrix element $\bra{x}H_x\ket{x}=1$ in the computational basis, and each application of $\nabla$ increases the number of nonzero matrix elements (equal to $\pm 1$) by a factor at most $2n$ (e.g., we have $\nabla H_x = \sum_{j=1}^n (\ket{x^{(j)}}\bra{x}-\ket{x}\bra{x^{(j)}})$.

Next, for the term corresponding to $R_2(\gamma)=-\frac{\gamma^2}{2}\nabla^2_C \nabla H_x -i\frac{\gamma^3}{3!}\nabla^3_C \nabla H_x+\dots$, suppose we have shifted $C$ such that $\|C\|=\|C\|_*$, which leaves $R_2(\gamma)$ and other commutators invariant. 
Consider $G:=\nabla H_x$ which %
similarly satisfies $\sum_{y,z}|\bra{y}G\ket{z}|\leq 2n$. 
From Lem.~\ref{lem:genGrads}, when $\ell$ is even $\langle \nabla_C^\ell G \rangle_0=0$, and for odd $\ell$ from expanding the commutators we have 
\begin{eqnarray*}
|\langle \nabla_C^\ell G \rangle_0| &=& \frac1{2^n} |\sum_{y,z}\bra{y}\nabla_C^\ell G\ket{z}|
=  \frac1{2^n} |\sum_{y,z}\bra{y}C^\ell G - \ell C^{\ell-1}GC+\dots \pm G C^\ell\ket{z}|\\
&\leq& \frac{2^\ell}{2^n} \|C\|_*^\ell \sum_{y,z} |\bra{y} G \ket{z}| \leq  \frac{2^{\ell+1}}{2^n} \|C\|_*^\ell n,
\end{eqnarray*}
which implies
\begin{eqnarray*}
|\langle R_2(\gamma) \rangle_0 | &=& | -i\frac{\gamma^3}{3!}\langle \nabla_C^3 \nabla H_x\rangle_0+
i\frac{\gamma^5}{5!}\langle \nabla_C^5 \nabla H_x\rangle_0+\dots|\\ 
&\leq& \frac{|\gamma|^3}{3!}  \frac{2^3}{2^n}2n\|C\|_*^3
 +\frac{|\gamma|^5}{5!}\frac{2^5}{2^n}2n\|C\|_*^5+\dots\\
 &\leq& \frac83\frac{n}{2^n}|\gamma|^3\|C\|_*^3 \cosh(2\gamma \|C\|_*) \,\,\leq\,\,  \frac{3n}{2^n}|\gamma|^3\|C\|_*^3 e^{2 |\gamma| \|C\|_*}.
\end{eqnarray*}
Hence %
using $|P_1(x)-P_0(x)
+\langle \nabla_c \nabla H_x \rangle_0|\leq
|\bra{\gamma} R_1 \ket{\gamma}|
+|\beta \langle R_2\rangle_0|$ gives (\ref{eq:probBound}).
\end{proof}

Finally, we use the lemmas to bound the error of the estimates of \thmref{thm1:smallAngles}. 

\begin{proof}[Proof of \corref{cor:thm1errorboundeps}]
From \lemref{lem:thm1errorbound}
for QAOA$_1$ with $k$-local %
$C$ 
we have
\begin{equation} \label{eq:errorBoundC2}
\left|
\langle C\rangle_1 - \,\left( \langle C \rangle_0 - \gamma \beta \langle \nabla_C \nabla C \rangle_0 \right) \right|
\,\leq\, \|C\|_* (
4|\beta|\gamma^2 k\|C\|_*^2 + 4\beta^2|\gamma|  k^2\|C\|_*).
\end{equation}
with $\langle \nabla_C \nabla C\rangle_0 = 2\langle C\, DC \rangle_0$. 
Selecting $\gamma,\beta$ such that $|\beta| \leq \frac{\sqrt{\e}}{2k}$ and $|\gamma| \leq \frac{\e^{1/4}}{2\|C\|_*} $ gives
$$4|\beta|\gamma^2 k\|C\|_*^2 + 4\beta^2|\gamma|  k^2\|C\|_*\leq \frac{\e}2 + \frac{\e^{5/4}}2 \leq \e$$
which then implies \eqref{eq:smallAnglesErrorBound} observing that $\|C_Z\| \geq \|C\|_*$. %

Similarly, the probability bound follows from \lemref{lem:thm1errorboundp} applying the smaller mixing angle $|\beta|\leq \frac25 \frac{\sqrt{\epsilon}}{n}$ to \eqref{eq:probBound} which %
gives an upper bound of $0.92\e/2^n$ and implies \eqref{eq:smallAnglesErrorBoundProb}. 
\end{proof}

\section{%
Quadratic Hamiltonians}
\label{app:klocal}
For $k$-local %
Hamiltonians, higher-order gradients have particularly nice properties.\footnote{Recall an operator on qubits is (strictly) $k$-local if all terms in its Pauli operator expansion act nontrivially on at most $k$ qubits, and strictly $k$-local if all non-Identity terms act on exactly $k$ qubits.} %
Here we prove \lemref{lem:quboGrads} for QUBOs by considering the linear and quadratic terms in turn.

\begin{lem}
For a linear ($1$-local) cost Hamiltonian $C= \sum_{j}c_{j}Z_j$, we have
\begin{eqnarray}
    \nabla C = -2i C_Y \;\;\;
    \text{ and } \;\;\; 
    \nabla^2 C = 4 C,
\end{eqnarray}
with $C_Y:=\sum_{j}c_{j}Z_j$, 
which for $\ell\in\naturals$ implies 
\begin{equation}
  \nabla^{2\ell} C = 4^\ell C\;\;\; \text{ and } \;\;\;\nabla^{2\ell+1} C = 4^\ell \nabla C.
\end{equation}
\end{lem}
\begin{proof}
The lemma follows trivially from the Pauli matrix commutation relations as $\nabla_j C = [X_j,C]=2ic_jY_j$ and $\nabla^2_j C = 4c_jZ_j$. The second statement follows by induction.
\end{proof}
\begin{lem}  \label{lem:quadraticGrads}
For a strictly quadratic ($2$-local) cost Hamiltonian $C=\sum_{uv}c_{uv}Z_uZ_v$, %
\begin{eqnarray}
    \nabla^2 C = 8C - 8C_Y \;\;\;
    \text{ and } \;\;\; 
    \nabla^3 C = 16 \nabla C,
\end{eqnarray}
with $C_Y:=\sum_{uv}c_{uv}Y_uY_v$, 
which for $\ell\in\naturals$ implies 
\begin{equation}
  \nabla^{2\ell} C = 16^{\ell-1} \nabla^2 C\;\;\; \text{ and } \;\;\;\nabla^{2\ell+1} C = 16^\ell \nabla C.
\end{equation}
\end{lem}
\begin{proof}
The first statement follws trivially as before using the cyclic Pauli relations $[X_j,Y_j]=2iZ_j$ and the linearity of $\nabla=[B,\cdot]$, and the second by induction. 
\end{proof}

\lemref{lem:quboGrads} %
then follows immediately from combining the linear and quadratic cases. %
Analaoous results follow for $k$-local Hamiltonians with $k>2$.

\subsection{QAOA$_1$ for quadratic cost Hamiltonians}
\label{sec:QUBOs}

\lemref{lem:quboGrads} %
may be applied directly to QAOA for QUBOs such as MaxCut where $ C=\frac{m}{2}-\frac12\sum_{(ij)\in E}Z_iZ_j$. Recall the quantities of \tabref{tab:summary}. 

\begin{proof}[Proof of \eqref{eq:expecC1maxcut}]
For MaxCut, applying Lem.~\ref{lem:quadraticGrads} to (\ref{eq:infinitesimalConj}) 
we have %
\begin{eqnarray} \label{eq:UMCQUBO}
    U_M^\dagger C U_M
    &=& %
    \sum_{k=0}^\infty \tfrac{(i\beta)^k}{k!}\nabla^k C \nonumber\,
    =\, C + \sum_{k=0}^\infty \tfrac{(i\beta)^{2k+1}}{(2k+1)!}16^k\nabla C + \sum_{k=1}^\infty \tfrac{(i\beta)^{2k}}{(2k)!}16^{k-1}\nabla^2 C\nonumber\\
    &=& C + \tfrac{i}{4}\nabla C \sum_{k=0}^\infty \tfrac{(-1)^k(4\beta)^{2k+1}}{(2k+1)!}
     +\tfrac{1}{16} \nabla^2 C \sum_{k=1}^\infty \tfrac{(-1)^k(4\beta)^{2k}}{(2k)!}\nonumber\\
     &=& C + \tfrac{\sin(4\beta)}{4}\,i\nabla C 
     +\tfrac{(\cos(4\beta)-1)}{16} \, \nabla^2 C. 
\end{eqnarray}
Hence, using $\cos (2x)=1-2\sin^2(x)$, taking the expectation value with respect to
$\ket{\gamma}:=U_P(\gamma)\ket{s}$ gives \eqref{eq:expecC1maxcut} via 
$\langle C\rangle_1=\bra{\gamma}U_M^\dagger C U_M\ket{\gamma}$ and $\bra{\gamma} C \ket{\gamma}=\bra{s}C\ket{s}=\langle C\rangle_0$. 
\end{proof}

The general QUBO formula~\eqrefp{eq:QAOA1QUBO} for $\langle C\rangle_1$ is
obtained similarly. %
Here we consider the explicit example of Balanced-Max-$2$-SAT as defined in \secref{sec:balancedMax2Sat}, which generalizes (and rederives along the way) previous results for QAOA$_1$ for MaxCut. We emphasize the calculations such as that of the following proof for more general problems may %
be most easily implemented %
numerically, in particular for $p>1$ (cf. Algorithm 2 in \secref{sec:classicalAlgGeneral2}).

\begin{proof}
[Proof of \eqref{eq:expecC1balmax2sat} (and \eqsref{eq:maxcutnablaCexpec}{eq:maxcutnabla2Cexpec}] 
As $C$ is strictly $2$-local (modulo the identity term), %
\eqref{eq:QAOA1QUBO} (cf. \eqref{eq:UMCQUBO}) gives $\langle C\rangle_1 =\bra{\gamma}U_M^\dagger C U_M\ket{\gamma}$ with $\ket{\gamma}=U_P\ket{s}$ and 
$$ U_M^\dagger C U_M = C +\frac{\sin4\beta}4 i\nabla C -\frac{\sin^22\beta}8 \nabla^2 C.$$
We follow but generalize the proof for MaxCut in \cite{hadfield2018thesis,wang2018quantum} (which is a simpler case with all signs the same of the $ZZ$ terms in the cost Hamiltonian).  
Recall that in the Pauli basis only strictly $I,X$ terms can contribute to QAOA initial state expectation values 
(i.e., if a Pauli string $A_j$ contains a $Y$ or $Z$ factor then $\bra{s}A_j\ket{s}=Tr(\ket{s}\bra{s}A_j)=0$).
Observing  that $i\nabla C=  -\frac12 \sum_{(ij)\in E}(-1)^{i\oplus j} (Y_iZ_j+Z_iY_j)$ (cf. \tabref{tab:summary}) we have
\begin{eqnarray}
U_P^\dagger i\nabla C U_P&=&-\frac12 \sum_{(ij)\in E}(-1)^{i\oplus j} U_P^\dagger(Y_iZ_j+Z_iY_j)U_P\nonumber\\
&=&-\frac12 \sum_{(ij)\in E}(-1)^{i\oplus j} (e^{2i\gamma C^{\{i\}}}Y_iZ_j+e^{2i\gamma C^{\{j\}}}Z_iY_j),
\end{eqnarray}
where as the terms in $C$ mutually commute we may write each exponential as a Pauli sum using $e^{\pm i\gamma Z_iZ_j/2}=I\cos(\gamma/2)\pm i\sin(\gamma/2)Z_iZ_j$ and expanding the product. Clearly, the only resulting term that can contribute is proportional to  
$(-1)^{i\oplus j}Z_iZ_j*(-1)^{i\oplus j}Y_iZ_j=-iX_i$ (i.e., all other terms will contain at least one $Z$ factor).
Hence, as each $C^{\{i\}}$ contains $d_i+1$ terms we have
\begin{eqnarray}
\bra{\gamma} i\nabla C \ket{\gamma} %
&=&\frac12 \sum_{(ij)\in E}\sin(\gamma/2 )(\cos^{d_i}(\gamma/2)+\cos^{d_j}(\gamma/2)).  
\end{eqnarray}
(A nearly identical argument gives \eqref{eq:maxcutnablaCexpec} for MaxCut after summing over the edges.)

Similarly, using $\nabla^2C=8C-8C_Y=-\frac14 \sum_{(ij)\in E}(-1)^{i\oplus j} (Z_iZ_j-Y_iY_j)$
gives 
\begin{eqnarray}
U_P^\dagger \nabla^2 C U_P
&=& 8C_Z+2 \sum_{(ij)\in E}(-1)^{i\oplus j} e^{2i\gamma (C^{\{i\}\setminus j}+C^{\{j\}\setminus i})}Y_iY_j
\end{eqnarray}
and so as $\langle C_Z \rangle_0 =0$ we have 
\begin{eqnarray} \label{eq:expoeq110}
\bra{\gamma} \nabla^2 C \ket{\gamma} = 2 \sum_{(ij)\in E}(-1)^{i\oplus j}\langle  e^{2i\gamma (C^{\{i\}\setminus j}+C^{\{j\}\setminus i})}Y_iY_j \rangle_0.
\end{eqnarray}
Consider a fixed edge $(ij)$ with $d:=deg(i)-1$ and $e:=deg(j)-1$ recall we define $f^+=f^+_{ij},f^-=f^-_{ij}$ as the positive and negative triangles containing $(ij)$ the the instance graph $E$. 
For convenience here we let $c=\cos(\gamma/2)$ and $s=\sin(\gamma/2)$. 
Writing the 
exponential in \eqref{eq:expoeq110} 
as a product of factors $e^{\pm i\gamma Z_uZ_v/2}=Ic\pm isZZ$ %
and expanding gives a linear combination of $2^{d+e}$ terms, which we view as a series in $s$.  
The lowest order (w.r.t $s$) term that can contribute \cite{hadfield2018thesis} corresponds to selecting a single triangle from among the factors and is
$$\binom{f^+}{1} c^{d+e-2}s^2-\binom{f^-}{1}c^{d+e-2}s^2=(f^+-f^-)c^{d+e-2}s^2,$$
i.e. negative triangles subtract, and there are $f^+,f^-$ ways to select a triangle of each type. 
(cf. the corresponding proof for MaxCut where all triangles are positive~\cite{hadfield2018thesis,wang2018quantum}.)
As $f^+,f^-$ become larger, 
pairs of additional triangles can combine to $\pm I$ such that the next contributing terms are proportional to $c^{d+e-6}s^6$ and involve $3$ triangles, %
with each sign determined by the number of positive- versus negative-parity triangles involved to give proportionality factor
\begin{eqnarray*}
\binom{f^+}{3}\binom{f^-}{0}-\binom{f^+}{2}\binom{f^-}{1}+\dots-\binom{f^+}{0}\binom{f^-}{3} &=&  \sum_{k=0}^3  \binom{f^+}{3-k}\binom{f^-}{k}(-1)^k
\end{eqnarray*}
(i.e. three negative-parity triangles give $-1$, and so on.)

Generalizing this argument to higher order, factors with $\ell=1,3,5$ triangles combine to contribute $c^{d+e-2\ell}s^{2\ell}$ times the factor\footnote{Without the factor $(-1)^k$ we would have  Vandermonde's identity $\sum_{k=0}^\ell\binom{f^+}{\ell-k}\binom{f^-}{k}=\binom{f}{\ell}$.}  
\begin{eqnarray} \label{eq:hypgepfactor}
 \sum_{k=0}^\ell  \binom{f^+}{\ell-k}\binom{f^-}{k}(-1)^k =: \binom{f^+}{\ell} \prescript{}{2}{\mathbf{F}}_1(-f^-,-\ell;f^+-\ell+1;-1)
\end{eqnarray}
where $\prescript{}{2}{\mathbf{F}}_1$ is the  Gaussian (ordinary) hypergeometric function~\cite{slater1966generalized}. Note %
\eqref{eq:hypgepfactor} 
evaluates to $f^+-f^-$ for $\ell=1$, and to $\binom{f^+}{3}-\binom{f^-}3-\tfrac12 f^+f^-(f^+-f^-)$ for $\ell=3$. 
Hence we have 
\begin{eqnarray*}
\langle  e^{2i\gamma (C^{\{i\}\setminus j}+C^{\{j\}\setminus i})}c_{ij}Y_iY_j \rangle_0 %
 &=&
c^{d+e-2f}\sum_{\ell=1,3,5,\dots}^{f^++f^-} c^{2(f-\ell)}s^{2\ell} \binom{f^+}{\ell} \prescript{}{2}{\mathbf{F}}_1(-f^-,-\ell;f^+-\ell+1;-1)\\
&=:& c^{2(f-\ell)}s^{2\ell}\, g(f^+,f^-;\gamma,\beta)
\end{eqnarray*}
where we have defined
$g(f^+,f^-;\gamma,\beta)=\sum_{\ell=1,3,5,\dots}^f c^{2(f-\ell)}s^{2\ell} \binom{f^+}{\ell} \prescript{}{2}{\mathbf{F}}_1(-f^-,-\ell;f^+-\ell+1;-1)$. 
It easy to see this term is $0$ for all edges with $f^+=f^-$ or $f=0$, i.e., $g(a,a)=0$ for $a=0,1,2,\dots$. Moreover, the case of strictly positive or strictly negative triangles $g(f,0)=-g(0,f)$ reproduces the analysis of MaxCut (up to constants) and is summable to give $$g(f,0)=\tfrac12(1-(c^2-s^2)^f)=\tfrac12(1-\cos^f\gamma)=-g(0,f),$$
which implies \eqref{eq:maxcutnabla2Cexpec} for MaxCut.

Hence, we have $\langle  e^{2i\gamma (C^{\{i\}\setminus j}+C^{\{j\}\setminus i})}c_{ij}Y_iY_j \rangle_0 = c^{d+e-2f^+-2f^-}g(f^+,f^-)$ and so 
$\bra{\gamma} \nabla^2C\ket{\gamma}=2 \sum_{ij}c^{d+e-2f^+-2f^-}g(f^+,f^-)$, which together with the result for $\bra{\gamma} \nabla^2C\ket{\gamma}$ gives \eqref{eq:expecC1balmax2sat}. 
\end{proof}
\end{appendix}

\end{document}